\documentclass[%
 reprint,
 table,
superscriptaddress,
nofootinbib,
 amsmath,
 amssymb,
 aps,
 prapplied,
floatfix,
]{revtex4-2}

\usepackage{lmodern}
\usepackage{graphicx}
\usepackage{caption}
\usepackage{subcaption}
\usepackage{ragged2e}
\usepackage{lipsum}
\usepackage{hhline}
\usepackage{booktabs}
\usepackage{multirow}
\makeatletter
\newcommand*\bigcdot{\mathpalette\bigcdot@{.5}}
\newcommand*\bigcdot@[2]{\mathbin{\vcenter{\hbox{\scalebox{#2}{$\m@th#1\bullet$}}}}}
\makeatother

\usepackage{dcolumn}
\usepackage{bm}

\usepackage[utf8]{inputenc}
\usepackage[english]{babel}
\usepackage{csquotes}
\usepackage[T1]{fontenc}
\usepackage{amssymb}
\usepackage{mathtools}
\usepackage{MnSymbol}
\usepackage{bbold}

\usepackage{tikz}
\definecolor{var}{HTML}{42743C}
\definecolor{amp}{HTML}{793CA3}
\usetikzlibrary {shapes.geometric}
\usetikzlibrary{arrows.meta,calc}

\definecolor{fnavy}{HTML}{103B60}
\definecolor{fenc}{HTML}{078A9C}
\definecolor{fvar}{HTML}{42743C}
\definecolor{famp}{HTML}{793CA3}

\tikzset{
  fwire/.style={
    draw=fnavy!80,line width=.45pt
  },
  fgate/.style={
    draw=fvar,fill=fvar!9,
    rounded corners=1.2pt,line width=.5pt,
    inner sep=1pt,minimum height=6mm
  },
  fencoding/.style={
    fgate,draw=fenc,fill=fenc!10
  },
  fbox/.style={
    rounded corners=2pt,line width=.5pt
  },
  fflow/.style={
    draw=fnavy,line width=.65pt,
    -{Latex[length=1.5mm]}
  }
}

\newcommand{\Fmeter}[2]{%
  \begin{scope}[shift={(#1,#2)}]
    \draw[fwire,fill=white,rounded corners=.8pt]
      (-2.6,-2.7) rectangle (2.6,2.7);
    \draw[fwire]
      (-1.9,-.6) arc[start angle=180,end angle=0,radius=1.9];
    \draw[fwire,-{Latex[length=.8mm,width=.5mm]}]
      (0,-.8)--(1.3,1.6);
    \fill[fnavy] (0,-.8) circle (.2);
  \end{scope}%
}

\newcommand{\Fbundle}[3]{%
  \foreach \dy in {-.55,0,.55}{
    \draw[fwire] (#1,{#3+\dy})--(#2,{#3+\dy});
  }%
}

\newcommand{\Fslash}[2]{%
  \draw[fwire] ({#1-1},{#2-1.4})--({#1+1},{#2+1.4});
  \node[font=\scriptsize,anchor=south]
    at (#1,{#2+1.2}) {$n$};
}

\usepackage{amsthm}
\usepackage{amsmath}
\usepackage{mathrsfs}
\DeclareFontFamily{U}{rsfs}{\skewchar\font127 }
\DeclareFontShape{U}{rsfs}{m}{n}{
  <-6> rsfs5
  <6-8> rsfs7
  <8-> rsfs10
}{}
\usepackage{dsfont}
\usepackage{braket}

\newcommand{\cross}{\times}

\usepackage[thinc]{esdiff}
\usepackage{nicefrac}
\usepackage{xfrac}
\usepackage{algorithm}
\usepackage{algpseudocode}

\usepackage{hyperref}
\usepackage[capitalise]{cleveref}
\creflabelformat{equation}{#2\textup{#1}#3}

\newtheorem{theorem}{Theorem}
\newtheorem{corollary}{Corollary}
\newenvironment{DIFnomarkup}{}{}

\begin{document}
\preprint{APS/123-QED}

\title{Fourier analysis of quantum neural network with non-linear data embedding}

\author{Haiyue Kang}
 \email{haiyuek@student.unimelb.edu.au}
 \affiliation{School of Physics, University of Melbourne, Parkville, VIC 3010, Australia}


\author{Martin Sevior}
 \affiliation{School of Physics, University of Melbourne, Parkville, VIC 3010, Australia}

\author{Muhammad Usman}
 \affiliation{School of Physics, University of Melbourne, Parkville, VIC 3010, Australia}
 \affiliation{Faculty of Information Technology, Monash University, Clayton, VIC 3168, Australia}

\begin{abstract}
Fourier analysis has become a crucial tool for understanding the expressivity of Variational Quantum Circuit (VQC) models, as well as an important indicator of barren plateaus (BP). While existing literature has only studied angle-embedded VQCs in a noiseless environment, here we develop the Fourier analysis of VQCs with non-linear data embedding, with particular focus on amplitude embedding, which provides a naturally compact encoding scheme. We first investigate a subtle difference in the domain of input features within amplitude embedding that leads to a distinct expressivity of the zero-frequency Fourier coefficient. Assuming that the ensemble of unitaries generated from the parameter space forms at least a unitary 2-design, we derive that, in integer-valued frequencies, the mean of the Fourier coefficients is concentrated at zero, and the variance is upper bounded by an inverse power-law function with respect to the infinity norm of the multi-dimensional frequency. A numerical evidence of exponentially decaying pattern with respect to the frequency magnitude is also established, providing a frequency-perspective of viewing VQC model expressivity and its trainability. When a noise channel with unitary Kraus operators and probabilities $\{p_k\}$ is taken into account, the variance is further suppressed by a factor $\left(\sum_kp_k^2\right)^{Q}<1$, where $Q$ is the number of channel instances applied. Furthermore, we demonstrate and validate the analytical results through simulations, both noiseless and noisy, and contrast its expressivity against the angle-embedded model Fourier features. We also examine the model's performance where target functions are decomposed into non-integer frequencies, highlighting the practical utility of the approach. Our results establish a rigorous Fourier framework for amplitude-encoded VQCs, offering both theoretical guarantees on expressivity, hence trainability scaling in the frequency domain, as well as practical simulations for deployment on noisy quantum devices.

\end{abstract}

\maketitle
\section{Introduction}\label{sec:introduction}

Quantum Machine Learning (QML) \cite{QML,QML2,VQA,VQE,QAOA,QPCA,QNN} has attracted considerable attention over the past years due to the anticipated computational advantage over the classical counterparts, particularly in problems that are inherently quantum \cite{VQE,Accelerated_VQE,VQE_variants_summary,VQSE,QAOA,quantum_chemistry}, or optimizations that classically scale poorly \cite{quantum_finance}. Recent applications also include quantum generative adversarial networks for image generation \cite{Dual_Discriminator_QGAN}, further exploring the expressive capabilities of quantum learning models. However, it is becoming increasingly hard to balance the issue of trainability, particularly barren plateaus (BP) \cite{barren_plateaus,barren_plateaus_summary,noise_induced_barren_plateaus}, and classical non-simulability of QML models \cite{QCNN_classically_simulable, classical_simulation}. Hence, a great amount of effort has been invested to tackle this problem \cite{local_cost_function_1, local_cost_function_2, pre_training_avoids_barren_plateaus, layer_by_layer_1, layer_by_layer_2,thermal_helps_learning,thermal_helps_learning2,reflection_invariant_qml,rotational_invariant_qml,non_unitary_qml,LCE_method}. Meanwhile, Fourier/Spectral analysis of Variational Quantum Circuit (VQC) models \cite{Fourier_analysis_Schuld_Maria} opens yet another avenue to understand or explore the expressivity of a QML model quantitatively, which could lead to the design of models with improved trainability. This is also an important indicator of BP that decays exponentially with respect to the number of qubits \cite{Fourier_analysis_BP}. In \cite{Fourier_analysis_BP}, the Fourier coefficient expressivity can be directly related with trainability through the spectral power sum of the circuit output in parameter space. For a VQC model with output decomposed into Fourier series
$f(\bm{x})=\sum_{\bm{k}}a_{\bm{k}}
e^{-i\bm{k}\cdot\bm{x}}$, Parseval's identity gives
\begin{equation}
    \sum_{\bm{k}}|a_{\bm{k}}|^2
    =\mathbb{E}_{\bm{\theta}}\!\left[|f(\bm{\theta})|^2\right].
\end{equation}
For a variational circuit whose parameters are encoded through single Pauli rotations and whose observable is a traceless Pauli operator, the parameter-shift rule further implies
\begin{equation}
    \operatorname{Var}_{\bm{\theta}}
    \!\left[\partial_{\theta_j}f\right]
    \leq \sum_{\bm{k}}|a_{\bm{k}}|^2
    =O(2^{-n}).
\end{equation}
This establishes a spectral connection to trainability:
exponentially suppressed spectral power constrains the
gradient variance, providing a Fourier-space characterization
of barren plateaus. Furthermore, through Fourier analysis of angle-embedded inputs, it was found that the possible frequency spectrum of the output with respect to the input variable constitutes the sums of possible eigenvalues from the encoding gate Hamiltonian \cite{Fourier_analysis_Schuld_Maria}, hence directly induce the corresponding Fourier coefficients \cite{Fourier_analysis_binary_tree}. In Ref. \cite{Fourier_analysis_expressivity_as_redundancies}, it is demonstrated that the expressivity of each Fourier coefficient is directly proportional to its redundancy, which is the number of combinations of eigenvalues that can add up to its corresponding frequency. In Ref. \cite{Fourier_analysis_expressivity_as_FIM}, a quantitative relationship between the scale of a model and whether it under/overfits the data was derived through the Fourier series decomposition and quantum Fisher information matrix (QFIM). Moreover, in Ref. \cite{Fourier_analysis_correlation_metric}, a metric was invented that better captures the expressivity of a model by computing the correlations between the Fourier coefficients across different frequencies. 

Despite these insightful analyses, to the best of our knowledge, all literature about Fourier analysis assumes the model uses angle embedding and the data re-uploading protocol of encoding to extend the available spectrum. Indeed, angle embedding has a unique advantage in the circuit depth, which is at most linear with respect to the dimensions of the features to encode \cite{quantum_embedding_summary_Schuld}, yet the number of qubits also scales linearly with the dimensions. This exponential gap in the scaling of the qubits becomes a serious problem when the number of features grows extremely large, which is actually very common, such as images, text, audio, video, and multi-mode features. In particular, for large language models (LLMs) nowadays, the number of tokens that are handled within their context window can easily go beyond $10^6$. At the same time, each of the tokens is also a high-dimensional feature \cite{google_gemma_3_tech_report,deepseekv4}. On the other hand, the number of qubits required only scales in $O(\log(N))$ for $N$ features in amplitude embedding. As a result, other non-linear embedding protocols, particularly the amplitude-embedding, are still indispensable for VQC algorithms, and it is now urgent to explore their characteristics on a VQC via Fourier analysis. 

In addition, only very limited investigation of the Fourier analysis on VQC models has taken into account the presence of noise channels \cite{fourier_coeff_theta_suppresed}. Modeling and analysis of a noisy quantum process is indeed more difficult; however, this is inevitable when algorithms are implemented on real quantum computers. Even with potentially quantum error-corrected (QEC) codes in the future, small logical errors will always be present \cite{surface_codes,qec_lattice_surgery}. Yet, knowing how much error can be tolerated is also an important step to reduce the huge overhead of QEC \cite{haiyue_partial_QEC}. Therefore, it is essential to demonstrate and contrast the behaviours of the Fourier coefficients with or without noise in simulations.

As proposed above, our work develops and applies Fourier analysis to VQCs with amplitude feature mapping, and examines how the corresponding Fourier coefficients scale with frequency and the fitted functions. We stress that, although this paper does highlight some advantages of amplitude embedding over angle embedding, the main focus is to give a quantitative analysis of expressivity and trainability with breakdown details in the frequency domain. The structure of this paper is described as follows: In \Cref{sec: Preliminaries}, we introduce the background theories about general variational quantum circuit (VQC) models in \Cref{subsec: VQC model} and some preliminary theories of Fourier analysis in \Cref{subsec: Fourier analysis}, both for the conventional angle embedding and the amplitude embedding of our interest. In \Cref{subsec: vanishing zero frequency}, we demonstrate how a subtle difference in the domain of encoding values within amplitude embedding results in a vanishing coefficient at zero frequency. We also discuss the mean, variance, and therefore expressivity of the Fourier coefficients as a function of their frequencies, observables, and feature dimensions, which are derived analytically, and simulated in \Cref{subsec: moments of fourier coeffs,subsec: simulations} on a small scale, respectively. In the last subsection, \Cref{subsec: non periodic functions}, we discuss the expressivity of the amplitude embedding but at non-integer or non-discretised frequencies.


\begin{figure*}[t!]
\centering

\begin{subfigure}{\textwidth}
\centering
\resizebox{\linewidth}{!}{%
\begin{tikzpicture}[
  x=1mm,y=1mm,
  font=\sffamily\fontsize{8}{9.5}\selectfont,
  text=fnavy
]
\path[use as bounding box] (0,-2) rectangle (180,52);

\draw[fbox,draw=fenc,fill=fenc!8]
  (1,8) rectangle (20,40);
\node at (10.5,20) {$\{(\mathbf{x},\mathbf{y})_i\}$};

\begin{scope}[shift={(10.5,31)}]
  \foreach \y in {-3,0,3}{
    \path[draw=fenc,fill=white,line width=.55pt]
      (-4.7,\y)
      --(-4.7,{\y-2.2})
      arc[start angle=180,end angle=360,
          x radius=4.7,y radius=1.3]
      --(4.7,\y);

    \draw[draw=fenc,fill=fenc!9,line width=.55pt]
      (0,\y) ellipse (4.7 and 1.3);

    \fill[fvar!85!black]
      (-2.8,{\y-1.85}) circle (.25);

    \draw[draw=fenc!70,line width=.4pt,line cap=round]
      (-1.6,{\y-1.85})--(1.8,{\y-1.85});
  }
\end{scope}

\node[anchor=north] at (10.5,6) {training data};

\draw[fflow]
  (10.5,40)--(10.5,49)--(157,49)--(157,41);
\node[fill=white,inner sep=1pt] at (84,49) {$\{y_i\}$};

\draw[fflow] (20,25)--(43,25);
\node[align=center] at (31.5,31) {input encoding};

\draw[fbox,draw=fvar,fill=fvar!3]
  (43,8) rectangle (94,41);
\node at (67,37.5) {QNN};
\node at (67,32.8) {$U(\mathbf{x},\theta)$};

\begin{scope}[shift={(86,36)}]
  \foreach \angle in {0,60,120}{
    \draw[draw=fvar,line width=.55pt,rotate=\angle]
      (0,0) ellipse (4.2 and 1.5);
  }
  \fill[famp] (0,0) circle (.65);
  \fill[fvar] (4.2,0) circle (.4);
\end{scope}

\foreach \y in {13,19,25}{
  \draw[fwire] (45,\y)--(89,\y);
  \Fmeter{90}{\y}
}

\foreach \x/\y in {49/22,59/16,69/22,79/16}{
  \node[
    fgate,minimum width=5mm,minimum height=9mm
  ] at (\x,\y) {};
}
\node[fill=white,inner sep=0pt] at (73,16) {$\cdots$};

\draw[fflow] (94,25)--(139,25);
\node at (116.5,29.5) {inferred value};
\node[font=\fontsize{7.4}{9}\selectfont] at (116.5,19)
  {$f_\theta(\mathbf{x})
    =\langle0|U^\dagger(\mathbf{x},\theta)
      OU(\mathbf{x},\theta)|0\rangle$};

\draw[fbox,draw=famp!75,fill=famp!5]
  (139,8) rectangle (179,41);
\node at (159,37.5) {Cost function};
\node at (159,31.5)
  {$\displaystyle\sum_i
    C(\mathbf{x}_i,f_\theta(\mathbf{x}_i),y_i)$};

\draw[fflow] (159,27)--(159,24);
\node[align=center] at (159,19.5)
  {gradient descent/\\back propagation};

\node[font=\fontsize{9.5}{11}\selectfont] at (159,11)
  {$\theta_{t+1}=\theta_t-J(g_t,g_{t-1},\cdots)$};

\begin{scope}[shift={(143.4,20)}]
  \draw[
    draw=famp,line width=.5pt,
    fill=white,rounded corners=.6pt
  ] (-2,-2) rectangle (2,2);

  \draw[
    draw=famp!75,line width=.4pt,
    fill=famp!12,rounded corners=.3pt
  ] (-1.1,-1.1) rectangle (1.1,1.1);

  \foreach \t in {-1.2,0,1.2}{
    \draw[draw=famp,line width=.5pt]
      (\t,-2.8)--(\t,-2)
      (\t,2)--(\t,2.8)
      (-2.8,\t)--(-2,\t)
      (2,\t)--(2.8,\t);
  }
\end{scope}

\draw[fflow] (159,8)--(159,0)--(70,0)--(70,8);
\node[fill=white,inner sep=1pt] at (114,0)
  {$\theta_{t+1}$};

\end{tikzpicture}%
}
\caption{General workflow of training a VQC}
\label{fig:VQC_general}
\vspace{8pt}
\end{subfigure}

\begin{subfigure}{0.49\textwidth}
\centering
\resizebox{\linewidth}{!}{%
\begin{tikzpicture}[
  x=1mm,y=1mm,
  font=\fontsize{8.5}{10}\selectfont,
  text=fnavy
]
\path[use as bounding box] (0,-2) rectangle (90,66);

\node[
  anchor=east,inner sep=0pt,
  font=\fontsize{11}{13}\selectfont
] at (28,56) {
  $|\psi\rangle=
  \displaystyle
  \mathop{\vcenter{\hbox{
    \scalebox{1.45}{$\displaystyle\bigotimes$}
  }}}\limits_{i=1}^{n}$
};

\draw[famp,line width=.45pt]
  (30,51)--(51,51)
  (30,51)--(30,61);

\fill[famp!85] (34,51) rectangle (38,58.5);
\fill[famp!85] (44,51) rectangle (48,55.5);

\node[anchor=south,inner sep=.3pt]
  at (36,58.7) {$\cos(x_i)$};
\node[anchor=south,inner sep=.3pt]
  at (46,55.7) {$\sin(x_i)$};

\node[anchor=north,inner sep=.5pt]
  at (36,50.8) {$|\uparrow\rangle$};
\node[anchor=north,inner sep=.5pt]
  at (46,50.8) {$|\downarrow\rangle$};

\draw[famp!70,dashed,line width=.4pt]
  (26.5,40.8)--(26.5,51);

\Fbundle{8}{84}{40}
\Fslash{10}{40}
\node[   anchor=east,   font=\fontsize{9}{11}\selectfont ] at (7,40) {$|0\rangle^{\otimes n}$};

\node[fencoding,minimum width=12mm] (s1) at (19,40)
  {$S^{(1)}(\mathbf{x})$};
\node[fgate,minimum width=18mm] (w1) at (37,40)
  {$W^{(1)}(\theta_1)$};
\node[fencoding,minimum width=12mm] (s2) at (55,40)
  {$S^{(2)}(\mathbf{x})$};
\node[fill=white,inner sep=.5pt] at (65,40) {$\cdots$};
\node[fgate,minimum width=15mm] at (75,40)
  {$W^{(L)}(\theta_L)$};
\Fmeter{87}{40}

\draw[fvar!70,dashed]
  (w1.south west)--(1,28)
  (w1.south east)--(50,28);
\draw[fenc!70,dashed]
  (s2.south west)--(55,28)
  (s2.south east)--(89,28);

\draw[fbox,draw=fvar!65,fill=fvar!2]
  (1,0) rectangle (50,28);
\node[
  fill=white,inner sep=1pt,text=fvar,
  font=\sffamily\fontsize{6.5}{7}\selectfont
] at (25.5,28) {variational layer $l$};

\foreach \y in {22,16,10,4}{
  \draw[fwire] (3,\y)--(48,\y);
}
\node[fgate,minimum width=20mm,minimum height=10mm]
  at (14,19) {$W^{(l)}_1(\theta_{l,1})$};
\node[fgate,minimum width=20mm,minimum height=10mm]
  at (14,7) {$W^{(l)}_2(\theta_{l,2})$};
\node[fgate,minimum width=20mm,minimum height=10mm]
  at (37,13) {$W^{(l)}_3(\theta_{l,3})$};
\node[fill=white,inner sep=.4pt] at (38,23) {$\cdots$};

\draw[fbox,draw=fenc!65,fill=fenc!2]
  (55,0) rectangle (89,28);
\node[
  fill=white,inner sep=1pt,text=fenc,
  font=\sffamily\fontsize{6.5}{7}\selectfont
] at (72,28) {encoding layer};

\foreach \y/\k in {22/1,14/2,4/d}{
  \draw[fwire] (57,\y)--(87,\y);
  \node[
    fencoding,minimum width=20mm,minimum height=5mm
  ] at (72,\y) {$e^{-ix_{\k}D}$};
}
\node[fill=white,inner sep=0pt] at (72,9) {$\vdots$};

\end{tikzpicture}%
}
\caption{Layout of $U(\bm{x},\bm{\theta})$ for angle-embedding}
\label{fig:angle_embedding_reupload_circ}
\end{subfigure}
%
\begin{subfigure}{0.49\textwidth}
\centering
\resizebox{\linewidth}{!}{%
\begin{tikzpicture}[
  x=1mm,y=1mm,
  font=\fontsize{8.5}{10}\selectfont,
  text=fnavy
]
\path[use as bounding box] (0,-2) rectangle (90,66);

\node[
  anchor=east,inner sep=0pt,
  font=\fontsize{11}{13}\selectfont
] at (27,55.5) {$|\psi\rangle=$};

\draw[famp,line width=.45pt]
  (29,51)--(62,51)
  (29,51)--(29,61);

\foreach \x/\h/\k in {
  32/6/0,
  40/9/1,
  48/3/2,
  58/5/{d-1}
}{
  \fill[famp!85]
    ({\x-1.5},51) rectangle ({\x+1.5},{51+\h});

  \node[
    anchor=south,inner sep=.4pt,
    font=\fontsize{9.5}{11}\selectfont
  ] at (\x,{51+\h}) {$x_{\k}$};
}

\foreach \x/\k in {32/0,40/1,48/2,59/{d-1}}{
  \node[anchor=north,inner sep=.5pt]
    at (\x,50.8) {$|\k\rangle$};
}

\node[inner sep=0pt] at (53.5,53) {$\cdots$};

\draw[famp!70,dashed,line width=.4pt]
  (33.5,40.8)--(33.5,47.5);

\Fbundle{13}{84}{40}
\Fslash{14}{40}
\node[
  anchor=east,
  font=\fontsize{9}{11}\selectfont
] at (12,40) {$|0\rangle^{\otimes n}$};
\node[fencoding,minimum width=17mm] at (25,40)
  {$\mathcal A(\mathbf{x})$};
\node[fgate,minimum width=35mm] (w) at (60,40)
  {$W(\theta)$};
\Fmeter{87}{40}

\draw[fvar!70,dashed]
  (w.south west)--(9,28)
  (w.south east)--(87,28);

\draw[fbox,draw=fvar!65,fill=fvar!2]
  (9,0) rectangle (87,28);
\node[
  fill=white,inner sep=1pt,text=fvar,
  font=\sffamily\fontsize{6.5}{7}\selectfont
] at (48,28) {variational layer $l$};

\foreach \y in {22,16,10,4}{
  \draw[fwire] (11,\y)--(85,\y);
}
\foreach \x/\y/\j in {
  22/19/1,22/7/2,46/13/3,70/19/4,70/7/5
}{
  \node[
    fgate,minimum width=20mm,minimum height=10mm
  ] at (\x,\y) {$W^{(l)}_{\j}(\theta_{l,\j})$};
}
\node[fill=white,inner sep=.4pt] at (47,23) {$\cdots$};

\end{tikzpicture}%
}
\caption{Layout of $U(\bm{x},\bm{\theta})$ for amplitude-embedding}
\label{fig:amplitude_embedding_circ}
\end{subfigure}

\caption{(a) The general workflow to train a quantum neural network with a variational quantum circuit. First, a training dataset $\{(\bm{x}_i,y_i)\}\subset \mathcal{X}\cross \mathcal{Y}$ is prepared, where each classical input features $\bm{x}_i$ are encoded onto the quantum state, together with the trainable variational circuits, denoted as $U(\bm{x}_i,\bm{\theta})$. The circuit is executed, and the expectation values with respect to some Hermitian observable $O$ are evaluated. After each batch, the expectation values, as the currently fitted function values $f_{\bm{\theta}}(\bm{x})$, are fed into the cost function to compare with the actual values. Finally, the parameters are optimized according to some rule based on the gradient descent method on the cost function landscape. (b) The structure of the circuit $U(\bm{x}_i,\bm{\theta})$ when input features are encoded via the angle-embedding method with data reuploading protocol, given that a fixed encoding Hamiltonian $D$ is used throughout the circuit. (c) The structure of the variational circuit $U(\bm{x}_i,\bm{\theta})$ when input features are encoded via the amplitude-embedding method. Unlike angle-embedding circuits, the encoding layers and variational layers are distinctively separated.}
\label{fig:VQC circuits demo}
\end{figure*}

\section{Preliminaries}\label{sec: Preliminaries}
\subsection{Variational quantum circuit models}\label{subsec: VQC model}
A VQC model \cite{QVC,QVC_2,VQC}, is in general, a parametrized function $f:\mathcal{X}\times\Theta\rightarrow\mathcal{Y}$ defined as the expectation value of some observable evaluated on a parametrized quantum circuit (PQC). The function $f$ is described as
\begin{equation}\label{eq: f(x)}
    f_{\bm{\theta}}(\bm{x})=\bra{0}U(\bm{x},\bm{\theta})^{\dagger}OU(\bm{x},\bm{\theta})\ket{0}
\end{equation}
where $\mathcal{X},\mathcal{Y},\Theta$ denote the feature, target, and parameter spaces, respectively. $U(\bm{x},\bm{\theta})$ is a unitary as a function of the input feature vector $\bm{x}=(x_1,x_2,\cdots,x_d)^{T}\in\mathcal{X}$ of dimension $d$ and parameters $\bm{\theta}=(\theta_1,\theta_2,\cdots,\theta_M)^T\in\Theta$. The ultimate goal of the model is to fit an unknown, yet highly non-trivial target function $h:\mathcal{X}\rightarrow\mathcal{Y}$, such that the parameters $\bm{\theta}$ are trained towards the minimum $\bm{\theta}^*$, where
\begin{equation}
\bm{\theta}^{*}\coloneqq\underset{\bm{\theta}\in\Theta}{\arg\min}\sum\limits_{i}\mathcal{C}(\bm{x}_i,f_{\bm{\theta}}(\bm{x}_i),y_i),
\end{equation}
On the equation above, $\{(\bm{x}_i,y_i)\}\subset \mathcal{X}\cross \mathcal{Y}$ is the training dataset, $\mathcal{C}$ is the cost function that evaluates how well the model is inferring the output $f_{\bm{\theta}}(\bm{x}_i)$ by contrasting it with the correct value $y_i$. As $\bm{\theta}^*$ is minimized, $f_{\bm{\theta}^*}(\bm{x})$ approaches $h(\bm{x})$. There are many ways to iterate $\bm{\theta}_t$ from time $t$ to $t+1$ \cite{VQE_variants_summary,gradient_descent_algorithms_summary,Quantum_natural_gradient,Quantum_exact_geodesic_transport}. In general, it can be described as some variant of a gradient descent method,
\begin{equation}
    \bm{\theta}_{t+1}=\bm{\theta}_t-J\left(\bm{g}_t,\bm{g}_{t-1},\cdots,\bm{\theta}_t,\bm{\theta}_{t-1},\cdots\right),
\end{equation}
where $\bm{g}_t=\vec\nabla_{\bm{\theta}_t}\mathcal{L}$, and $J$ is some arbitrary function in terms of those gradients and parameter values along the training path. We illustrate these steps as a flowchart in \Cref{fig:VQC_general}.

\subsection{Fourier analysis on Quantum Machine Learning}\label{subsec: Fourier analysis}
To understand the behaviour of the VQC model in the frequency domain, $f$ can also be expressed as a truncated Fourier series
\begin{equation}\label{eq: Fourier series}
    f_{\bm{\theta}}(\bm{x})=\sum\limits_{\bm{\omega}\in\Omega}c_{\bm{\omega}}(\bm{\theta})e^{-i\bm{\omega}\cdot\bm{x}},
\end{equation}
where $\bm{\omega}=(\omega_1,\omega_2,\cdots,\omega_d)\in\Omega$ are the frequencies comprising the Fourier spectrum, $c_{\bm{\omega}}$ are the corresponding Fourier coefficients viewed as functions of the parameters $\bm{\theta}$. 

At a high level, irrespective of the specific embedding scheme, $U(\bm{x},\bm{\theta})$ can be decomposed into two principal components: quantum gate blocks that encode the input features $\bm{x}$ into the circuit, and parametrized gate blocks that serve as the variational component, with $\bm{\theta}$ tuned during training. The real upshot is that the feature-encoding and parameterization blocks are arranged differently depending on the VQC layout, thereby giving rise to distinct Fourier spectra and expressive capacities.
\subsubsection{Angle embedding with data reuploading}
In the well-studied case of angle embedding, the unitary consists of alternating layers of encodings and parametrizations such that
\begin{equation}\label{eq: angle embedding and reuploading unitary}
    U(\bm{x},\bm{\theta})=W^{(L+1)}(\bm{\theta}_{L+1})\prod\limits_{l=1}^{L}S^{(l)}(\bm{x})W^{(l)}(\bm{\theta}_{l}),
\end{equation}
where $L$ is the number of layers, $S^{(l)}$ denotes the encoding layer and $W^{(l)}$ denotes the parametrization layer at $l^{\text{th}}$ layer index. As the name implies, angle encoding encodes each feature $x_k$ as the angle of a unitary rotation, and is given by
\begin{equation}
    S^{(l)}(\bm{x})=\prod_{k=1}^{D}S^{(l)}(x_k)=\prod_{k=1}^{D}\exp(-ix_kH^{(l)}_k),
\end{equation}
where $H_k^{(l)}$ are Hamiltonian generators \cite{Fourier_analysis_Schuld_Maria,Fourier_analysis_expressivity_as_redundancies}. Although the superscript $l$ on $S^{(l)}$ permits, in principle, distinct encoding blocks at each layer, it is customary to take the encoding blocks to be identical across layers, $S^{(l)}\coloneqq S$. For this reason, the architecture is also referred to as data reuploading, reflecting the repeated application of identical encoding blocks to the quantum system.

With this definition of $U(\bm{x},\bm{\theta})$ as in \Cref{eq: angle embedding and reuploading unitary}, together with the assumption that $S^{(l)}(x_k)$ mutually commute for all $k$, and that $S^{(l)}(\bm{x})$ can be simultaneously diagonalized as $S^{(l)}(\bm{x})=P_l\prod_{k=1}^{D}\exp(-ix_kD^{(l)}_k)P_l^{\dagger}$ with $D^{(l)}_k=diag(\lambda^{(k,l)}_1,\lambda^{(k,l)}_2,\cdots,\lambda^{(k,l)}_{2^n})$, it was established in \cite{Fourier_analysis_Schuld_Maria,Fourier_analysis_expressivity_as_redundancies} that the following holds:

\begin{equation}\label{eq: angle encoding spectrum}
\begin{aligned}
    \Omega = \biggl\{
    \bm{\omega}=(\omega_1,\omega_2,\cdots,\omega_d)|\omega_k=\Lambda_{\bm{J}}^{(k)}-\Lambda_{\bm{J}^{\prime}}^{(k)},\bm{J},\bm{J}^{\prime}\in[\![1,2^n]\!]^{L}\biggr\},
\end{aligned}
\end{equation}
where $\Lambda_{\bm{J}}^{(k)}=\sum\limits_{l=1}^{L}\lambda^{(k,l)}_{j_{l}}$. Particularly, if the encoding Hamiltonians are fixed such that $D^{(l)}_k=I_{2^{k-1}}\otimes D\otimes I_{2^{n-k}}$ acts as a single-qubit gate acting on qubit $k$ with fixed $D$ and remains identical across all layers, the circuit layout of $U(\bm{x},\bm{\theta})$ would become as in \Cref{fig:angle_embedding_reupload_circ}, and each frequency component satisfies $\omega_k\in\{-L,-L+1,\cdots,L\}$. This implies the corresponding Fourier coefficients $c_{\bm{\omega}}$ are strictly equal to 0 whenever $\bm{\omega}$ falls beyond this range. 

Moreover, Ref. \cite{Fourier_analysis_expressivity_as_redundancies} also demonstrates that the variance of the Fourier coefficients with respect to the parameter space $\Theta$ is almost directly proportional to the `redundancies' of the current frequency $\omega$. Namely, this is how many distinct pairs $\bm{J}, \bm{J}^{\prime}$ can give rise to the target frequency $\omega$. Specifically, with exactly the same assumptions as in \Cref{eq: angle encoding spectrum} but for scalar input feature $x$, meanwhile if every trainable layer $W^{(l)}(\bm{\theta})$ independently forms a unitary 2-design, then
    \begin{equation}\label{eq: angle embedding mean}
        \mathbb{E}_{\bm{\theta}}(c_{\omega}(\bm{\theta}))=\frac{\operatorname{Tr}(O)}{2^n}\delta_{\omega}^0
    \end{equation}
    \begin{equation}\label{eq: angle embedding variance decay}
        \mathbb{V}\text{ar}_{\bm{\theta}}(c_{\omega}(\bm{\theta}))\le O\left(\alpha\frac{|\Tilde{R}(\omega)|}{2^n}\right),
    \end{equation}
    where $\alpha\coloneq\left(2^n\|O\|_2^2-\operatorname{Tr}(O)^2\right)/2^{2n}$ is a constant. Also,
    \begin{equation}
        R(\omega)=\left\{(\bm{J},\bm{J}^{\prime})|\bm{J},\bm{J}^{\prime}\in[\![1,2^n]\!]^{L},w-\Lambda_{\bm{J}}+\Lambda_{\bm{J}^{\prime}}=0\right\}
    \end{equation}
    are the eigenvalue indices that produce the target frequency $\omega$,
    \begin{equation}
        |R(\omega)|=\sum_{\bm{J}, \bm{J}^{\prime}}\delta\left(\omega-(\Lambda_{\bm{J}}-\Lambda_{\bm{J}^{\prime}})\right)=
    \begin{pmatrix}
  2nL \\
  nL-|\omega|
\end{pmatrix}
    \end{equation}
    is the redundancy of $\omega$, and $|\Tilde{R}(\omega)|=|R(\omega)|/2^{2n}$ is the normalized redundancy.

Since $\alpha$ can be bounded by a constant, and the redundancy $|R(\omega)|$ decays exponentially, this indicates that the upper bound on the variance of the Fourier coefficients both decays exponentially in the magnitude of the frequency from Ref. \cite{Fourier_analysis_expressivity_as_redundancies}, and in the number of qubits as a consequence of BP from Ref. \cite{Fourier_analysis_BP}. Although \Cref{eq: angle embedding variance decay} is formulated for a scalar input $x$ (and hence a scalar frequency $\omega$), the same exponential decay behaviour readily extends to multi-dimensional inputs $\bm{x}$, we can easily demonstrate that the general exponentially decaying claim still holds for multi-dimensional feature input $\bm{x}$.

\begin{corollary} (Exponentially decaying redundancies with multi-dimensional frequencies)\label{corollary: multi dim freq decaying redundancies}
Considering a VQC model defined as in \Cref{eq: f(x)} with $L$ layers, $n$ qubits, input feature $\bm{x}$ of dimension $d=n$, angle-embedded $U(\bm{x},\bm{\theta})$ defined as in \Cref{eq: angle embedding and reuploading unitary}, fixed encoding Hamiltonians, as well as the same premises in \Cref{eq: angle encoding spectrum}, then,
    \begin{equation}\label{eq: angle embedding variance upper bound}
        \mathbb{V}\text{ar}_{\bm{\theta}}(c_{\bm{\omega}}(\bm{\theta}))\le O\left(\alpha\frac{|\Tilde{R}(\bm{\omega})|}{2^n}\right)\le O(\exp(-\|\bm{\omega}\|_1)).
\end{equation}
    See proof in \Cref{appendix: angle embedding proofs}.
\end{corollary}
Remarkably, a more advanced protocol of trainable ternary encoding has recently been developed \cite{Fourier_trainable_tenary_encoding}. Instead of using a fixed encoding Hamiltonian, it makes the encoding gate become $S^{(l)}(x,\alpha_l)=\exp(i\alpha_l3^lxH)$, where $\alpha_l$ is a trainable parameter. In this way, it achieves an exponential reduction in the encoding gate number for the same frequency spectrum, as well as making the frequencies themselves highly trainable, thereby relaxing the restriction that frequencies must be integer-valued. However, despite these improvements, its angle-embedding nature means that its scaling with respect to the number of features is still exponentially worse compared to amplitude embedding. Hence, we start to address this in the next subsection.

\subsubsection{Amplitude embedding}
The structure of a VQC with amplitude embedding first encodes the input features into the amplitudes of the computational basis states. A trainable parametrized unitary, which may comprise multiple layers, is then applied, giving
\begin{equation}\label{eq: amplitude encoding and unitary}
    U(\bm{x},\bm{\theta})=W(\bm{\theta})\mathcal{A}(\bm{x}),
\end{equation}
as shown in \Cref{fig:amplitude_embedding_circ}, where
\begin{equation}\label{eq: amplitude encoding}
    \mathcal{A}(\bm{x})\ket{0}^{\otimes n}=\ket{\bm{x}}=\frac{1}{\|\bm{x}\|_2}\sum\limits_{i=0}^{2^n-1}x_i\ket{i}.
\end{equation}
One immediately sees that the dimension of $\bm{x}$ must be the same as the dimension of the quantum system Hilbert space; otherwise, some zero-padding is required when the dimension of the input features is smaller. For simplicity, we assume $x_i$ are always real throughout this paper, so no zero-padding is involved.

Although in general, apart from the Born rule normalizations, there is no restriction on the domain of the amplitudes \cite{encoding_summary}. However, in practice, the values of the classical information can take for encoding on each amplitude naturally abide by certain bounds for consistency. In this work, we consider two encoding scenarios applied to the same classical input features: \textit{non-negative} domain, which is natural for tasks including image processing \cite{haiyue_partial_QEC}, probability distributions \cite{encoding_prob_dists}, classical kernel mapping \cite{encoding_kernel_mappings}, and \textit{symmetric} domain, which is common when one needs to exploit phase as signs. In the former case, the components of a classical vector $\bm{x}$ are mapped to non-negative values while preserving the original differences between components, such that $\tilde{x}_i=x_i+r$ for some $r=-\min\limits_i\{x_i\}$ and $\tilde{x}_i\in[0,R]$. Naturally, one can freely define their mapping that preserves their features of interest, provided that $\tilde{x}_i\in[0,R]$ is satisfied for some $R$. In the latter case, we take $\tilde{x}_i\in[-R/2,R/2]$. Note that in both scenarios, the same original vector $\bm{x}$ is being encoded, yet the amplitudes and hence the quantum state that the model learns will differ. This is obvious since the shifted vectors $\tilde{\bm{x}}$ are different and consequently so is $\mathcal{A}(\tilde{\bm{x}})$.

Nevertheless, for a model defined as in \Cref{eq: f(x)}, we deduce that
\begin{equation}
    f_{\bm{\theta}}(\bm{x})=\frac{1}{\|\bm{x}\|_2^2}x_{i'}W(\bm{\theta})^{\dagger}_{i'j'}O_{j'j}W(\bm{\theta})_{ji}x_i,
\end{equation}
where we use Einstein's implicit summation over $i,i',j,j'$. Hence, the Fourier coefficients prescribed by \Cref{eq: Fourier series} must take the following form:
\begin{equation}\label{eq: amplitude embedding Fourier coefficients}
    c_{\bm{\omega}}(\bm{\theta})=\int_{D}d^{(d)}V\frac{x_ix_{i'}W(\bm{\theta})^{\dagger}_{i'j'}O_{j'j}W(\bm{\theta})_{ji}}{(2\pi)^d\|\bm{x}\|_2^2}e^{i\bm{\omega}\cdot\bm{x}},
\end{equation}
where $D$ is the domain over which the input features are encoded. Since understanding the behaviour of the Fourier coefficients given in \Cref{eq: amplitude embedding Fourier coefficients} is the central objective, to this end, we perform both analytical and numerical simulations in the subsequent sections. 

\section{Results}\label{sec: Results}
In this section, we present our main results regarding the behaviour of the Fourier coefficients defined in \Cref{eq: amplitude embedding Fourier coefficients} under a variety of conditions and scenarios. It includes their expressivity with respect to frequencies, dimensions, domain of the input features, both in the absence and in the presence of noise.
\begin{figure*}[thb!]
\centering
\includegraphics[width=\textwidth]{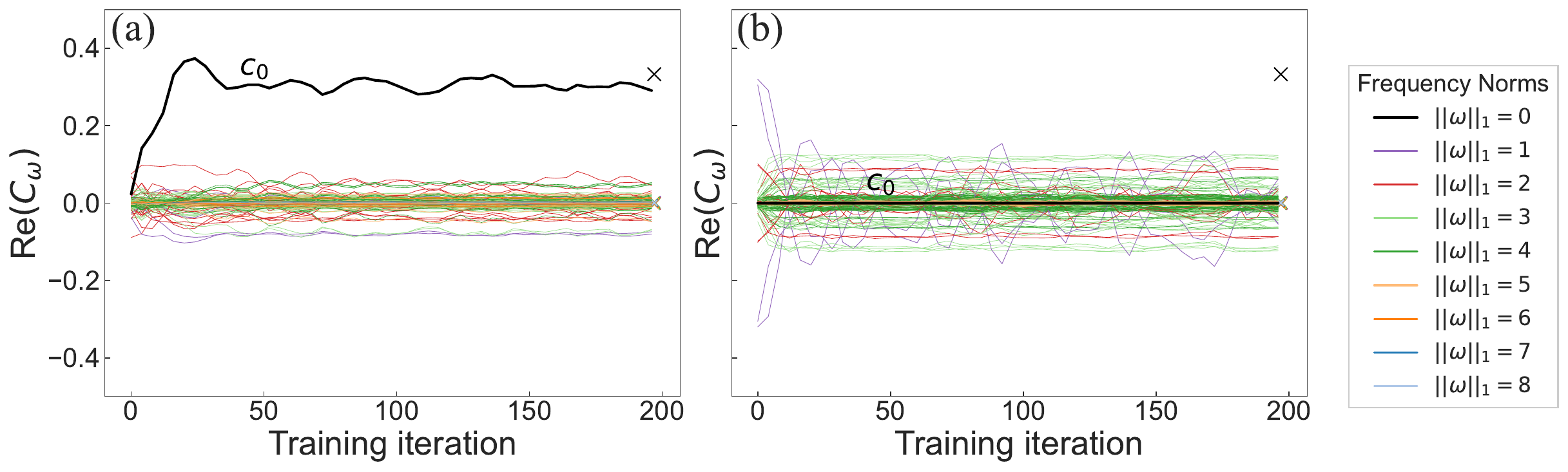}
    \caption{The values of the sampled Fourier coefficients with (a) non-negative encoding and (b) symmetric encoding, plotted for the real part. The crosses at the end denote the values of the corresponding Fourier coefficients of the target function to fit. Notably, when using non-negative encoding, $c_0$ is clearly well trained towards the target Fourier coefficient, whereas for symmetric encoding it perfectly remains at 0. The training iteration has a one-to-one mapping to the training sample size, with 400 points as one batch trained together per iteration.}
    \label{fig:symmetric_and_positive_real}
\end{figure*}
\subsection{Vanishing coefficient at zero frequency}\label{subsec: vanishing zero frequency}
Before delving into the analysis of the mean and variance of the Fourier coefficients, we first point out that the values over which each input feature can take play a pivotal role in the expressivity of the coefficients. Despite encoding the same original vector in principle, the impact on the Fourier coefficients, especially for $\bm{\omega}=0$ is starkly different between non-negative and symmetric encodings. Specifically, we show that $c_{0}(\bm{\theta})\ne0$ when the non-negative domain is used for encoding, while $c_{0}(\bm{\theta})=0$ holds identically for all $\bm{\theta}$ when the symmetric domain for encoding is employed.
\begin{theorem} (Symmetric and non-negative domain encoding zeroth coefficient)\label{theorem: zero coefficient}
    Consider a VQC model defined as in \Cref{eq: f(x)} equipped with parametrized unitary defined in \Cref{eq: amplitude encoding and unitary}, amplitude encoding in \Cref{eq: amplitude encoding}, and a traceless Hermitian observable $\operatorname{Tr}(O)=0$. Under the symmetric domain scaling, the zeroth Fourier coefficient as defined in \Cref{eq: amplitude embedding Fourier coefficients}, satisfies
    \begin{equation}
        c_{0}(\bm{\theta})=0,
    \end{equation}
    i.e., it vanishes. Under non-zero domain scaling, it is shown that 
    \begin{equation}
        c_{0}(\bm{\theta})=\frac{1}{(2\pi)^d}\sum\limits_{i>i'}2\text{Re}\left(\tilde{O}_{i,i'}(\bm{\theta})\right)\int_{D}d^{(d)}V\frac{x_ix_{i'}}{\|\bm{x}\|_2^2}\ne0.
    \end{equation}
    See the proof in \Cref{appendix: zero coefficient proof}.
\end{theorem}
\Cref{theorem: zero coefficient} makes it clear that even within the same amplitude-embedding encoding strategy, the expressivity of the model could still be substantially different. The fact that $c_{0}(\bm{\theta})=0  $ for symmetric encoding implies the model is incapable of learning any function that contains some non-zero constant values. Particularly, capturing the low-frequency features while selectively neglecting high-frequency features is almost regarded as a hallmark for a model that generalizes well and does not overfit in learning theory \cite{spectral_methods_discussion}. The inability to learn the feature at the lowest frequency would severely impair the model's performance. Beyond the derivations above, we also conduct a simulation that compares the performance and the sampled Fourier coefficients of the model between symmetric and non-negative encoding. The key findings are presented in \Cref{fig:symmetric_and_positive_real}. For completeness, we include other relevant diagnostic plots (\Cref{fig:symmetric_and_positive_MSE,fig:symmetric_and_positive_deviations}) in \Cref{appendix: supplementary plots for zeroth coefficient} for reference. In this simulation, we train the model to fit a function $g(\bm{x})$ with four-dimensional input $\bm{x}$, 
\begin{equation}
    g(\bm{x})=\sum\limits_{\bm{\omega}\in \Omega}c_{\bm{\omega}}e^{-i\bm{\omega}\cdot\bm{x}}\in[-1,1].
\end{equation}
The Fourier coefficients $c_{\bm{\omega}}$ are, for simplicity, predefined to be real numbers, with a deliberate amplification of $c_{\bm{0}}$ to highlight any potential difference in performance between the two encoding strategies during training. The target coefficients are indicated by the cross marks in \Cref{fig:symmetric_and_positive_real} that denote the target coefficients. The VQC consists of 2 qubits (hence encodes 4 dimensions of features), with the variational unitary defined as 
\begin{equation}\label{eq: variational circuit definition for simulation}
    U(\bm{\theta})=\prod\limits_{l}W_{l}\prod\limits_{m}e^{-i\theta_{lm}H_{lm}},
\end{equation}
where $W_l=\prod_{k=1}^{n-1} CZ_{k,k+1}$ is an staggered layer of Controlled-$Z$ entangling gates identical for all $L=30$ layers, and $W_l=CZ_{1,2}$ for this specific 2-qubit case. $H_{lm}=\hat{\bm{n}}_{lm}\cdot\bm{\sigma}$, $\bm{\sigma}=(X,Y,Z)^{T}$ are the generating Hamiltonian drawn from the Pauli group with axis $\hat{\bm{n}}_{lm}$. Each rotation $R_{lm}=\exp(-i\theta_{lm}H_{lm})$ possesses three degrees of freedom, which can be parametrized into the Euler angles ($\alpha, \beta, \gamma$), 
\begin{equation}\label{eq: euler angle decomposition}
    R_{lm}=R_{Z}(\alpha_{lm})R_{Y}(\beta_{lm})R_{Z}(\gamma_{lm}).
\end{equation}
The observable is taken to be a single-qubit Pauli-$Z$ operator, $O=Z_1$. The parameters are updated using the ADAM optimizer, with a learning rate of 0.01, 400 data points per batch $\mathcal{B}\subset\mathcal{X}$, and a total of 50 batches in the training set $\mathcal{X}$. The cost function is defined as the mean squared error (MSE) of the predicted value to the actual values:
\begin{equation}
    \mathcal{C}(g',\mathcal{B})=\frac{1}{N}\sum\limits_{i\in\mathcal{B}}\left(g'(\bm{x}_i)-y_i\right)^2,
\end{equation}
where $N=|\mathcal{B}|$ is the size of the batch. All simulations are executed on a noiseless simulator. 

From \Cref{fig:symmetric_and_positive_real}, it is apparent that the zeroth coefficient $c_0(\bm{\theta})$ for non-negative encoding converges towards the target coefficient during training (marked by the black cross), while it stays perfectly at 0 for symmetric encoding regardless of the parameter trajectory. Moreover, in \Cref{fig:symmetric_and_positive_MSE}, we barely see the test set MSE decrease when using symmetric encoding, reflecting its complete loss of expressive power. Meanwhile, the MSE for non-negative encoding decays significantly across training. This corroborates our prediction drawn from \Cref{theorem: zero coefficient}, that the difference in the expressivity of one yet most significant Fourier coefficient can lead to a qualitatively different performance of the model. In terms of the non-zero coefficients $c_{\bm{\omega}\ne0}$, however, both symmetric and non-negative encoding models seem to exhibit a certain level of expressivity. In particular, under non-negative encoding, smaller frequency norms tend to admit larger and more volatile Fourier coefficient values. The imaginary part plot are not included since the imaginary part of the targets is zero anyway. In the next subsection, we focus on discussing the exact relationship between the expressivity of Fourier coefficients and their frequency, input dimension, and choice of the observable.

\subsection{Analytic moments of non-negative embedding Fourier coefficients}\label{subsec: moments of fourier coeffs}
This subsection presents the core findings of this paper: To understand the expressivity of the Fourier coefficients under amplitude embedding, we conduct analytic derivations of their first and second moments with respect to the parameter space under non-negative encodings. Hence, we show that, similar to angle embedding, amplitude embedding also exhibits exponential decay of expressivity with frequency, but with subtle differences in the precise scaling behaviour.

We first begin with the noiseless case. That is, $U(\bm{\theta})$ acts as a purely unitary channel. We establish the following theorem:
\begin{theorem} (Moments of Fourier coefficients for noiseless amplitude-embedding VQCs)\label{theorem: noiseless mean and variance}
Consider a VQC model defined as in \Cref{eq: f(x)} with $L$ layers, $n$ qubits, with non-negative amplitude embedding featuring $d=2^n$ input dimensions and $x_i\in[0,2\pi]$, and the parametrized unitary defined in \Cref{eq: amplitude encoding and unitary} and \Cref{eq: amplitude encoding}. Assuming the image of the parameter space $\text{Im}(\mathcal{W})$, $\mathcal{W}:\Theta\rightarrow \mathcal{U}$ within the unitary group $\mathcal{U}\subset\text{U}(d)$ is a 2-design, the corresponding Fourier coefficients as defined in \Cref{eq: amplitude embedding Fourier coefficients} for discretized, integer-valued frequency $\bm{\omega}$ will have its mean and variance
    \begin{equation}\label{eq: amplitude embedding noiseless mean}
        \mathbb{E}_{\bm{\theta}}(c_{\bm{\omega}}(\bm{\theta}))=\frac{\operatorname{tr}(O)}{d}\delta_{\bm{\omega}}^0,
    \end{equation}
    \begin{equation}\label{eq: amplitude embedding noiseless variance}
    \begin{aligned}
        \mathbb{V}\text{ar}_{\bm{\theta}}(c_{\bm{\omega}}(\bm{\theta}))&=\frac{1}{d^2-1}\biggl[\left(\operatorname{tr}(O)^2-\frac{\|O\|_F^2}{d}\right)\delta_{\bm{\omega}}^0\\
        &+\left(\|O\|_F^2-\frac{\operatorname{tr}(O)^2}{d}\right)\frac{\mathrm{I}(\bm{\omega})}{(2\pi)^{2d}}\biggr]-\frac{\operatorname{tr}\left(O\right)^2\delta_{\bm{\omega}}^0}{d^2},
     \end{aligned}
    \end{equation}
    where
    \begin{equation}\label{eq: I(w)}
        \mathrm{I}(\bm{\omega})=\sum\limits_{i,j=1}^d\left|\int d^{(d)}V\frac{e^{-i\bm{\omega}\cdot\bm{x}}x_ix_j}{\|\bm{x}\|_2^2}\right|^2.
    \end{equation}
    See the proof in \Cref{appendix: noiseless mean and variance proof}.
\end{theorem}
\begin{figure}[t!]
    \includegraphics[width=\linewidth]{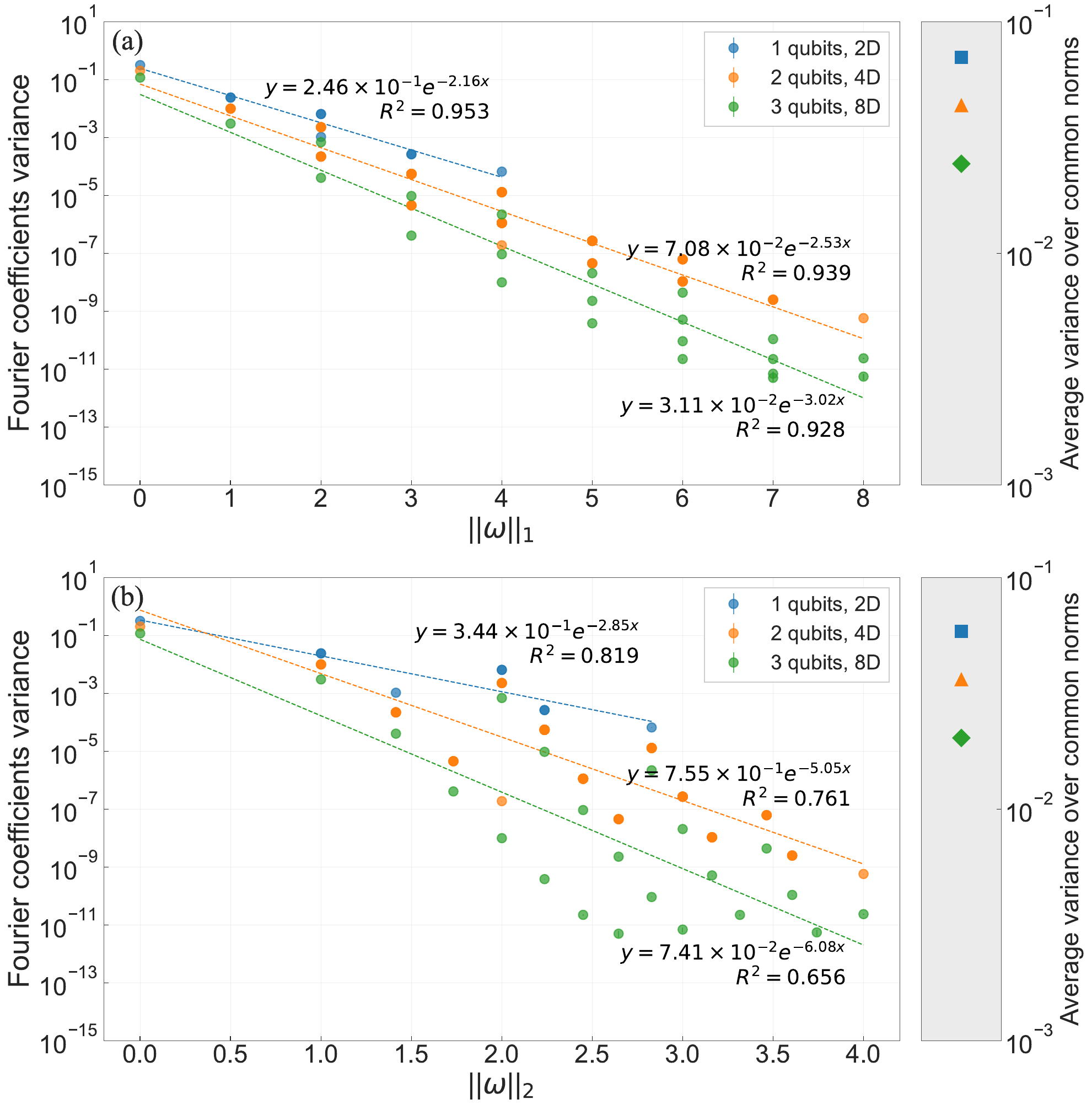}
    \caption{Variance of the Fourier coefficients under non-negative amplitude embedding VQC model with simplified condition on observable $O\in\mathcal{P}_n$ discussed as in \Cref{eq: amplitude embedding noiseless variance Tr(O)=0}, plotted against the (a) $L_1$ and (b) $L_2$ frequency norms. The right panel plots the variance of the coefficients averaged over all common frequencies ($\|\bm{\omega}\|_1\le4$ and $\|\bm{\omega}\|_2\le\|(2,2)\|_2$) for each qubit count. The integral $\mathrm{I}(\bm{\omega})$ is evaluated numerically via Mathematica, with its error $\epsilon$ estimated according to $\epsilon=\alpha+\beta\mathrm{I}(\bm{\omega})$, where $\alpha=10^{-3},\beta=10^{-5}$ are the absolute and proportional tolerance, respectively. However, due to restrictions on computation resources, Monte Carlo methods are used to evaluate $\mathrm{I}(\bm{\omega})$ for the $d=8$ case, where we mitigate the estimated error according to \Cref{appendix: Monte Carlo Integration error estimation}. In both cases, the error bars are negligible in the scale of this plot.}
    \label{fig: amplitude_embedding_coefficients_variance_scaling}
\end{figure}
The above theorem clearly demonstrates the dependence of the Fourier coefficients' mean and variances on the feature dimension, observable, and the frequency. In particular, we see a special term $\mathrm{I}(\bm{\omega})$ appearing in the variance, which dominates the overall pattern of the variance as a function of the frequencies. Although integer-valued frequencies are assumed in \Cref{theorem: noiseless mean and variance}, it does not deny the existence of non-zero amplitude Fourier coefficient with non-integer frequencies, which are demonstrated in \Cref{subsec: non periodic functions}. However, $\mathrm{I}(\bm{\omega})$ is defined as an integral over the input feature space, which is difficult to evaluate in closed form. Therefore, we take a step back and instead derive an power-law upper-bound of $\mathrm{I}(\bm{\omega})$ to narrow it down within a certain family of functions (\Cref{corollary: upper bound of I(w)}), as well as numerically evaluate $\mathrm{I}(\bm{\omega})$ for a range of angular frequencies $\bm{\omega}$ across different input dimensions. 
\begin{corollary} (Upper-bound of \Cref{eq: I(w)} $\mathrm{I}(\bm{\omega})$)\label{corollary: upper bound of I(w)}
    For the sum of squared integrals as described in \Cref{eq: I(w)} with $S=\{l,\omega_l\ne0\}$, $s=|S|$ are the set of non-zero frequency indices and its counts, given $d>s$, its upper bound can be shown as
    \begin{equation}\label{eq: upper bound of I(w)}
        \mathrm{I}(\bm{\omega})\le d^2(2\pi)^{2d}\left(\frac{2}{\pi}\right)^{2s}\left(1+\frac{2}{d-s}\right)^2\prod\limits_{l\in S}|\omega_l|^{-2}.
    \end{equation}

    See proof in \Cref{appendix: upper-bound of I(w) proof}.
\end{corollary}
In addition, one may impose further conditions on the observable $O$ by requiring $\operatorname{tr}(O)=0$, a criterion commonly adopted for VQC models, which is satisfied by all Pauli observables $O\in \mathcal{P}_n$. In particular, if we focus on Pauli observables, then $\|O\|_F^2=d$, and \Cref{eq: amplitude embedding noiseless mean,eq: amplitude embedding noiseless variance} in \Cref{theorem: noiseless mean and variance} is simplified as
\begin{equation}\label{eq: amplitude embedding noiseless mean Tr(O)=0}
     \mathbb{E}_{\bm{\theta}}(c_{\bm{\omega}}(\bm{\theta}))=0
\end{equation}
\begin{equation}\label{eq: amplitude embedding noiseless variance Tr(O)=0}
    \mathbb{V}\text{ar}_{\bm{\theta}}(c_{\bm{\omega}}(\bm{\theta}))=\frac{d}{d^2-1}\biggl[\frac{\mathrm{I}(\bm{\omega})}{(2\pi)^{2d}}-\frac{1}{d}\delta_{\bm{\omega}}^0\biggr].
\end{equation}
\Cref{corollary: upper bound of I(w)} demonstrates a clear directional dependence of the $\mathrm{I}(\bm{\omega})$ and hence an upper bound of the coefficient variance as long as the frequency vector contains at least one zero component. In the limit of only non-zero components, the upper bound is undefined. However, if the frequency vector is almost uniform with only one zero component in all directions, it is easy to see that $\mathrm{I}(\bm{\omega})=O(\|\bm{\omega}\|_{\infty}^{-2(d-1)})$. In the limit of a frequency vector spiked in one direction , it is $\mathrm{I}(\bm{\omega})=O(\|\bm{\omega}\|_{\infty}^{-2})$ and the rest scenarios would sit somewhere in between. Nevertheless, it is clear that the Fourier coefficient variance is bounded by at most a power-law function with respect to the frequency magnitude. Together with the numerically evaluated $\mathrm{I}(\bm{\omega})$ via Mathematica, we plot the variance prescribed by \Cref{eq: amplitude embedding noiseless variance Tr(O)=0} against its frequency norm in \Cref{fig: amplitude_embedding_coefficients_variance_scaling}. From numerical evidence, at least within the range of the plot, the Fourier coefficient variance appears to decay exponentially as the frequency norm increases. Meanwhile, the coefficient variance also decays exponentially in the number of qubits, as indicated by the average variance over common frequency norms in \Cref{fig: amplitude_embedding_coefficients_variance_scaling}. This is in agreement with the prediction of \cite{Fourier_analysis_BP}. Of the two choices of the norms examined in the plots, the $L_1$ norm plot provides a better fit and is more consistent with the standard exponentially decaying curve, as indicated by best fit lines and the $R^2$ coefficient of determination metric. As a result, in the noiseless case of non-symmetric amplitude embedding VQC model, there is a strong numerical evidence of $\mathbb{V}\text{ar}_{\bm{\theta}}(c_{\bm{\omega}}(\bm{\theta}))=O\left(\exp(-\|\bm{\omega}\|_1)\right)$, which is enough to provide a useful quantitative guidance for ones who aims to train an amplitude-embedded model with reasonable target frequencies and amplitudes. Also, it is certainly upper bounded by a power-law ordered function with respect to the frequency magnitude $\mathbb{V}\text{ar}_{\bm{\theta}}(c_{\bm{\omega}}(\bm{\theta}))\le O\left(\|\bm{\omega}\|_{\infty}^{-2s}\right)$, $1\le s <d$. At this stage, it might seem that the variance scaling is very similar for both amplitude encoding and angle embedding. However, we remark that the scaling for the angle embedding previously demonstrated in \Cref{eq: angle embedding variance upper bound} is an upper bound, rather than an exact expression. In the next subsection, building upon the simulation with noisy quantum channels, we will also compare the scaling of the standard deviations between the standard angle embedding and non-negative amplitude embedding, hence highlighting the subtle differences in their expressivity.

\subsection{Simulation of the noisy Fourier coefficients}\label{subsec: simulations}
This subsection discusses the role of different noise channels in modulating the Fourier coefficients of non-negative amplitude embedding VQC models, which is a topic that lacks exploration in current studies. We first consider a subset of general time-independent noise channels defined on a set of Kraus operators that are themselves unitaries. Namely,
\begin{equation}\label{eq: general unitary Kraus channel}
    \mathcal{E}(\rho)=\sum\limits_{k=1}^{K}p_kE_k\rho E_k^{\dagger},
\end{equation}
where $p_k$ is the probability of breaking evolving to the state $E_k\rho E_k^{\dagger}$ in each instance that satisfies Born rule $\sum_kp_k=1$. $E_k\in U(d)$ are the Kraus operators that are also unitaries, and $\sum_kE_kE_k^{\dagger}=\mathds{1}$ is trace-preseriving. To model such a channel as part of the VQC, we assume that noise acts after each gate in the quantum circuit that combines into the unitary $W(\bm{\theta})$ we defined previously \Cref{eq: amplitude encoding and unitary}. That is, the combined channel reads $\tilde{\Lambda_i}\coloneqq \mathcal{E}_i\circ \Lambda_i$ for each gate indexed by $i$ in the circuit, where $\Lambda_i(\rho)=W_i\rho W_i^{\dagger}$ and $\mathcal{E}_i$ acts on the same qubits as $\Lambda_i$. In fact, with this general noise model in mind, we can derive the following theorem:

\begin{theorem} (Moments of Fourier coefficients for amplitude-embedding VQCs under depolarizing noise)\label{theorem: depolarizing noise mean and variance}
Consider a VQC model defined as in \Cref{eq: f(x)}, with non-negative amplitude embedding satisfyingan $x_i\in[0,2\pi]$, and the parametrized unitary defined as in \Cref{eq: amplitude encoding and unitary}, where each gate in the circuit is followed by a noise channel as defined in \Cref{eq: general unitary Kraus channel}. Assuming $\text{Im}(\mathcal{W})$, $\mathcal{W}:\Theta\rightarrow \mathcal{U}$ within the unitary group $\mathcal{U}\subset\text{U}(d)$ is a 2-design, and the perturbed unitaries $\tilde{W}(\bm{\theta})_k,\tilde{W}(\bm{\theta})_{k'}$ induced by any two distinct indices $k\ne k'$ corresponding to two different noise realizations are assumed to be approximately uncorrelated, then the corresponding Fourier coefficients, as defined in \Cref{eq: amplitude embedding Fourier coefficients} with discretized, integer frequencies $\bm{\omega}$, will have its mean and variance
\begin{widetext}
    \begin{equation}\label{eq: amplitude embedding noisy mean}
        \mathbb{E}_{\bm{\theta}}(\tilde{c}_{\omega}(\bm{\theta}))=\frac{\operatorname{tr}(O)}{d}\delta_{\bm{\omega}}^0,
    \end{equation}
    \begin{equation}\label{eq: amplitude embedding noisy variance}
    \begin{aligned}
        \mathbb{V}\text{ar}_{\bm{\theta}}(\tilde{c}_{\omega}(\bm{\theta}))=\frac{1}{d^2-1}\biggl[\left(\operatorname{tr}(O)^2\left(2-\frac{1}{d^2}\right)-\frac{\|O\|_F^2}{d}\right)\delta_{\bm{\omega}}^0+\left(\|O\|_F^2-\frac{\operatorname{tr}(O)^2}{d}\right)&\frac{\mathrm{I}(\bm{\omega})}{(2\pi)^{2d}}\biggr]\left(\sum\limits_{k=1}^Kp_k^2\right)^{Q}\\
        &-\frac{\operatorname{tr}\left(O\right)^2\delta_{\bm{\omega}}^0}{d^2},
    \end{aligned}
    \end{equation}
\end{widetext}
    where $Q$ is the number of times the channel ever acts on the system in $W(\bm{\theta})$, and 
    \begin{equation}
        \mathrm{I}(\bm{\omega})=\sum\limits_{i,j=1}^d\left|\int d^{(d)}V\frac{e^{-i\bm{\omega}\cdot\bm{x}}x_ix_j}{\|\bm{x}\|_2^2}\right|^2
    \end{equation}
    retains the same definition as in \Cref{theorem: noiseless mean and variance}. See proof in \Cref{appendix: noisy amplitude embedding proofs}.
\end{theorem}
We remark that the approximation of vanishing correlations between the outputs of different perturbed unitaries has already been practiced and studied in scrambling properties of random circuits \cite{dalzell2024random}. In Ref.~\cite{dalzell2024random}, a similar approximation is adopted in the heuristic analysis of noisy random circuits, where correlations between distinct Pauli error patterns are neglected. Besides, the correlation between two perturbed unitaries has been studied analytically in Ref.~\cite{west2025graph}. Consider $\tilde{W}_s=VU$ and $\tilde{W}_{s'}=VPU$, where $U$ and $V$ are independently sampled from Haar measures of $\mathrm{U}(d)$, $d=2^n$, and $P,O\in\mathcal{P}_n$ are the Pauli strings that mark the difference in the error trajectories and observable, respectively. By applying twirling over $U,V$, followed by the averaged out-of-time-order correlator identity in Ref.~\cite{west2025graph}, we obtain
\begin{equation}
\label{eq: perturbed unitary correlator}
\begin{aligned}
    &\frac{1}{d}\mathbb{E}_{U,V}\left[
    \operatorname{tr}\left(\tilde{W}_s^\dagger O\tilde{W}_s
    \tilde{W}_{s'}^\dagger O\tilde{W}_{s'}\right)\right]\\
    &\quad=\frac{1}{d}\mathbb{E}_{U,V}\left[\operatorname{tr}\left(U^\dagger V^\dagger OVUU^{\dagger}PV^\dagger OVPU\right)\right]\\
    &\quad=\frac{1}{d}\mathbb{E}_{V}\left[\operatorname{tr}\left(V^\dagger OVPV^\dagger OVP\right)\right]
    =-\frac{1}{d^2-1},
\end{aligned}
\end{equation}
which exactly quantifies the amount of correlations between different local operations scrambled by the circuit. When $P=\mathds{1}$, the correlator becomes exactly 1, whereas it scales in $O(4^{-n})$ when $P\ne\mathds{1}$. This exponential gap supports our approximate in this scrambling limit particularly at a reasonable number of qubits where VQC remain useful, while neglecting cross terms between arbitrary noise realizations remains an additional approximation.

We immediately see from the theorem above that the noise induces a prefactor $\left(\sum_kp_k^2\right)^{Q}<1$ that suppresses the variance of the Fourier coefficients exponentially with respect to the number of times the channel acts on the quantum system. Usually, one can safely assume $Q=O\left(\text{poly}(nL)\right)$. That is, the number of times the channel acts on the system scales polynomially with respect to the product of the number of qubits and the depth of the VQC. Adopting the same simplification $O\in\mathcal{P}_n$ as in \Cref{eq: amplitude embedding noiseless mean Tr(O)=0,eq: amplitude embedding noiseless variance Tr(O)=0}, we can further express \Cref{eq: amplitude embedding noisy mean,eq: amplitude embedding noisy variance} as
\begin{equation}\label{eq: amplitude embedding noisy mean Tr(O)=0}
     \mathbb{E}_{\bm{\theta}}(c_{\bm{\omega}}(\bm{\theta}))=0
\end{equation}
\begin{equation}\label{eq: amplitude embedding noisy variance Tr(O)=0}
        \mathbb{V}\text{ar}_{\bm{\theta}}(c_{\bm{\omega}}(\bm{\theta}))=\frac{d}{d^2-1}\biggl[\frac{\mathrm{I}(\bm{\omega})}{(2\pi)^{2d}}-\frac{1}{d}\delta_{\bm{\omega}}^0\biggr]\left(\sum\limits_{k=1}^Kp_k^2\right)^{Q}.
\end{equation}
Evidently, when a noise channel with general unitary Kraus operators is taken into consideration, the sole effect is the appearance of a decay factor that multiplies the variance, which exponentially suppresses the variance as either the number of noise applications $Q$ grows.

In addition to the foregoing analytical treatment, we also complement it with numerical simulations of noisy VQC. Among available options that satisfy \Cref{eq: general unitary Kraus channel}, Pauli channels constitute the most important and most thoroughly studied class of noise models in quantum computing. Due to their highly symmetric structure and applications in QEC \cite{surface_codes,qec_lattice_surgery}, arbitrary noise processes can be mapped into Pauli errors upon syndrome measurement. Hence, we choose the depolarizing channel at different strengths $p$ to be analyzed in this subsection, such that
\begin{equation}\label{eq: depol channel}
    \mathcal{E}(\rho)=(1-p)\rho+\frac{p}{3}\left(X\rho X+Y\rho Y+Z\rho Z\right).
\end{equation}
In this simulation, we adopt the same general setup and hyperparameters as in the previous demonstration of vanishing coefficient at zero frequency. This includes the definition of variational circuit as in \Cref{eq: variational circuit definition for simulation} that contains $L=30$ layers of alternating single-qubit rotations and two-qubit $CZ$ gates, choice of observable, learning rate, ADAM optimizer, and cost functions. However, for simplicity, we restrict the batch size to 100 data points for a total of 50 batches. The target functions are deliberately constructed such that selected coefficients are set to equal, non-zero values, while the rest remain strictly at zero. 
\begin{figure*}[thb!]
    \centering
    \begin{subfigure}{.49\textwidth}
        \hspace{7pt}
        \includegraphics[width=.99\columnwidth]{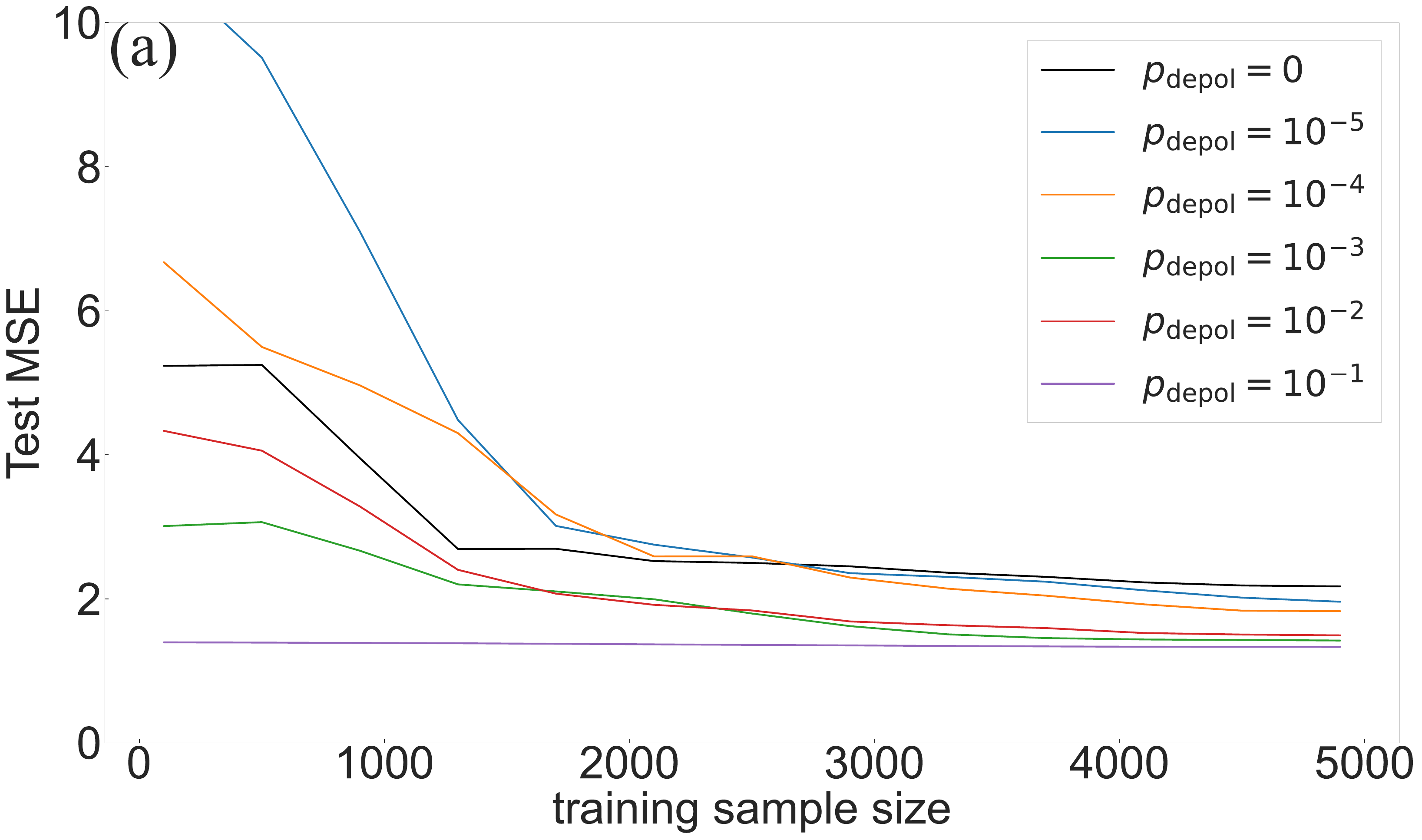}
        \phantomcaption
        \label{fig:2q_depol_MSE}
    \end{subfigure}
    \begin{subfigure}{.49\textwidth}
        \hspace{2pt}
        \includegraphics[width=\columnwidth]{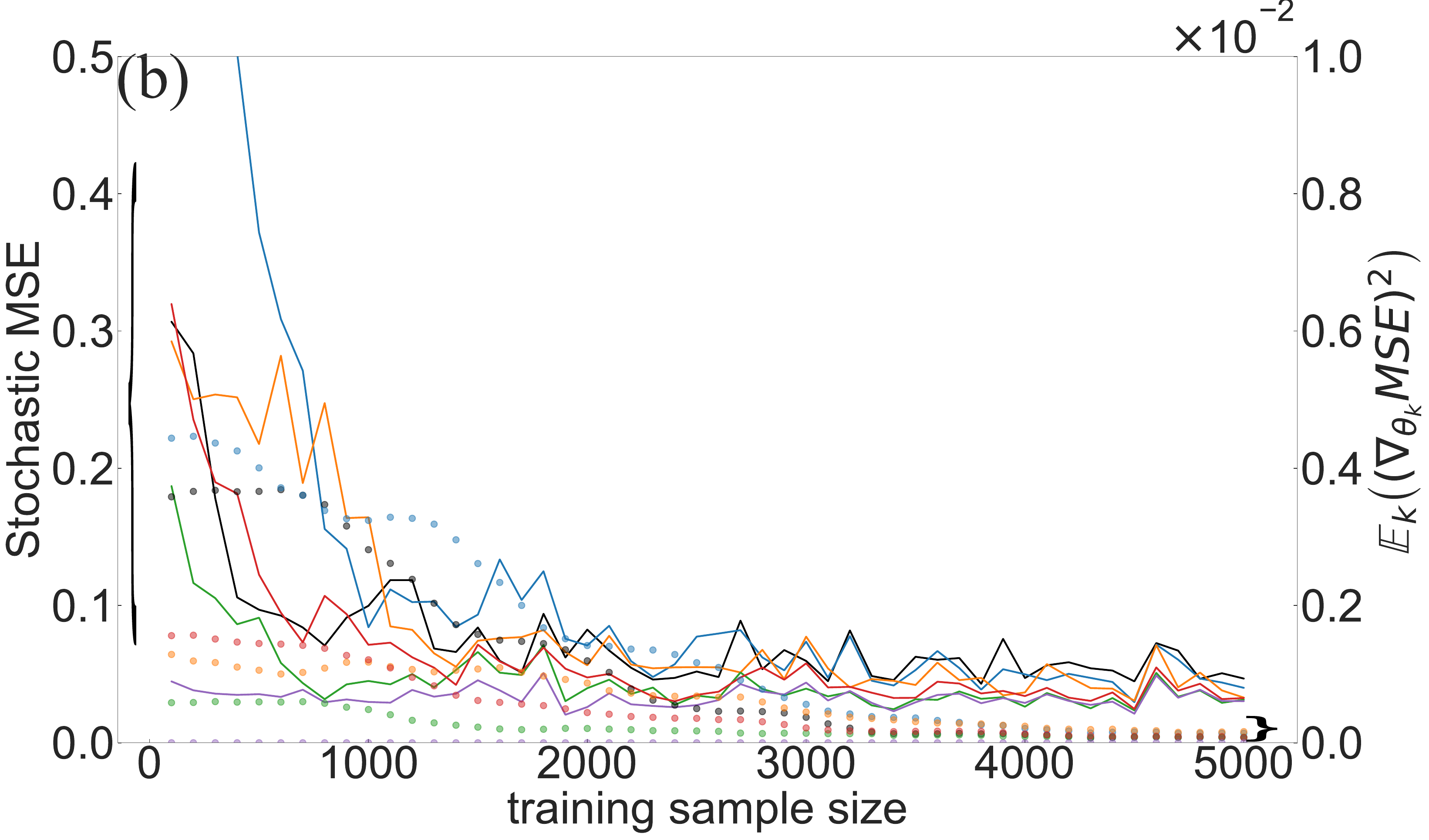}
        \phantomcaption
        \label{fig:2q_depol_MSE_gradients}
    \end{subfigure}
    \begin{subfigure}{.49\textwidth}
        \centering
        \includegraphics[width=\columnwidth]{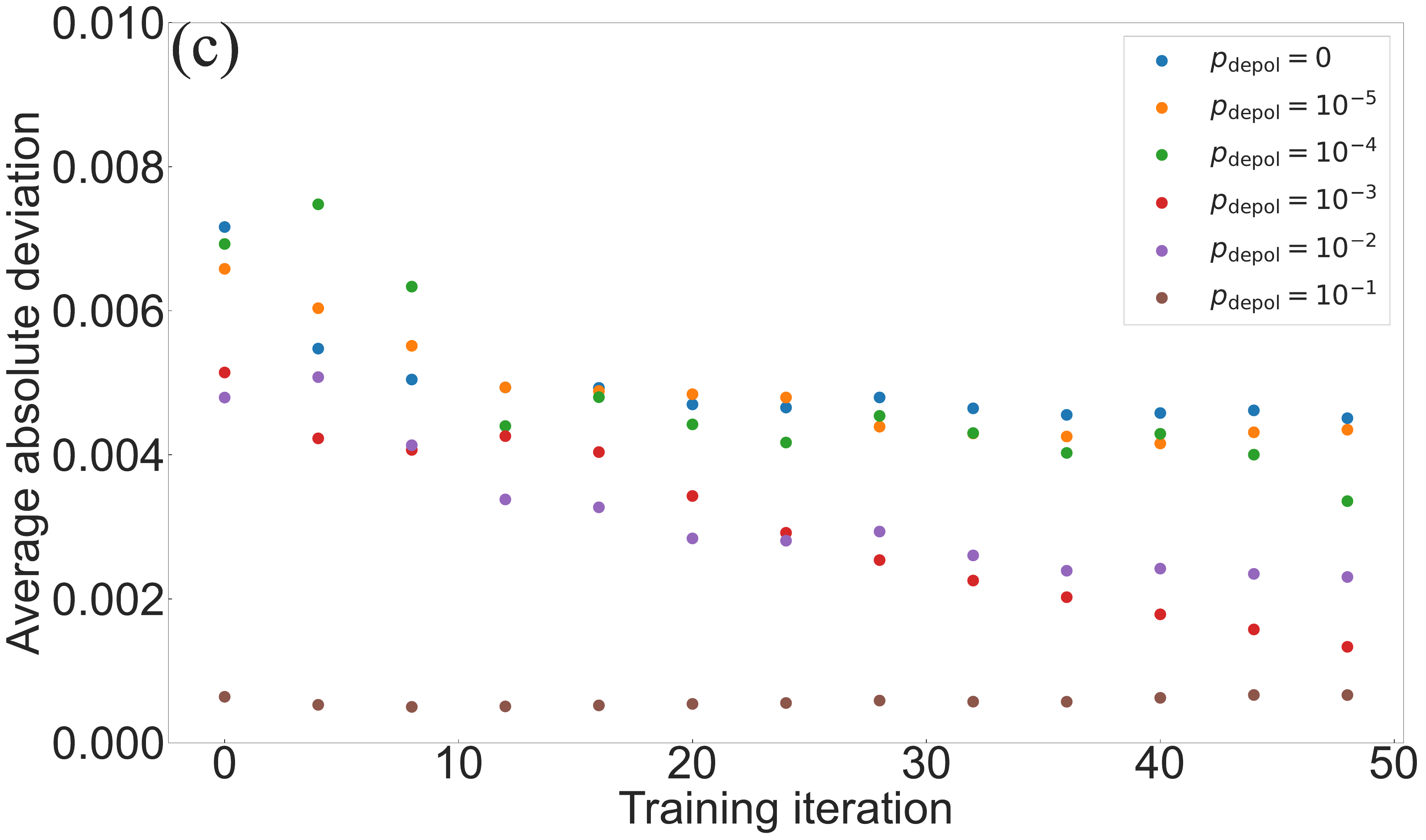}
        \phantomcaption
        \label{fig:2q_depol_zero_coeff_average_deviations}
    \end{subfigure}
    \begin{subfigure}{.49\textwidth}
        \centering
        \includegraphics[width=\columnwidth]{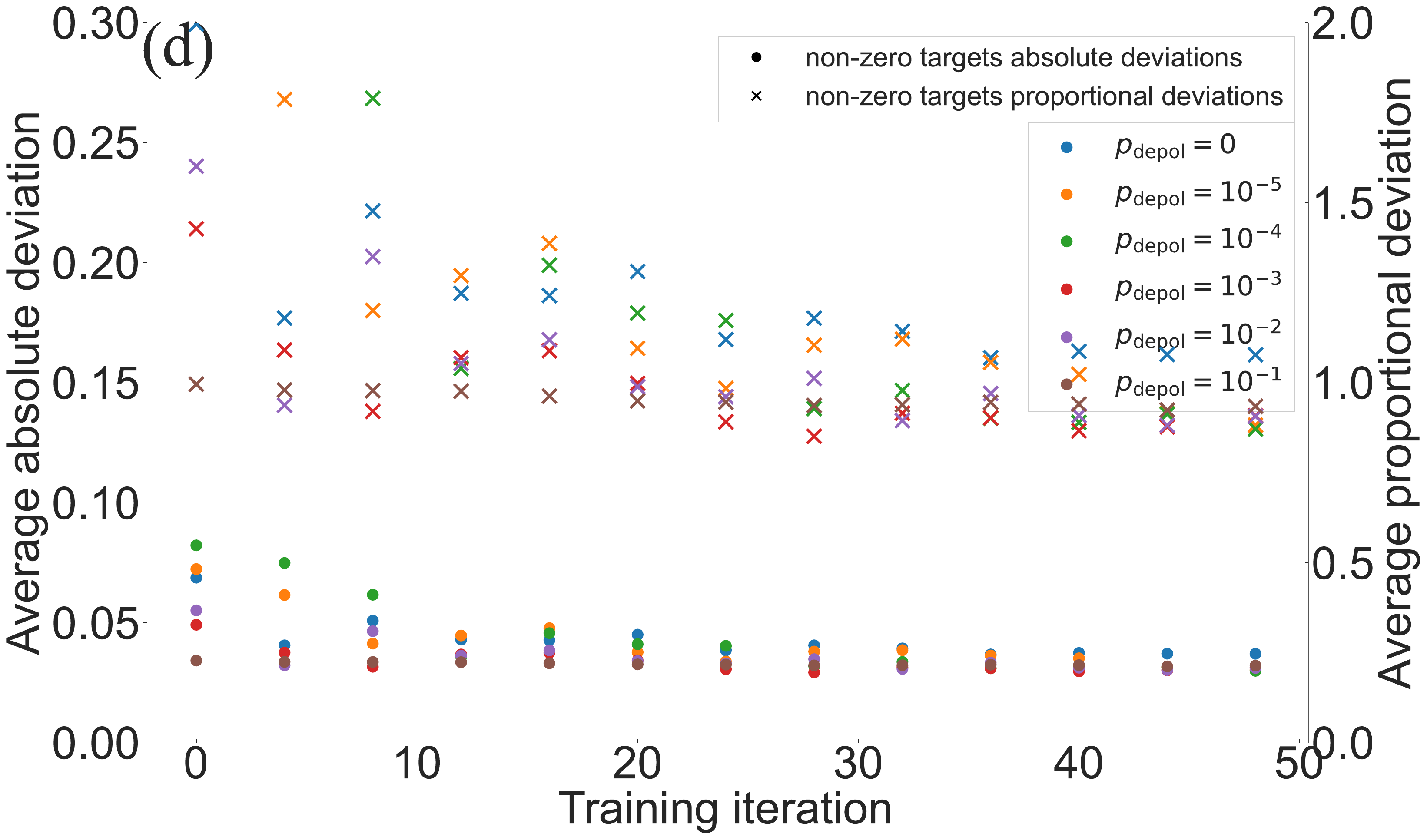}
        \phantomcaption
        \label{fig:2q_depol_nonzero_coeff_average_deviations}
    \end{subfigure}
    \caption{(a) Test-set MSE versus the number of training samples. (b) Left panel: stochastic batch-wise MSE (lines). Right panel: mean squared gradient (scattered dots), also plotted against the number of training samples. (c) The mean absolute deviation of the sampled Fourier coefficients from their zero-valued targets, $\mathop{\mathbb{E}}\limits_{\bm{\omega}\in\Omega,\,c_{\bm{\omega}}=0}\bigl(|c_{\bm{\omega}}-c'_{\bm{\omega}}|\bigr)$. where $c_{\bm{\omega}}$ and $c'_{\bm{\omega}}$ denote the target and sampled Fourier coefficients, respectively. (d) The mean absolute deviation of the sampled Fourier coefficients from their non-zero target values (scattered dots, left axis), $\mathop{\mathbb{E}}\limits_{\bm{\omega}\in\Omega,\,c_{\bm{\omega}}\neq 0}\bigl(|c_{\bm{\omega}}-c'_{\bm{\omega}}|\bigr)$, together with the corresponding proportional deviation $\mathop{\mathbb{E}}\limits_{\bm{\omega}\in\Omega,\,c_{\bm{\omega}}\neq 0}\bigl(\bigl|\frac{c_{\bm{\omega}}-c'_{\bm{\omega}}}{c_{\bm{\omega}}}\bigr|\bigr)$ (scattered crosses, right axis). Across all frequencies, the deviations decrease with the number of training iterations yet eventually plateau at a non-zero constant, indicating a potential limitation in the expressive capacity of the model.}
    \label{fig: 2q depol channel coefficients analysis part 1}
\end{figure*}
\begin{figure*}[t!]
    \begin{subfigure}{\textwidth}
        \centering
        \includegraphics[width=.9\columnwidth]{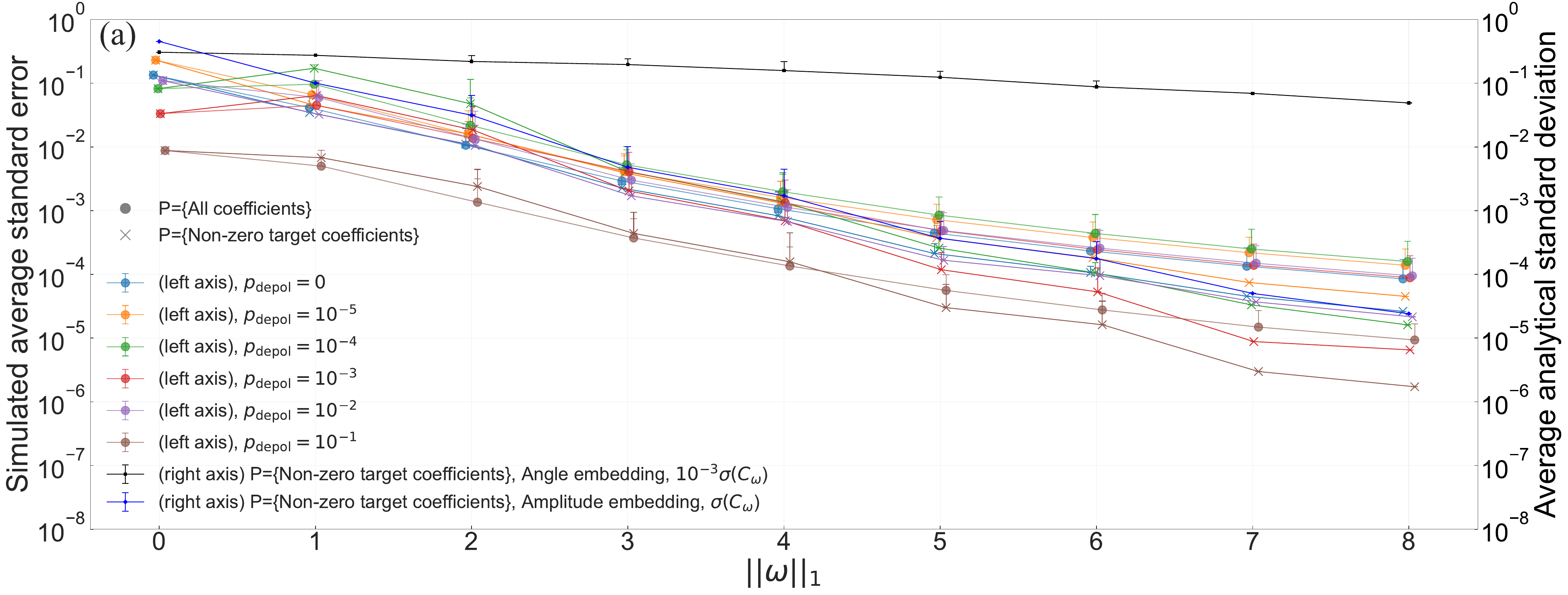}
        \phantomcaption
        \label{fig:2q_depol_coeff_std_frequency_L1}
    \end{subfigure}
    \begin{subfigure}{\textwidth}
        \centering
        \includegraphics[width=.9\columnwidth]{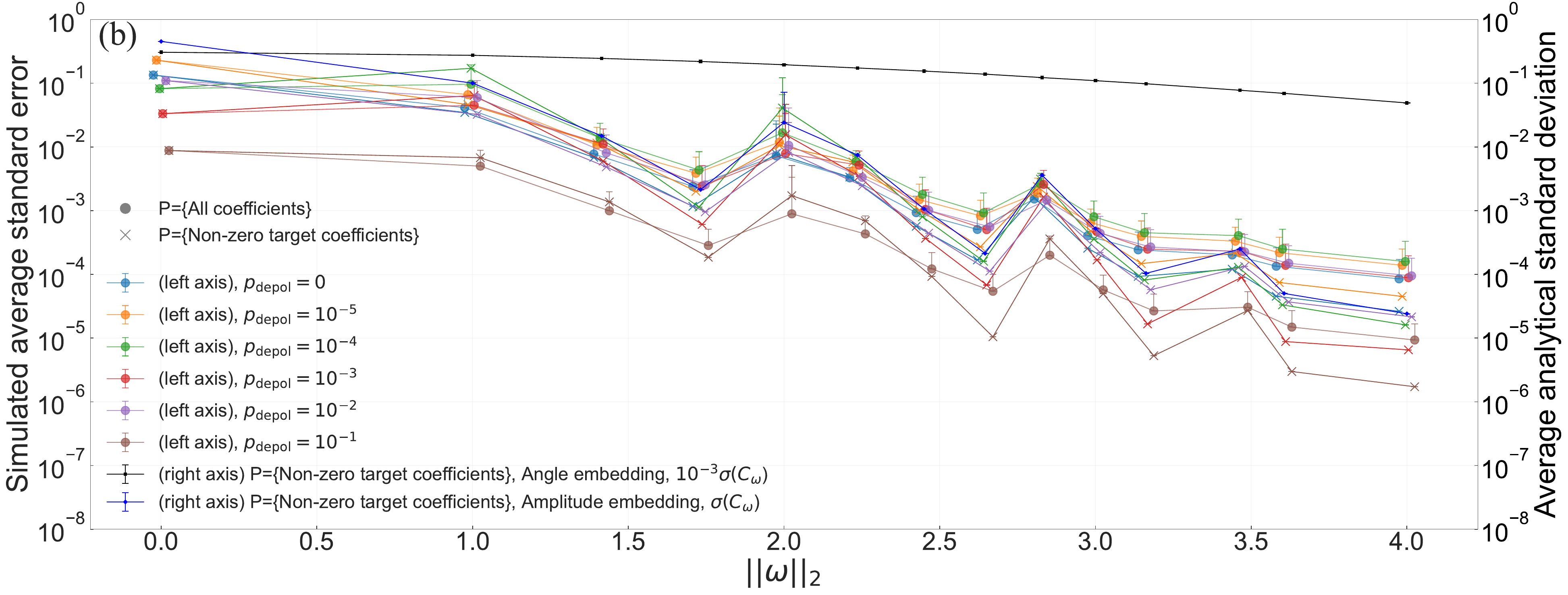}
        \phantomcaption
        \label{fig:2q_depol_coeff_std_frequency_L2}
    \end{subfigure}
    \caption{
    (a) Left axis: empirical mean of the standard errors of the Fourier coefficients versus the $L_1$ norm of their frequency, $\|\bm{\omega}\|_1$. Each standard error is computed as $\sigma(c_{\bm{\omega}})=\sqrt{\mathbb{E}_{t\in T}\bigl[|c_{\bm{\omega}}|^2\bigr]-\bigl|\mathbb{E}_{t\in T}[c_{\bm{\omega}}]\bigr|^2}$, where $T$ denotes the set of training iterations.  After computing $\sigma(c_{\bm{\omega}})$ for each frequency individually, the results are grouped by the norm $\|\bm{\omega}\|_1$ and by membership in a set $P$, yielding the conditional mean $\mathbb{E}_{\bm{\omega}: \|\bm{\omega}\|_1 = \ell,\,\bm{\omega}\in P}\,[\sigma(c_{\bm{\omega}})]$.  Two choices of $P$ are considered: $P=\{\bm{\omega}:c_{\bm{\omega}}\neq 0\}$ (non-zero target coefficients only, shown as crosses) and $P=\Omega$ (all coefficients, shown as dots). For visual clarity and to avoid infinitely long error bars on the log-scale, only the upper half of each error bar is displayed. Results are shown for depolarizing noise strengths ranging from $p=0$ to $p=10^{-1}$, distinguished by colour. Right axis: standard deviations averaged over the same frequency groups, computed from the exact noiseless amplitude-embedding variance derived in \Cref{eq: amplitude embedding noiseless variance Tr(O)=0} and from the angle-embedding variance upper bound (scaled by $1/1000$ to fit the plot) given in \Cref{eq: angle embedding variance upper bound}. We note that the upper bound on the angle-embedded standard deviation is included only for comparison of the general trend, not for the exact magnitudes. (b) Same as (a), but with the $L_2$ norm $\|\bm{\omega}\|_2$ on the horizontal axis.}
    \label{fig: 2q depol channel coefficients analysis part 2}
\end{figure*}

In \Cref{fig: 2q depol channel coefficients analysis part 1}, we demonstrate the key diagnostics of the VQC training process, including the MSE of both the stochastic training and testing datasets, its gradient norms squared, the corresponding average deviations $\mathbb{E}_{c_{\bm{\omega}}}|c_{\bm{\omega}}-c_{\bm{\omega}'}|$, as well as proportional average deviations $\mathbb{E}_{c_{\bm{\omega}}}\left|\frac{c_{\bm{\omega}}-c_{\bm{\omega}'}}{c_{\bm{\omega}}}\right|$ for both zero and non-zero target Fourier coefficients values. From \Cref{fig:2q_depol_MSE}, it is clear that the higher the depolarizing strength, the more the MSE and its gradients are suppressed, agreeing with the prediction about noise-induced barren plateaus \cite{noise_induced_barren_plateaus}. Nevertheless, the model gradually converges to the target function,  evidenced by the fact that the average deviations of the Fourier coefficients from their target values decay as training proceeds, shown in \Cref{fig:2q_depol_nonzero_coeff_average_deviations,fig:2q_depol_zero_coeff_average_deviations}. More interestingly, we observe that regardless of a zero or non-zero target coefficient, the average deviations appear to be smaller when the depolarizing strength is actually greater. This is because the target coefficient values are sufficiently small For most target coefficients, suppression of the Fourier coefficients due to depolarizing noise effectively pushes them towards the `correct' directions. However, we still recognize that, despite clear signs of learning in the first several iterations, the average deviation quickly asymptotes, even when $p=0$, indicating that the model itself still has limitations in expressivity. 
\begin{figure*}[t!]
    \centering
        \begin{subfigure}{.49\textwidth}
            \centering
            \includegraphics[width=\columnwidth]{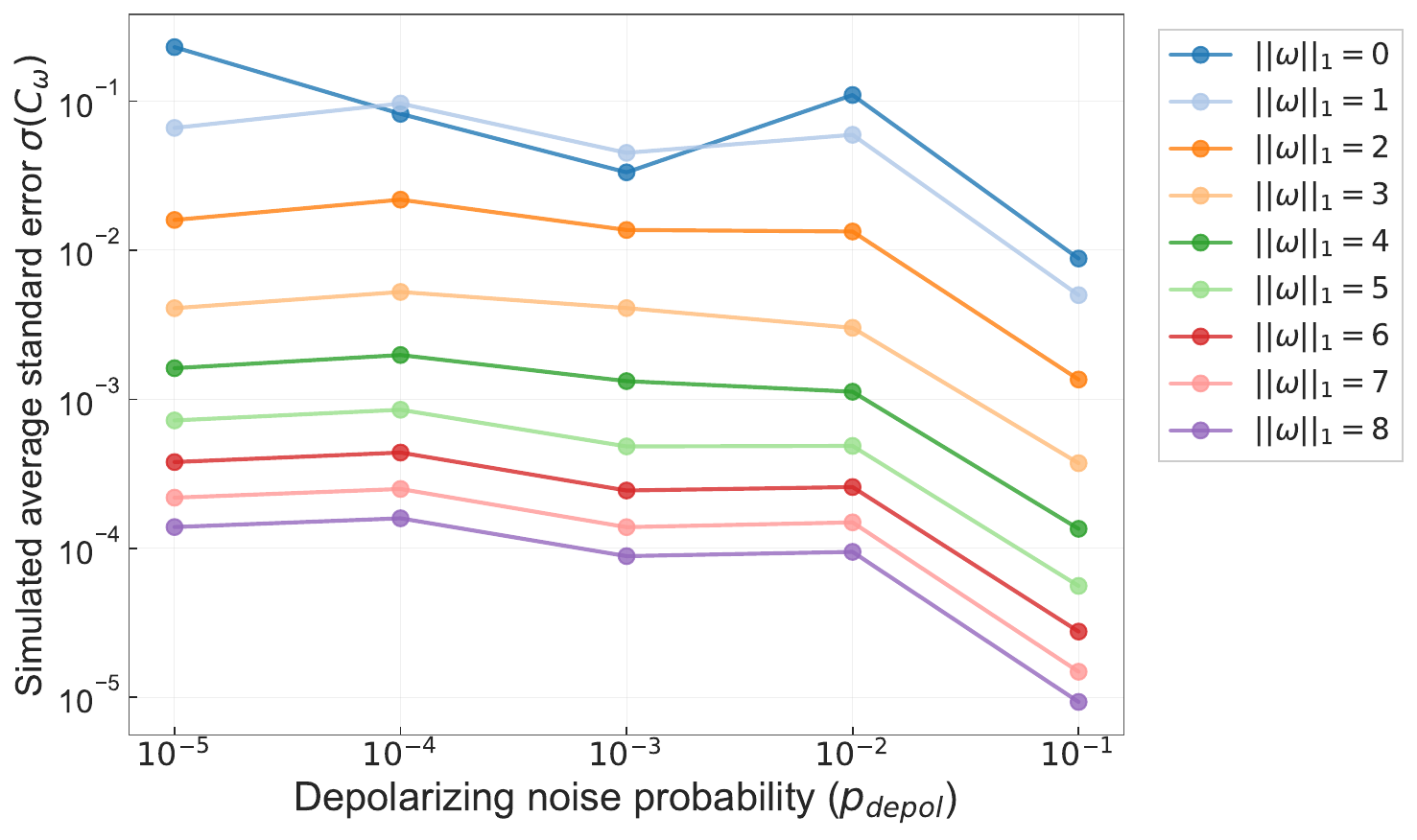}
            \phantomcaption
            \label{fig:2q_depol_coeff_std_vs_pdepol}
        \end{subfigure}
        \begin{subfigure}{.49\textwidth}
            \centering
            \includegraphics[width=\columnwidth]{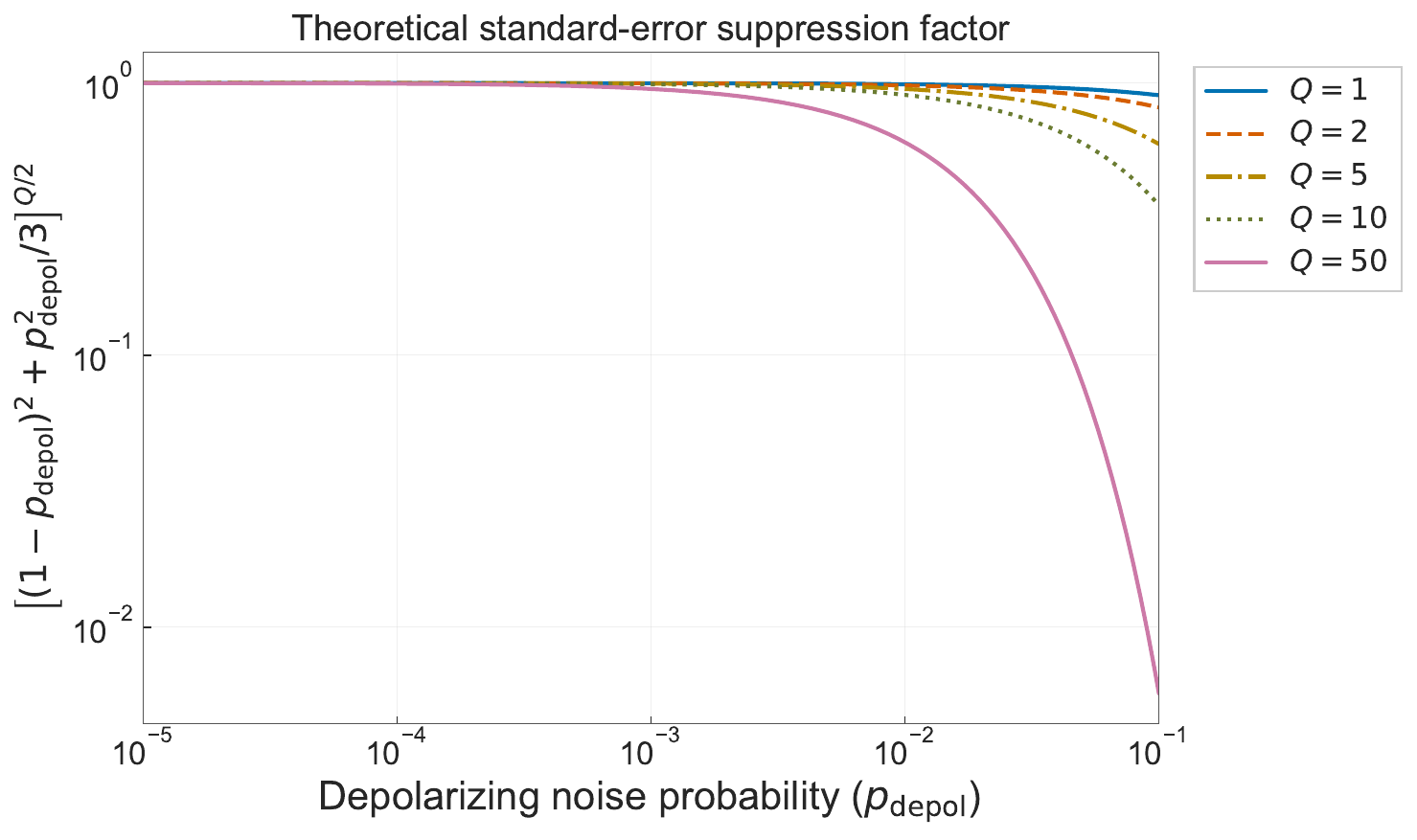}
            \phantomcaption
            \label{fig:2q_depol_theoretical_std_vs_pdepol}
        \end{subfigure}
        \caption{(a) The empirical average of the Fourier coefficient standard errors calculated the same in \Cref{fig: 2q depol channel coefficients analysis part 2}, but plotted against the depolarizing noise probability, $p_{\text{depol}}$. Both axes are in the logarithmic scale. The colours denote the frequency $L_2$ norms. (b) The theoretical standard error supression factor $(\sum\limits_kp_k^2)^Q=\left(\frac{p^2}{3}+(1-p)^2\right)^Q$ plotted against $p$, where $p=p_{\text{depol}}$ is the single-qubit depolarizing channel strength same across the entire circuit.}
\end{figure*}

To address this, we present the scaling of the standard deviation as a function of the frequency norms derived both from analytical expressions in \Cref{eq: amplitude embedding noiseless variance Tr(O)=0,eq: amplitude embedding noisy variance Tr(O)=0} and from numerical simulations, shown in \Cref{fig: 2q depol channel coefficients analysis part 2}. Compared with the analytical scaling of the amplitude embedding standard deviations, the simulation curves show good agreement with the analytical predictions, in both overall magnitude and the detailed decaying trend. In particular, for \Cref{fig:2q_depol_coeff_std_frequency_L2}, the curves from the simulation follow almost the same zig-zag pattern as the analytical outcomes, regardless of the presence of noise, despite differences in magnitudes. This alignment with the prediction also demonstrates the credibility and versatility of the premises made: Even if the analytical result is averaged over the full parameter space, whereas the simulation only averages over the parameters visited along the training trajectory, the agreement in the plot implies the impact on the variance is very limited. Moreover, it is observed that standard errors for the non-zero target coefficients consistently underperform the global average of all coefficients. Since the vast majority of the target coefficients are zero, this phenomenon is likely due to non-zero target Fourier coefficients having more contribution to the loss function, which converges faster and fluctuates less in later iterations. This is somewhat reflected in \Cref{fig:2q_depol_nonzero_coeff_average_deviations} that the most deviations from non-zero targets stop varying substantially after 25 iterations, whereas it is not the case for zero targets in \Cref{fig:2q_depol_zero_coeff_average_deviations}. This also implies the Fourier coefficients with non-zero target values have potentially shorter trajectories in the parameter spacetime. In \Cref{appendix: supplementary plots for depol channel}, we include additional plots detailing the evolution of the Fourier coefficients in this simulation, where frequencies with different norms are plotted with different colours to emphasize how the variance scales with frequency. On the other hand, as expected, the curves from amplitude-embedding simulation results match poorly with the analytical upper bound (scaled by a factor of 1/1000 to fit into the plot) for the angle embedding standard deviation. Although the numerical evidence suggests that the standard deviations for both angle and amplitude embedding decay approximately exponentially with the frequency norm, the differences remain clear. Specifically, on the log-scale plot, the overall trend of the amplitude embedding standard deviation follows nearly a straight line, especially when $L_1$ norm is used for frequency, which decays faster than the case of angle-embedding given the same frequency norms. However, as the frequency increases, the decay of the angle embedding standard deviation upper bound starts to accelerate, as evidenced by the downward curved line on the plot, indicating that the model would progressively become less expressive at high frequencies. Second, a power-law upper bound for the amplitude-embedded model indicates that the possibility of having a different scaling pattern at even higher frequencies cannot be negated. 

Putting it all together, we show that, compared to angle-embedding, a VQC model with non-negative amplitude embedding exhibits less expressivity for low-frequency data, while it becomes relatively more expressive at high frequencies. We remark again that the curves for the angle-embedding are derived from an upper bound, meaning that this is already a worst-case comparison for amplitude embedding in terms of expressivity, not even to mention that the spectrum cuts off for angle-embedding VQCs when $|\omega_k|=L$. Moreover, in \Cref{fig:2q_depol_coeff_std_vs_pdepol}, we plot the average standard errors of Fourier coefficients versus the depolarizing noise strength. In the log-log scale, the plot exhibits a non-linear, accelerating decay, which agrees with the prediction from \Cref{eq: amplitude embedding noisy variance Tr(O)=0} that the scaling of the variance (hence standard deviation) with respect to the depolarizing strength $p$ is 
$\mathbb{V}\text{ar}_{\bm{\theta}}(\tilde{c}_{\bm{\omega}}(\bm{\theta}))=O\left(\left((1-p)^2+\frac{p^2}{3}\right)^{Q}\right)=O\left(\text{poly}(p)^Q\right)$ which its theoretical scaling is plotted in \Cref{fig:2q_depol_theoretical_std_vs_pdepol}. We also conducted a 3-qubit simulation with depolarizing noise with its plots summarized in the supplementary material. In fact, it also shows a uniform, proportional suppression of the Fourier coefficient standard deviations regardless of the frequencies as the strength of the noise increases, the behaviour same as what is observed in the two-qubit case. Beyond depolarizing noise, we also examined the performance of the model and the scaling of the Fourier coefficients under various other noise channels described in the supplementary material, such as amplitude damping, dephasing, and $ZZ$-coupled idling noise. Interestingly, despite different levels of suppression of the Fourier coefficients, the overall behaviour of the scaling with respect to the frequency norms remain very similar to \Cref{fig: 2q depol channel coefficients analysis part 2}. Hence, it indicates that regardless of noise channel structure, the scaling of the Fourier coefficient variance remains invariant. This implies the high stability in the expressivity of amplitude embedding VQCs.

\subsection{Truncated Fourier series with non-integer frequencies}\label{subsec: non periodic functions}
Technically, if the encoding Hamiltonian is chosen appropriately, it is in general possible to train a model with a non-integer frequency spectrum in angle embedding. However, it is not possible to decompose the function into arbitrary frequencies within a given VQC layout. As a result, in addition to the discretized, integer-valued frequencies, we also demonstrate that Fourier coefficients can be modulated for arbitrary non-integer-valued frequency spectrum. We note that any theorems derived so far should not be assumed in this subsection, as premises about integer-valued frequencies are made in the derivations that affect the final outcome. Particularly the $\delta_{\bm{\omega},0}$ condition appears in the Fourier coefficient moments which assumes orthogonality of integer-valued angular frequencies in a finite-spaced domain. Hence, the simulations done in this subsection are only for the purpose of demonstrating the existence of non-zero amplitude Fourier coefficients with non-integer frequencies. In \Cref{fig: non int freqs plots}, as well as the additional plots \Cref{fig: non int freqs plots continue} in \Cref{appendix: supplementary plots for non-int freqs}, we examined the performance of the same 2-qubit VQC model used in \Cref{subsec: simulations}, but fitting with a function whose decomposed frequencies are non-integer. For the target Fourier coefficients that are zero, the coefficients converge smoothly towards zero, as shown in \Cref{fig:2q_non_int_freqs_depol_zero_coeff_average_deviations}. However, from \Cref{fig:2q_non_int_freqs_depol_nonzero_coeff_average_deviations}, the deviation from the non-zero target Fourier coefficients is almost unchanged during training, meaning that it stays near 0 all the time irrespective of highly non-zero target values. Despite being highly inexpressive, this does not disprove that the non-integer frequency Fourier coefficients of amplitude embedding VQCs are not modulating. As illustrated in \Cref{fig:2q_non_int_freqs_depol_nonzero_coeff_average_deviations}, by zooming in on the scale, the average deviation from non-zero target Fourier coefficients is still actively being modulated as the parameters iterate. Furthermore, as shown by the decaying MSE in \Cref{fig:2q_non_int_freqs_depol_MSE}, both from the test and training batches, the model exhibits clear behaviour of convergence. Hence, just like for integer frequencies, although the expressivity for non-integer frequencies is still poor at high frequencies, it is indeed non-zero and can be modulated by the parameters. 
\begin{figure*}[t!]
    \centering
    \begin{subfigure}[b]{.47\textwidth}
        \centering
        \includegraphics[width=1\columnwidth]{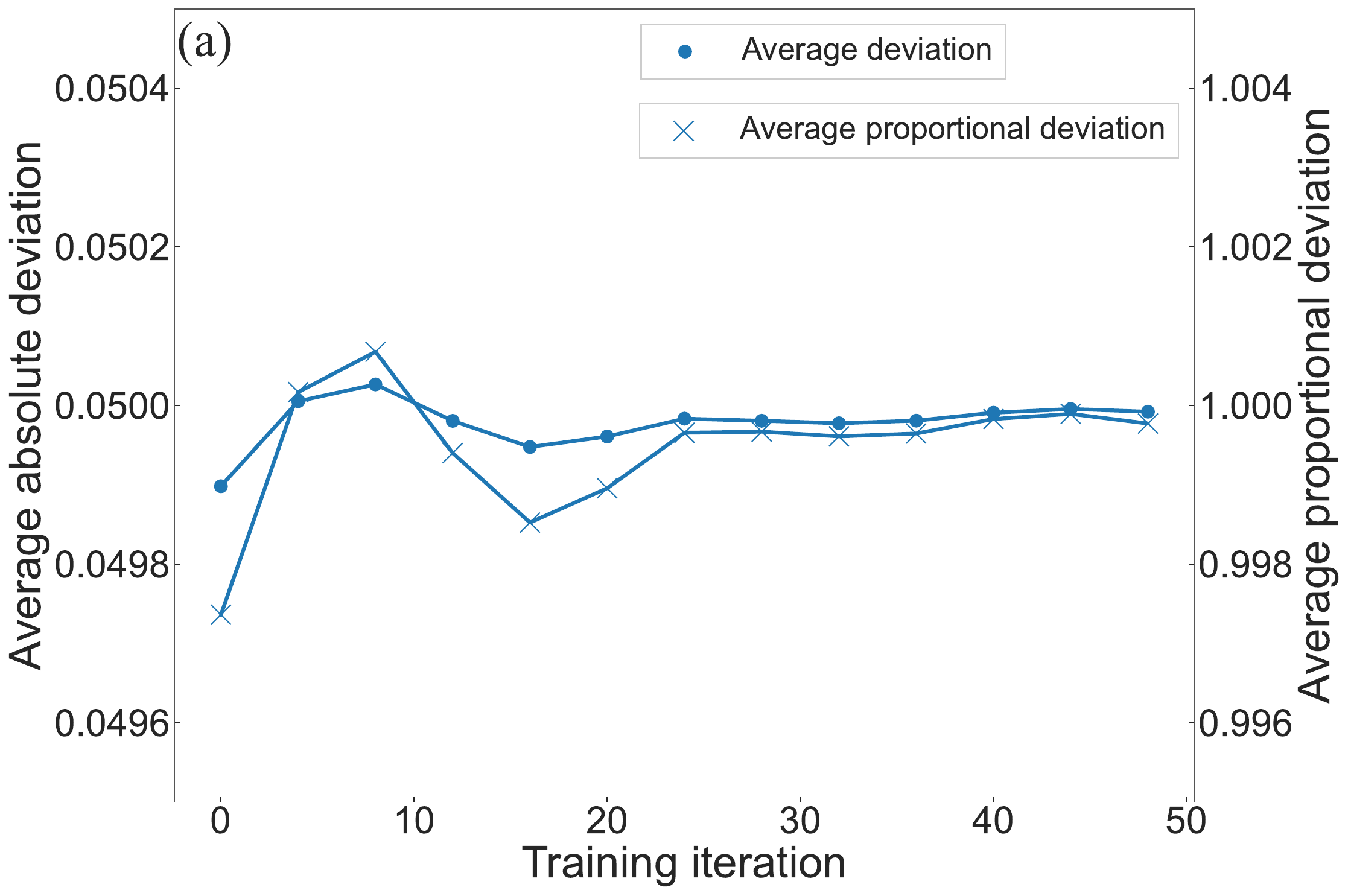}
        \phantomcaption
        \label{fig:2q_non_int_freqs_depol_nonzero_coeff_average_deviations}
    \end{subfigure}
    \begin{subfigure}[b]{.49\textwidth}
        \centering
        \includegraphics[width=1\columnwidth]{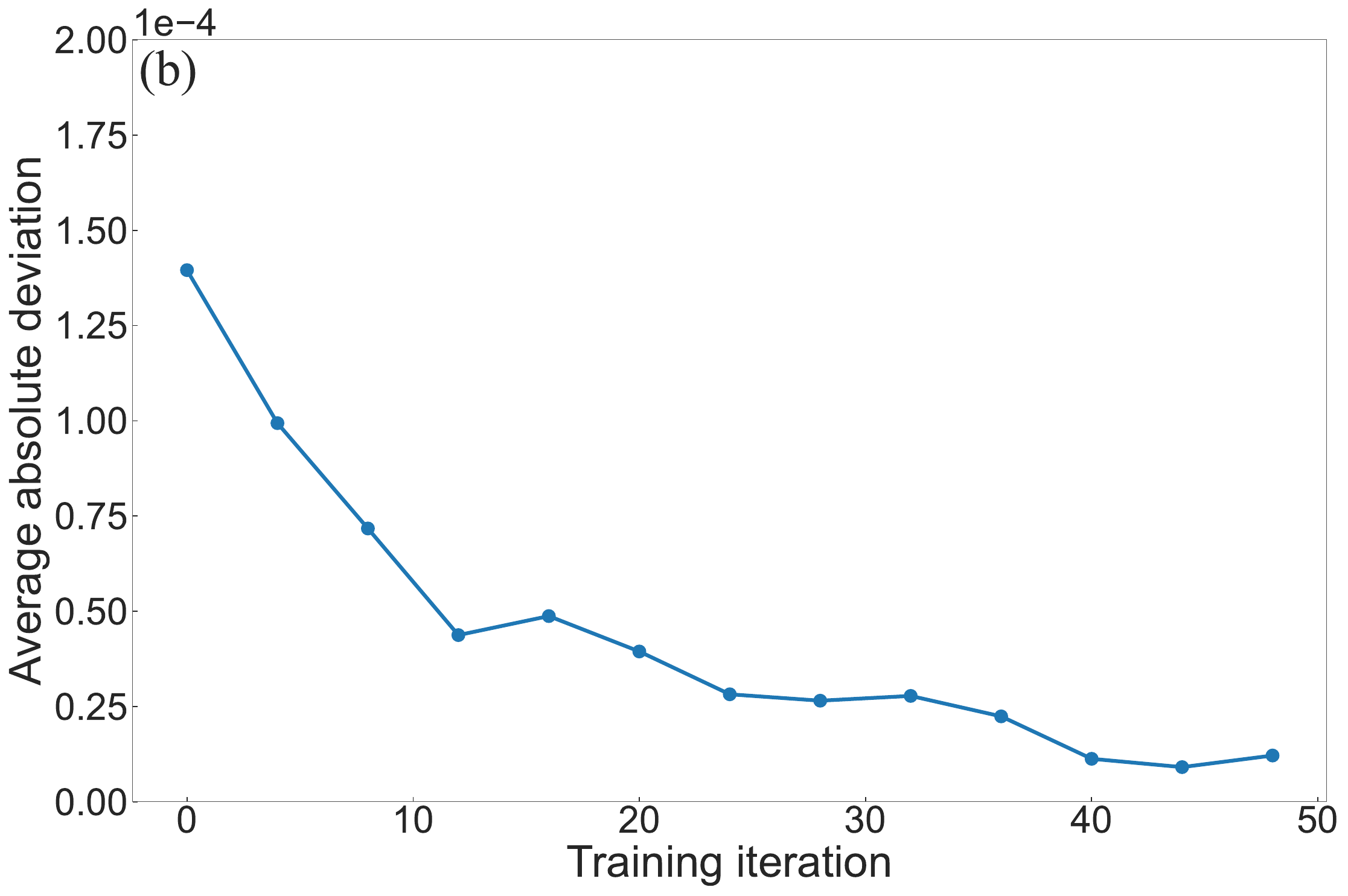}
        \phantomcaption
        \label{fig:2q_non_int_freqs_depol_zero_coeff_average_deviations}
    \end{subfigure}
    \caption{(a) The average deviation, and average of the proportional deviation from non-zero target Fourier coefficients for the model same as in \Cref{subsec: simulations}, but trained with respect to a function, where the $y$ axis is zoomed in to articulate the changes due to modulating parameters. (b) is the deviation of the Fourier coefficients from zero target values.}
    \label{fig: non int freqs plots}
\end{figure*}
\section{Discussion}\label{sec: Discussion}
In summary, while Fourier analysis of VQC models has been confined to angle-embedding protocols in other studies, this paper presents a detailed analysis of the expressivity of VQC models under amplitude embedding. We first argue that the range of values each feature can take has a huge impact on the model performance, even within a single encoding scheme such as amplitude embedding. Specifically, we conduct careful analytical derivations with two common examples, namely \textit{symmetric} and \textit{non-negative} domains. We find that the Fourier coefficient at zero frequency is not identically zero when the input features are encoded onto a non-negative domain, whereas it vanishes identically under a symmetric domain. This marked difference in the expressivity of the zeroth Fourier coefficient is further corroborated by the stark contrast in model performance: The model using non-negative encoding has its zeroth coefficient converge to the target value, while the coefficient in the model with symmetric encoding remains pinned at zero. This phenomenon of a vanishing coefficient at zero frequency provides a deeper understanding of amplitude-embedded VQC models on what encoding domains should be chosen if one wants to maximize the expressivity and hence trainability of the model.

Having established the utility of non-negative encoding for expressivity, we proceed to derive the mean and variance of the Fourier coefficients of the non-negative amplitude-embedding VQC model at discretized, integer-valued frequencies. Under idealized assumptions that the ensemble of unitaries generated by the parameter space forms a 2-design and that a traceless Pauli observable is employed, we find that the variance of the Fourier coefficients is upper bounded by an inverse power-law function in the magnitude of the frequency, with numerical evidence that it decays exponentially within the sampled range of the frequencies. In addition, when a noise channel with probabilities $\{p_k\}$ for unitary Kraus operators is incorporated, an additional factor less than one, is multiplied by the variance of the Fourier coefficient, leading to a general suppression of the coefficient variance and expressivity of the coefficient. This behaviour is also reflected in our simulations across different noise strengths of a depolarizing channel, where the variance, and consequently the standard error of the sampled Fourier coefficients, decays exponentially as the frequency increases, and polynomially as the noise strength $p$ increases. The agreement between simulation and theory reinforces each other and consolidates the validity of the premises made before derivations, including assumptions about the 2-design, and independence between distinct perturbed unitaries, even for circuits of only 2 qubits and 30 variational layers. Furthermore, despite both being exponentially decaying with respect to frequency magnitudes at high levels, the subtle difference in the decaying pattern reveals that amplitude-embedded models are likely to have their expressivity decay faster compared to the vanilla angle-embedded models at small frequencies. However, the situation would be reversed at high frequencies, as the decay of angle-embedded coefficient variance accelerates. As a result, this provides a better understanding of how different encoding strategies have expressivity, and hence trainability emphasis on different frequencies. Yet, the overall decay at an exponential rate still governs at a high level. Moreover, we find that the Fourier coefficients with the amplitude-embedding model can still be modulated even at non-integer frequencies, which is forbidden in the case of angle-embedding. Thereby, this highlights its utility in fitting functions with arbitrary frequencies.

\begin{acknowledgements}
The authors greatly acknowledge the computational resources provided by the University of Melbourne's Research Computing Services. H.K. was supported by the Australian Government Research Training Program Scholarship.
\end{acknowledgements}

\bigskip
\noindent
\textbf{Data Availability Statement}\par
\noindent
The source code and raw data used in the current study will be openly available at \cite{repo} with the final accepted version of the paper.

\bigskip
\bigskip
\bigskip
\bigskip
\noindent
\textbf{Supplementary Data}\par
\noindent
Supplementary material contains additional simulation results
for amplitude damping, dephasing, and $ZZ$-coupled idling noise.

\bigskip
\noindent
\textbf{Author contributions}\par
\noindent
HK conceived the original idea. MU and MS supervised the project. HK conducted all analytical derivations, developed the QML simulation framework, and carried out all experiments. HK, MU, and MS analyzed the data. HK wrote the manuscript with input from MU and MS. An AI model was used for English grammar proofreading only.

\onecolumngrid
\appendix
\section{Proofs}
\subsection{Proof of \Cref{corollary: multi dim freq decaying redundancies}}\label{appendix: angle embedding proofs}
\begin{proof}
By definition, $|R(\bm{\omega})|$ counts the number of index pairs $(\bm{J}, \bm{J}')$ that produce the target frequency vector $\bm{\omega}$.  Since each component $\omega_k$ is determined independently by the eigenvalues of the encoding Hamiltonian acting on qubit~$k$, we have
\begin{equation}
\begin{aligned}
|R(\bm{\omega})|
&= \sum_{\bm{J}, \bm{J}'} \prod_{k=1}^{d} \delta\bigl(\omega_k - (\Lambda^{(k)}_{\bm{J}} - \Lambda^{(k)}_{\bm{J}'})\bigr) \\[4pt]
&\le \prod_{k=1}^{d} \sum_{\bm{J}, \bm{J}'} \delta\bigl(\omega_k - (\Lambda^{(k)}_{\bm{J}} - \Lambda^{(k)}_{\bm{J}'})\bigr) \\[4pt]
&= \prod_{k=1}^{d} \binom{2L}{L - |\omega_k|} \\[4pt]
&= O\bigl(\exp(-\|\bm{\omega}\|_1)\bigr),
\end{aligned}
\end{equation}
where the inequality follows from the elementary set-theoretic bound $|\bigcap_k A_k| \le \prod_k |A_k|$ with $A_k = \{(\bm{J},\bm{J}') : \omega_k = \Lambda^{(k)}_{\bm{J}} - \Lambda^{(k)}_{\bm{J}'}\}$, and the second equality invokes the one-dimensional redundancy formula from~\Cref{eq: angle embedding variance decay} applied to each qubit individually (each qubit contributes $L$ layers of data reuploading, giving the binomial coefficient $\binom{2L}{L-|\omega_k|}$).

For the multi-dimensional frequency spectrum, following the notation of Ref.~\cite{Fourier_analysis_expressivity_as_redundancies} (Appendices~A.1 and~A.1.1), the Fourier coefficient can be expanded as
\begin{equation}
c_{\bm{\omega}} = \sum_{\substack{\bm{J},\bm{J}' \in R(\bm{\omega}),\\[2pt] k,k'}} 
W_{j'_1 0}^{(2)*} W_{j'_2 j'_1}^{(1)*} \cdots W_{k' j'_{L}}^{(L+1)*} \, O_{k'k} \, W_{k j_{L}}^{(L+1)} \cdots W_{j_2 j_{1}}^{(2)} W_{j_1 0}^{(1)}.
\end{equation}
Applying the identical derivation as in Ref.~\cite{Fourier_analysis_expressivity_as_redundancies}---which relies solely on the 2-design property of the trainable layers and the pairwise decorrelation of distinct index pairs---yields
\begin{equation}
\operatorname{Var}_{\bm{\theta}}(c_{\bm{\omega}}(\bm{\theta}))
\le O\!\left(\alpha\,\frac{|\widetilde{R}(\bm{\omega})|}{2^n}\right)
\le O\bigl(\exp(-\|\bm{\omega}\|_1)\bigr),
\end{equation}
where $|\widetilde{R}(\bm{\omega})| = |R(\bm{\omega})| / 2^{2n}$ is the normalized redundancy.  This completes the proof.
\end{proof}

\subsection{Proof of \Cref{theorem: zero coefficient}}\label{appendix: zero coefficient proof}
\begin{proof}
At zero frequency, the Fourier coefficient given in \Cref{eq: amplitude embedding Fourier coefficients} reduces to
\begin{equation}\label{eq: zero coefficient}
    c_{\bm{0}}(\bm{\theta})=\frac{1}{(2\pi)^d}\int_{D}d^{(d)}V\,\frac{x_ix_{i'}\,W(\bm{\theta})^{\dagger}_{i'j'}O_{j'j}W(\bm{\theta})_{ji}}{\|\bm{x}\|_2^2},
\end{equation}
where $D=[-R/2,R/2]^{d}$ (symmetric domain) or $D=[0,R]^d$ (non-negative domain). In either case, we split the contraction over $i,i'$ into the diagonal ($i=i'$) and off-diagonal ($i\neq i'$) contributions and treat them separately.

For $i=i'$, we have
\begin{equation}\label{eq: i=i' proof equation 1}
\begin{aligned}
&\frac{1}{(2\pi)^d}\sum_{i=i'\in[d]}\int_{D}d^{(d)}V\,\frac{x_ix_{i'}\,W(\bm{\theta})^{\dagger}_{i'j'}O_{j'j}W(\bm{\theta})_{ji}}{\|\bm{x}\|_2^2}\\[4pt]
    =&\frac{1}{(2\pi)^d}\int_{D}d^{(d)}V\sum_{i\in[d]}\frac{x_i^2\,W(\bm{\theta})^{\dagger}_{ij'}O_{j'j}W(\bm{\theta})_{ji}}{\|\bm{x}\|_2^2}.
\end{aligned}
\end{equation}
Define $\mathcal{X}^2(\bm{x})$ as the diagonal matrix with entries $\mathcal{X}^2(\bm{x})_i = x_i^2/\|\bm{x}\|_2^2$. Then
\begin{equation}
\begin{aligned}
\sum_{i\in[d]}\frac{x_i^2\,W(\bm{\theta})^{\dagger}_{ij'}O_{j'j}W(\bm{\theta})_{ji}}{\|\bm{x}\|_2^2}
&= \operatorname{Tr}\bigl(\mathcal{X}^2(\bm{x})\,W(\bm{\theta})^{\dagger}O\,W(\bm{\theta})\bigr) \\[4pt]
&= \sum_{i} \mathcal{X}^2(\bm{x})_i \,\operatorname{diag}(\tilde{O})_{i},
\end{aligned}
\end{equation}
where $\tilde{O}=W(\bm{\theta})^{\dagger}O\,W(\bm{\theta})$, and the second equality uses the fact that $\mathcal{X}^2(\bm{x})$ is diagonal.

For $i\neq i'$, we obtain
\begin{equation}\label{eq: i not i' proof equation 1}
\begin{aligned}
&\frac{1}{(2\pi)^d}\sum_{i\neq i'\in[d]}\int_{D}d^{(d)}V\,\frac{x_ix_{i'}\,W(\bm{\theta})^{\dagger}_{i'j'}O_{j'j}W(\bm{\theta})_{ji}}{\|\bm{x}\|_2^2}\\[4pt]
    =&\frac{1}{(2\pi)^d}\int_{D}d^{(d)}V\sum_{i>i'\in[d]}\frac{x_ix_{i'}}{\|\bm{x}\|_2^2}\Bigl(W(\bm{\theta})^{\dagger}_{i'j'}O_{j'j}W(\bm{\theta})_{ji} + \mathrm{c.c.}\Bigr)\\[4pt]
    =&\frac{1}{(2\pi)^d}\int_{D}d^{(d)}V\sum_{i>i'\in[d]}\frac{x_ix_{i'}}{\|\bm{x}\|_2^2}\bigl(\tilde{O}_{i'i} + \mathrm{c.c.}\bigr),
\end{aligned}
\end{equation}
where in the second line we paired the terms with indices $i\leftrightarrow i'$ into a single sum over $i>i'$, noting that the two summands are mutual complex conjugates. Since $\tilde{O}_{ii'}=\tilde{O}_{i'i}^{*}$, \cref{eq: i not i' proof equation 1} simplifies to
\begin{equation}\label{eq: i not i' proof equation 2}
    \frac{1}{(2\pi)^d}\int_{D}d^{(d)}V\sum_{i>i'\in[d]}\frac{2x_ix_{i'}}{\|\bm{x}\|_2^2}\,\operatorname{Re}\bigl(\tilde{O}(\bm{\theta})_{i'i}\bigr).
\end{equation}

\medskip
\noindent\textbf{Symmetric domain.}
The observable $O$ is traceless, and the trace is invariant under unitary conjugation, so that $\operatorname{Tr}\bigl(W(\bm{\theta})^{\dagger}O\,W(\bm{\theta})\bigr) = \operatorname{Tr}(O) = 0$, hence
\begin{equation}\label{eq: traceless O condition for proofs}
    \sum_i \operatorname{diag}(\tilde{O})_i = 0.
\end{equation}
Consequently, the diagonal contribution~\eqref{eq: i=i' proof equation 1} evaluates to
\begin{equation}
\begin{aligned}
&\frac{1}{(2\pi)^d}\int_{D}d^{(d)}V\sum_{i\in[d]}\mathcal{X}^2(\bm{x})_i \,\operatorname{diag}(\tilde{O})_i \\[4pt]
&\qquad = \frac{1}{(2\pi)^d}\sum_{i\in[d]}\operatorname{diag}(\tilde{O})_i \int_{D}d^{(d)}V\,\frac{x_i^2}{\|\bm{x}\|_2^2} \\[4pt]
&\qquad = \frac{1}{(2\pi)^d}\sum_{i\in[d]}\operatorname{diag}(\tilde{O})_i\; I_d \\[4pt]
&\qquad = 0.
\end{aligned}
\end{equation}
In the second line we noted that the integral is independent of the index label $i$ and defined $I_d = \int_{D}d^{(d)}V\,x_i^2/\|\bm{x}\|_2^2$; the third line then invokes \cref{eq: traceless O condition for proofs}.

For the off-diagonal part, observe that on the symmetric domain $D=[-R/2,R/2]^{d}$, every $i,i'$ term integrand $\frac{x_ix_{i'}}{\|\bm{x}\|_2^2}$ is odd with respect to either axis ($x_i$ or $x_{i'}$) in \cref{eq: i not i' proof equation 2}, the off-diagonal contribution~\eqref{eq: i not i' proof equation 2} vanishes identically. Combining the diagonal and off-diagonal results, we conclude that $c_{\bm{0}}(\bm{\theta}) = 0$ strictly for the symmetric domain.

\medskip
\noindent\textbf{Non-negative domain.}
The diagonal ($i=i'$) contribution is again zero, following exactly the same argument as for the symmetric case, since the integral $I_d$ remains invariant under permutations of the coordinate labels. For the off-diagonal ($i\neq i'$) contribution, the integrand in~\eqref{eq: i not i' proof equation 2} is everywhere non-negative ($2x_ix_{i'}/\|\bm{x}\|_2^2 \ge 0$ on $D=[0,R]^d$), hence
\begin{equation}
    \frac{1}{(2\pi)^d}\sum_{i>i'\in[d]}\int_{D}d^{(d)}V\,\frac{2x_ix_{i'}}{\|\bm{x}\|_2^2} = \text{constant} > 0.
\end{equation}
Adding the diagonal and off-diagonal parts, we obtain $c_{\bm{0}}(\bm{\theta}) = \text{constant} \times \sum_{i>i'} \operatorname{Re}\bigl(\tilde{O}(\bm{\theta})_{i'i}\bigr) \neq 0$, since the sum of the real parts of the off-diagonal entries of a Hermitian matrix can take any real value.
\end{proof}

\subsection{Proof of \Cref{theorem: noiseless mean and variance}}\label{appendix: noiseless mean and variance proof}
\begin{proof}
Without loss of generality, and for ease of migration to the noisy case proof, we can first rewrite \Cref{eq: f(x)} as
\begin{equation}
\begin{aligned}
    f_{\bm{\theta}}(\bm{x})&=\operatorname{Tr}\left(OU(\bm{x},\bm{\theta})^{\dagger}\ket{0}\bra{0}U(\bm{x},\bm{\theta})\right)\\
    &=\operatorname{Tr}\left(OW(\bm{\theta})^{\dagger}\rho_{\bm{x}}W(\bm{\theta})\right),
\end{aligned}
\end{equation}
where $\rho_{\bm{x}}=\ket{\bm{x}}\bra{\bm{x}}=\sum_{ij=1}^{d}\frac{x_ix_j\ket{i}\bra{j}}{\|\bm{x}\|_2^2}$ is the initial state density matrix after feature encoding. Adopting Feynman's path integral that expands $\rho_{\bm{x}}$ into standard basis, it can be further expanded as
\begin{equation}
        \rho_{\bm{x}}=\sum\limits_{i,j=1}^{d}\operatorname{Tr}(\rho_{\bm{x}}e_{ij})e_{ij},
\end{equation}
where $e_{ij}=\ket{i}\bra{j}$. Hence,
\begin{equation}
    \begin{aligned}
        &f_{\bm{\theta}}(\bm{x})=\sum\limits_{i,j=1}^{d}\operatorname{Tr}\left(OW(\bm{\theta})^{\dagger}e_{ij}W(\bm{\theta})\right)\operatorname{Tr}(\rho_{\bm{x}}e_{ij})\\
        &=\sum\limits_{i,j,i',j'=1}^{d}\operatorname{Tr}\left(Oe_{i'j'}\right)\operatorname{Tr}\left(e_{i'j'}W(\bm{\theta})^{\dagger}e_{ij}W(\bm{\theta})\right)\operatorname{Tr}(\rho_{\bm{x}}e_{ij}).
    \end{aligned}
\end{equation}
Therefore, \Cref{eq: amplitude embedding Fourier coefficients}, and thus the corresponding first moment of the Fourier coefficient, becomes
\begin{equation}\label{eq: Fourier coefficients first moment}
\begin{aligned}
    c_{\bm{\omega}}(\bm{\theta})&=\frac{1}{(2\pi)^d}\sum\limits_{i,j,i',j'}\int d^{(d)}Ve^{-i\bm{\omega}\cdot\bm{x}}\operatorname{Tr}\left(Oe_{i'j'}\right)\operatorname{Tr}\left(e_{i'j'}W(\bm{\theta})^{\dagger}e_{ij}W(\bm{\theta})\right)\operatorname{Tr}(\rho_{\bm{x}}e_{ij})\\
    \Rightarrow  \mathop{\mathbb{E}}\limits_{\bm{\theta}\sim\Theta}\left(c_{\bm{\omega}}(\bm{\theta})\right)&=\frac{1}{(2\pi)^d}\sum\limits_{i,j,i',j'}\int d^{(d)}Ve^{-i\bm{\omega}\cdot\bm{x}}\operatorname{Tr}\left(Oe_{i'j'}\right)\operatorname{Tr}\left(e_{i'j'}\mathop{\mathbb{E}}\limits_{\bm{\theta}\sim\Theta}\left[W(\bm{\theta})^{\dagger}e_{ij}W(\bm{\theta})\right]\right)\operatorname{Tr}(\rho_{\bm{x}}e_{ij})\\
    &=\frac{1}{(2\pi)^d}\sum\limits_{i,j,i',j'}\int d^{(d)}Ve^{-i\bm{\omega}\cdot\bm{x}}\operatorname{Tr}\left(Oe_{i'j'}\right)\operatorname{Tr}\left(e_{i'j'}\mathop{\mathbb{E}}\limits_{U\sim\operatorname{U}(d)}\left[U^{\dagger}e_{ij}U\right]\right)\operatorname{Tr}(\rho_{\bm{x}}e_{ij})\\
    &=\frac{1}{d(2\pi)^d}\sum\limits_{i,j,i',j'}\int d^{(d)}Ve^{-i\bm{\omega}\cdot\bm{x}}\operatorname{Tr}\left(Oe_{i'j'}\right)\operatorname{Tr}\left(e_{i'j'}\right)\operatorname{Tr}\left(e_{ij}\right)\operatorname{Tr}\left(\rho_{\bm{x}}e_{ij}\right)\\
    &=\frac{1}{d(2\pi)^d}\sum\limits_{i,j,i',j'}\int d^{(d)}Ve^{-i\bm{\omega}\cdot\bm{x}}\operatorname{Tr}\left(Oe_{i'j'}\right)\delta_{i'j'}\delta_{ij}\operatorname{Tr}\left(\rho_{\bm{x}}e_{ij}\right)\\
     &=\frac{1}{d(2\pi)^d}\int d^{(d)}Ve^{-i\bm{\omega}\cdot\bm{x}}\operatorname{Tr}\left(O\right)\operatorname{Tr}\left(\rho_{\bm{x}}\right)\\
     &=\frac{\operatorname{Tr}\left(O\right)\delta_{\bm{\omega}}^0}{d}
\end{aligned}
\end{equation}
In the third line we apply the assumption that  $\text{Im}(\mathcal{W})$, $\mathcal{W}:\Theta\rightarrow \mathcal{U}$ on the unitary group $\mathcal{U}\subset\text{U}(d)$ forms a 2-design, and hence must also forms a 1-design. In the fourth line, we use the first-moment Weingarten formula for Haar integration over the unitary group \cite{Haar_measure_tutorial}. On line five and six, we use the identity of trace on diagonal and off-diagonal matrices, and contract the indices accordingly. On the last line, we use the fact that a density matrix has unit trace, as well as the Dirac-$\delta$ condition.

Similarly, to obtain the variance of the Fourier coefficients, we first compute the second moment of the coefficients.
\begin{equation}\label{eq: Fourier coefficient second moment phase 1}
\begin{aligned}
    &|c_{\bm{\omega}}(\bm{\theta})|^2=\left|\frac{1}{(2\pi)^d}\sum\limits_{i,j,i',j'}\int d^{(d)}Ve^{-i\bm{\omega}\cdot\bm{x}}\operatorname{tr}\left(Oe_{i'j'}\right)\operatorname{tr}\left(e_{i'j'}W(\bm{\theta})^{\dagger}e_{ij}W(\bm{\theta})\right)\operatorname{tr}(\rho_{\bm{x}}e_{ij})\right|^2\\
    &=\sum\limits_{\substack{i,i',j,j'\\k,k',l,l'}}\iint d^{(2d)}VV'\frac{e^{-i\bm{\omega}\cdot(\bm{x}-\bm{x}')}}{{(2\pi)^{2d}}}\operatorname{tr}\left(O\otimes O^* e_{i'j'}\otimes e_{k'l'}\right)\operatorname{tr}\left(e_{i'j'}W(\bm{\theta})^{\dagger}e_{ij}W(\bm{\theta})\right)\\
    &\text{\hspace{6cm}}\times\operatorname{tr}\left(e_{k'l'}W(\bm{\theta})^{\dagger}e_{kl}W(\bm{\theta})\right)^*\operatorname{tr}\left(\rho_{\bm{x}}\otimes \rho_{\bm{x}'}^*e_{ij}\otimes e_{kl}\right)\\
    &=\sum\limits_{\substack{i,i',j,j'\\k,k',l,l'}}\iint d^{(2d)}VV'\frac{e^{-i\bm{\omega}\cdot(\bm{x}-\bm{x}')}}{{(2\pi)^{2d}}}\operatorname{tr}\left(O\otimes O^* e_{i'j'}\otimes e_{k'l'}\right)\operatorname{tr}\left(\rho_{\bm{x}}\otimes \rho_{\bm{x}'}^*e_{ij}\otimes e_{kl}\right)W_{ji'}W_{kl'}W_{ij'}^*W_{lk'}^*\\
    &\Rightarrow\mathop{\mathbb{E}}\limits_{\bm{\theta}\sim\Theta}\left(|c_{\bm{\omega}}(\bm{\theta})|^2\right)=\sum\limits_{\substack{i,i',j,j'\\k,k',l,l'}}\iint d^{(2d)}VV'\frac{e^{-i\bm{\omega}\cdot(\bm{x}-\bm{x}')}}{{(2\pi)^{2d}}}\operatorname{tr}\left(O\otimes O^* e_{i'j'}\otimes e_{k'l'}\right)\\
    &\text{\hspace{6cm}}\times\operatorname{tr}\left(\rho_{\bm{x}}\otimes \rho_{\bm{x}'}^*e_{ij}\otimes e_{kl}\right)\mathop{\mathbb{E}}\limits_{\bm{\theta}\sim\Theta}\left(W_{ji'}W_{kl'}W_{ij'}^*W_{lk'}^*\right)
\end{aligned}
\end{equation}
where in the third line we expand 
\begin{equation}\label{eq: trace to tensor indices}
\begin{aligned}
    \operatorname{tr}\left(e_{i'j'}W(\bm{\theta})^{\dagger}e_{ij}W(\bm{\theta})\right)&=\sum\limits_{s}\bra{s}\ket{i'}\bra{j'}W(\bm{\theta})^{\dagger}\ket{i}\bra{j}W(\bm{\theta})\ket{s}\\
    &=\sum_{s}\delta_{si'}W^{\dagger}_{j'i}W_{js}\\
    &=W^{\dagger}_{j'i}W_{ji'}=W_{ij'}^*W_{ji'},
\end{aligned}
\end{equation}
where $\bm{\theta}$ is omitted and absorbed into $W$. Similarly, $\operatorname{tr}\left(e_{k'l'}W(\bm{\theta})^{\dagger}e_{kl}W(\bm{\theta})\right)^*=W_{kl'}W_{lk'}^*$. Hence, under the premise of 2-design and using the identity of Haar integration of the second moment of unitaries \cite{Haar_measure_tutorial}, yields
    \begin{equation}\label{eq: Haar integration 2nd moment}
        \begin{aligned}
            \mathop{\mathbb{E}}\limits_{\bm{\theta}\sim\Theta}\left(W_{ji'}W_{kl'}W_{ij'}^*W_{lk'}^*\right)&=\mathop{\mathbb{E}}\limits_{U\sim\operatorname{U}(d)}\left(U_{ji'}U_{kl'}U_{ij'}^*U_{lk'}^*\right)\\
            &=\frac{\delta_{ji}\delta_{kl}\delta_{i'j'}\delta_{l'k'}+\delta_{jl}\delta_{ki}\delta_{i'k'}\delta_{l'j'}}{d^2-1}-\frac{\delta_{ji}\delta_{kl}\delta_{i'k'}\delta_{l'j'}+\delta_{jl}\delta_{ki}\delta_{i'j'}\delta_{l'k'}}{d(d^2-1)}.
        \end{aligned}
    \end{equation}
Substituting \Cref{eq: Haar integration 2nd moment} back to \Cref{eq: Fourier coefficient second moment phase 1},
    \begin{equation}\label{eq: Fourier coefficient second moment phase 2}
        \begin{aligned}
            \mathop{\mathbb{E}}\limits_{\bm{\theta}\sim\Theta}\left(|c_{\bm{\omega}}(\bm{\theta})|^2\right)&=\sum\limits_{\substack{i,i',j,j'\\k,k',l,l'}}\iint d^{(2d)}VV'\frac{e^{-i\bm{\omega}\cdot(\bm{x}-\bm{x}')}}{{(2\pi)^{2d}}}O_{i'j'}O_{k'l'}^*(\rho_{\bm{x}})_{ij}(\rho_{\bm{x}'}^*)_{kl}\\
            &\times\biggl[\frac{\delta_{ji}\delta_{kl}\delta_{i'j'}\delta_{l'k'}+\delta_{jl}\delta_{ki}\delta_{i'k'}\delta_{l'j'}}{d^2-1}-\frac{\delta_{ji}\delta_{kl}\delta_{i'k'}\delta_{l'j'}+\delta_{jl}\delta_{ki}\delta_{i'j'}\delta_{l'k'}}{d(d^2-1)}\biggr]\\
            &=\sum\limits_{i,j,k,l}\iint d^{(2d)}VV'\frac{e^{-i\bm{\omega}\cdot(\bm{x}-\bm{x}')}}{{(2\pi)^{2d}}}(\rho_{\bm{x}})_{ij}(\rho_{\bm{x}'}^*)_{kl}\\
            &\times\biggl[\frac{\delta_{ji}\delta_{kl}\operatorname{tr}(O)\operatorname{tr}(O^*)+\delta_{jl}\delta_{ki}\|O\|_F^2}{d^2-1}
            -\frac{\delta_{ji}\delta_{kl}\|O\|_F^2+\delta_{jl}\delta_{ki}\operatorname{tr}(O)\operatorname{tr}(O^*)}{d(d^2-1)}\biggr]\\
            &=\iint d^{(2d)}VV'\frac{e^{-i\bm{\omega}\cdot(\bm{x}-\bm{x}')}}{{(2\pi)^{2d}}}\\
            &\times\biggl[\frac{\operatorname{tr}(\rho_{\bm{x}})\operatorname{tr}(\rho_{\bm{x}'}^*)\operatorname{tr}(O)\operatorname{tr}(O^*)+\operatorname{tr}(\rho_{\bm{x}}\rho_{\bm{x}'}^{\dagger})\|O\|_F^2}{d^2-1}\\
            &-\frac{\operatorname{tr}(\rho_{\bm{x}})\operatorname{tr}(\rho_{\bm{x}'}^*)\|O\|_F^2+\operatorname{tr}(\rho_{\bm{x}}\rho_{\bm{x}'}^{\dagger})\operatorname{tr}(O)\operatorname{tr}(O^*)}{d(d^2-1)}\biggr],
        \end{aligned}
    \end{equation}
where in line two the $i',j',k',l'$ indices are contracted, in line three the $i,j,k,l$ indices are further contracted, and $\|O\|_F^2=\sum_{ij}|O_{ij}|^2$ is the Frobenius norm of the operator $O$. When contracting indices for the density matrices, we used the identity $(\rho_{\bm{x}})_{kl}^*=(\rho_{\bm{x}})_{lk}^{*T}=(\rho_{\bm{x}})_{lk}^{\dagger}$, which leads to 
\begin{equation}
\begin{aligned}
    \sum\limits_{i,j,k,l}(\rho_{\bm{x}})_{ij}(\rho_{\bm{x}'}^*)_{kl}\delta_{jl}\delta_{ki}&=\sum\limits_{i,j,k,l}(\rho_{\bm{x}})_{ij}(\rho_{\bm{x}})_{lk}^{\dagger}\delta_{jl}\delta_{ki}\\
    &=\sum\limits_{i,j}(\rho_{\bm{x}})_{ij}(\rho_{\bm{x}})_{ji}^{\dagger}\\
    &=\operatorname{tr}(\rho_{\bm{x}}\rho_{\bm{x}'}^{\dagger}).
\end{aligned}
\end{equation}
We then expand $\operatorname{tr}(\rho_{\bm{x}}\rho_{\bm{x}'}^{\dagger})$ explicitly as 
\begin{equation}\label{eq: trace of rhox rhox prime expansion}
\begin{aligned}
    \operatorname{tr}(\rho_{\bm{x}}\rho_{\bm{x}'}^{\dagger})&=\sum_{k}\bra{k}\sum_{ij}\frac{x_ix_j\ket{i}\bra{j}}{\|\bm{x}\|_2^2}\sum_{i'j'}\frac{x'_{i'}x'_{j'}\ket{i'}\bra{j'}}{\|\bm{x}'\|_2^2}\ket{k}\\
    &=\frac{1}{\|\bm{x}\|_2^2\|\bm{x}'\|_2^2}\sum_{k,i,j,i'j'}x_ix_jx'_{i'}x'_{j'}\delta_{ki}\delta_{ji'}\delta_{j'k}\\
    &=\frac{1}{\|\bm{x}\|_2^2\|\bm{x}'\|_2^2}\sum_{i,j}
x_jx_ix'_{j}x'_{i}.
\end{aligned}
\end{equation}
Substitute \Cref{eq: trace of rhox rhox prime expansion} back to \Cref{eq: Fourier coefficient second moment phase 2}, and note that the trace of density matrices is invariant, the second moment of the Fourier coefficient simplifies to
    \begin{equation}
        \begin{aligned}
            &\mathop{\mathbb{E}}\limits_{\bm{\theta}\sim\Theta}\left(|c_{\bm{\omega}}(\bm{\theta})|^2\right)=\iint d^{(2d)}VV'\frac{e^{-i\bm{\omega}\cdot(\bm{x}-\bm{x}')}}{{(2\pi)^{2d}}}\biggl[\frac{\operatorname{tr}(O)\operatorname{tr}(O^*)+\operatorname{tr}(\rho_{\bm{x}}\rho_{\bm{x}'}^{\dagger})\|O\|_F^2}{d^2-1}\\
            &\hspace{6cm}-\frac{\|O\|_F^2+\operatorname{tr}(\rho_{\bm{x}}\rho_{\bm{x}'}^{\dagger})\operatorname{tr}(O)\operatorname{tr}(O^*)}{d(d^2-1)}\biggr]\\
            &=\iint d^{(2d)}VV'\frac{e^{-i\bm{\omega}\cdot(\bm{x}-\bm{x}')}}{{(2\pi)^{2d}}}\biggl[\frac{\operatorname{tr}(O)\operatorname{tr}(O^*)}{d^2-1}-\frac{\|O\|_F^2}{d(d^2-1)}\biggr]\\
            &+\frac{1}{{(2\pi)^{2d}}}\sum_{i,j}\left|\int d^{(d)}V\frac{e^{-i\bm{\omega}\cdot\bm{x}}x_ix_j}{\|\bm{x}\|_2^2}\right|^2\biggl[\frac{\|O\|_F^2}{d^2-1}-\frac{\operatorname{tr}(O)\operatorname{tr}(O^*)}{d(d^2-1)}\biggr]\\
            &=\delta_{\bm{\omega}}^0\biggl[\frac{\operatorname{tr}(O)\operatorname{tr}(O^*)}{d^2-1}-\frac{\|O\|_F^2}{d(d^2-1)}\biggr]+\frac{1}{{(2\pi)^{2d}}}\mathrm{I}(\bm{\omega})\biggl[\frac{\|O\|_F^2}{d^2-1}-\frac{\operatorname{tr}(O)\operatorname{tr}(O^*)}{d(d^2-1)}\biggr]\\
            &=\frac{1}{d^2-1}\biggl[\left(\operatorname{tr}(O)\operatorname{tr}(O^*)-\frac{\|O\|_F^2}{d}\right)\delta_{\bm{\omega}}^0+\left(\|O\|_F^2-\frac{\operatorname{tr}(O)\operatorname{tr}(O^*)}{d}\right)\frac{\mathrm{I}(\bm{\omega})}{(2\pi)^{2d}}\biggr].
        \end{aligned}
    \end{equation}
Hence, the variance of the Fourier coefficients is evaluated as
\begin{equation}
    \begin{aligned}
        \mathbb{V}\text{ar}_{\bm{\theta}}(c_{\bm{\omega}}(\bm{\theta}))&=\mathop{\mathbb{E}}\limits_{\bm{\theta}\sim\Theta}\left(|c_{\bm{\omega}}(\bm{\theta})|^2\right)-\left|\mathop{\mathbb{E}}\limits_{\bm{\theta}\sim\Theta}\left(c_{\bm{\omega}}(\bm{\theta})\right)\right|^2\\
        &=\frac{1}{d^2-1}\biggl[\left(\operatorname{tr}(O)\operatorname{tr}(O^*)-\frac{\|O\|_F^2}{d}\right)\delta_{\bm{\omega}}^0+\left(\|O\|_F^2-\frac{\operatorname{tr}(O)\operatorname{tr}(O^*)}{d}\right)\frac{\mathrm{I}(\bm{\omega})}{(2\pi)^{2d}}\biggr]\\
        &\hspace{8cm}-\frac{\operatorname{tr}\left(O\right)\operatorname{tr}\left(O^*\right)\delta_{\bm{\omega}}^0}{d^2}\\
        &=\frac{1}{d^2-1}\biggl[\left(\operatorname{tr}(O)^2-\frac{\|O\|_F^2}{d}\right)\delta_{\bm{\omega}}^0+\left(\|O\|_F^2-\frac{\operatorname{tr}(O)^2}{d}\right)\frac{\mathrm{I}(\bm{\omega})}{(2\pi)^{2d}}\biggr]-\frac{\operatorname{tr}\left(O\right)^2\delta_{\bm{\omega}}^0}{d^2},
    \end{aligned}
\end{equation}
where in the third line we use the fact that the observable is Hermitian.
\end{proof}

\subsection{Proof of Corollary~\ref{corollary: upper bound of I(w)}}\label{appendix: upper-bound of I(w) proof}
\begin{proof}
For the ease of later derivations, we first change the integration variable by substitute $x_i=2\pi u_i$, such that 
\begin{equation}\label{eq: u substitution}
        K_{ij}(\bm{\omega})=\int_{\left[0,2\pi\right]^{d}}d^{(d)}\bm{x}\frac{e^{-i\bm{\omega}\cdot\bm{x}}x_ix_{j}}{\|\bm{x}\|_2^2}=(2\pi)^d\int_{\left[0,1\right]^{d}}d^{(d)}\bm{u}\frac{e^{-2\pi i\bm{\omega}\cdot\bm{u}}u_iu_{j}}{\|\bm{u}\|_2^2}.
\end{equation}
Then, we rewrite the product $u_iu_j$ as
\begin{equation}\label{eq: u_iu_j}
    u_iu_j=\prod\limits_{l=1}^{d}u_{l}^{p_l}
\end{equation}
for $p_l=\delta_{li}+\delta_{lj}$. Regarding the quotient $\frac{1}{\|\bm{u}\|_2^2}$, we adopt the Laplace's transform for $f(t)=1$ to express it as an integral with respect to a Gaussian function that can be further separated into products of single-variable functions,
\begin{equation}\label{eq: laplace transform}
    \mathcal{L}\{1\}(\|u\|_2^2)=\frac{1}{\|u\|_2^2}=\int\limits_{0}^{\infty}e^{-t\|u\|_2^2}dt=\int\limits_{0}^{\infty}\prod_{k=1}^{d}e^{-tu_k^2}dt.
\end{equation}
Combining \Cref{eq: u_iu_j,eq: laplace transform} and subtitute back to \Cref{eq: u substitution}, we obtain
\begin{equation}\label{eq: I(w) upper bound derivation step 1}
\begin{aligned}
    \int_{\left[0,1\right]^{d}}d^{(d)}\bm{u}\frac{e^{-2\pi i\bm{\omega}\cdot\bm{u}}u_iu_{j}}{\|\bm{u}\|_2^2}&=\int_{\left[0,1\right]^{d}}d^{(d)}\bm{u}\prod\limits_{l=1}^{d}u_{l}^{p_l}\left(\int\limits_{0}^{\infty}dt\prod_{m=1}^{d}e^{-tu_m^2}\right)e^{-2\pi i \bm{\omega}\cdot\bm{u}}\\
    &=\int\limits_{0}^{\infty}dt\int_{\left[0,1\right]^{d}}d^{(d)}\bm{u}\prod\limits_{l=1}^{d}u_{l}^{p_l}\prod_{m=1}^{d}e^{-tu_m^2}\prod_{n=1}^{d}e^{-2\pi i \omega_{n}u_n}\\
    &=\int\limits_{0}^{\infty}dt\int_{\left[0,1\right]^{d}}d^{(d)}\bm{u}\prod\limits_{l=1}^{d}u_{l}^{p_l}e^{-tu_l^2}e^{-2\pi i \omega_{l}u_l}\\
    &=\int\limits_{0}^{\infty}dt\prod\limits_{l=1}^{d}\int\limits_{0}^{1}du_lu_{l}^{p_l}e^{-tu_l^2}e^{-2\pi i \omega_{l}u_l}\\
    &=\int\limits_{0}^{\infty}dt\prod\limits_{l=1}^{d}J(t,\omega_l,p_l),
\end{aligned}
\end{equation}
where in line two we swap the order of integration and in line four we expand the $d$-dimension integration using Fubini's theorem. We also define the integral $J(t,\omega_l,p_l)=\int\limits_{0}^{1}u^{p_l}e^{-tu^2}e^{-2\pi i \omega_{l}u}du$. Hence, defining $f_{p_l}(u)=u^{p_l}e^{-tu^2},g'(u)=e^{-2\pi i \omega_{l}u}$ and using integration by parts, we obtain
\begin{equation}
    \begin{aligned}
        J(t,\omega_l,p_l)=\frac{f_{p_l}(1)g(1)-f_{p_l}(0)g(0)}{-2\pi i \omega_l}+\frac{1}{2\pi i \omega_l}\int\limits_{0}^{1}f'_{p_l}(u)g(u)du.
    \end{aligned}
\end{equation}
Applying the triangle-inequality, the magnitude of the above integral is bounded by
\begin{equation}\label{eq: upper bound of J omega!=0}
    \begin{aligned}
        |J(t,\omega_l,p_l)|\le\frac{|f_{p_l}(1)|+|f_{p_l}(0)|}{2\pi|\omega_l|}+\frac{1}{2\pi|\omega_l|}\int\limits_{0}^{1}|f'_{p_l}(u)|du,
    \end{aligned}
\end{equation}
where we used the fact that $|g(u)|=1$. Now, since $u_l\in\{0,1,2\}$, there are three possible variants of $f_{p_l}(u)$:
\begin{equation}
    \begin{aligned}
        &f_0(u)=e^{-tu^2},\\
        &f_1(u)=ue^{-tu^2},\\
        &f_2(u)=u^2e^{-tu^2},
    \end{aligned}
\end{equation}
as well as three possible $f'_{p_l}(u)$:
\begin{equation}
    \begin{aligned}
        &f'_0(u)=e^{-tu^2}(-2tu),\\
        &f'_1(u)=e^{-tu^2}(1-2tu^2),\\
        &f'_2(u)=e^{-tu^2}(2u-2tu^3).
    \end{aligned}
\end{equation}
Clearly, $f_{p_l}(u)\ge0$ $\forall u\in\left[0,1\right],t\ge0$, and $f'_0(u)\le0$, $f'_1(u)$ crosses $0$ once at $u=(2t)^{\frac{-1}{2}}$, $f'_2(u)$ crosses $0$ once at $u=(t)^{\frac{-1}{2}}$. Also, defining $M_{p_l}=\max\limits_{u\in\left[0,1\right]}f_{p_l}(u)$, it is obvious that
\begin{equation}
    M_0=0, M_1=\begin{cases}
        1,0\le t \le1\\
        (2et)^{-\frac{1}{2}},t>1
        \end{cases}, M_2=\begin{cases}
            1,0\le t \le1,\\
            (et)^{-1}, t>1
        \end{cases},
\end{equation}
and a general-case inequality
\begin{equation}
    M_{p_l}\le \begin{cases}
        1,0\le t\le 1,\\
        t^{-\frac{p_l}{2}}, t> 1
    \end{cases}
\end{equation}
can be established. Hence, the absolute-valued integration on $f'_l(u)$ is expanded as:
\begin{equation}
    \begin{aligned}
        &\int\limits_{0}^{1}|f'_0(u)|du=-\int\limits_{0}^{1}f'_0(u)du=f_0(0)-f_0(1)\le f_0(0)\le M_{0}\le2M_0\\
        &\int\limits_{0}^{1}|f'_1(u)|du=\int\limits_{0}^{(2t)^{\frac{-1}{2}}}f'_1(u)du-\int\limits_{(2t)^{\frac{-1}{2}}}^{1}f'_1(u)du=f_1\left((2t)^{\frac{-1}{2}}\right)-f_1(0)-f_1(1)+f_1\left((2t)^{\frac{-1}{2}}\right)\\
        &\le2f_1\left((2t)^{\frac{-1}{2}}\right)\le 2M_{1}\\
        &\int\limits_{0}^{1}|f'_2(u)|du=\int\limits_{0}^{t^{\frac{-1}{2}}}f'_2(u)du-\int\limits_{t^{\frac{-1}{2}}}^{1}f'_1(u)du=f_2\left(t^{\frac{-1}{2}}\right)-f_2(0)-f_2(1)+f_2\left(t^{\frac{-1}{2}}\right)\le2f_2\left(t^{\frac{-1}{2}}\right)\le2M_{2}.
    \end{aligned}
\end{equation}
Finally, the upper bound \Cref{eq: upper bound of J omega!=0} is evaluated as
\begin{equation}\label{eq: upper bound of J omega!=0 part 2}
    |J(t,\omega_l,p_l)|\le\frac{2M_{p_l}}{2\pi|\omega_l|}+\frac{2M_{p_l}}{2\pi|\omega_l|}=\frac{2M_{p_l}}{\pi|\omega_l|}.
\end{equation}
However, the above identity only holds for $\omega_l\ne0$ otherwise would be undefined. In the case $\omega_l=0$, the upper bound of $J(t,0,p_l)$ can be evaluated as
\begin{equation}\label{eq: upper bound of J omega=0}
    |J(t,0,p_l)|\le\int\limits_{0}^{1}|u^{p_l}e^{-tu^2}e^{-2\pi i \omega_{l}u}|du\le\int\limits_{0}^{1}u^{p_l}e^{-tu^2}du.
\end{equation}
For $0\le t\le1$, it is clear that $\int\limits_{0}^{1}u^{p_l}e^{-tu^2}du\le1$. For $t\ge1$, we can further extend the range of integration terminal and apply the Gamma integral identity
\begin{equation}
        \int\limits_{0}^{1}u^{p_l}e^{-tu^2}du\le\int\limits_{0}^{\infty}u^{p_l}e^{-tu^2}du=\frac{\Gamma\left(\frac{p_l+1}{2}\right)}{2t^{\frac{p_l+1}{2}}}\le t^{-\frac{p_l+1}{2}},
\end{equation}
where $\Gamma$ is the Gamma function. Ultimately, combining the cases for $\omega_l=0$ (\Cref{eq: upper bound of J omega=0}), $\omega_l\ne0$ (\Cref{eq: upper bound of J omega!=0 part 2}), and define $S=\{l,\omega_l\ne0\}, s=|S|$, the magnitude of the integral $\frac{|K_{ij}(\bm{\omega})|}{(2\pi)^d}$ is bounded by
\begin{equation}
\begin{aligned}
    \frac{|K_{ij}(\bm{\omega})|}{(2\pi)^d}&\le\int\limits_{0}^{\infty}dt\prod\limits_{l=1}^{d}|J(t,\omega_l,p_l)|\\
    &=\int\limits_{0}^{1}dt\prod\limits_{l=1}^{d}|J(t,\omega_l,p_l)|+\int\limits_{1}^{\infty}dt\prod\limits_{l=1}^{d}|J(t,\omega_l,p_l)|\\
    &\le\int\limits_{0}^{1}dt\prod\limits_{\omega_l\ne0}\frac{2}{\pi|\omega_l|}+\int\limits_{1}^{\infty}dt\prod\limits_{\omega_l\ne0}\frac{2t^{\frac{-p_l}{2}}}{\pi|\omega_l|}\prod\limits_{\omega_l=0}t^{-\frac{p_l+1}{2}}\\
    &=\int\limits_{0}^{1}dt\left(\frac{2}{\pi}\right)^s\frac{1}{\prod_{l\in S}|\omega_l|}+\int\limits_{1}^{\infty}dt\left(\frac{2}{\pi}\right)^s\frac{t^{-\frac{1}{2}(\sum_{l\in S}p_l+\sum_{l\notin S}p_l+1)}}{\prod_{l\in S}|\omega_l|}\\
    &=\left(\frac{2}{\pi}\right)^s\frac{1}{\prod_{l\in S}|\omega_l|}\left(1+\int\limits_{1}^{\infty}dt\,t^{-\frac{2+d-s}{2}}\right)\\
    &=\left(\frac{2}{\pi}\right)^s\frac{1}{\prod_{l\in S}|\omega_l|}\left(1+\frac{2}{d-s}\right)
\end{aligned}
\end{equation}
where we use the property $\sum_{l=1}^dp_l=2$ in the last line. Note that the integral in the second last line converges only if $d>s$, hence for the frequency vector that only contains non-zero components, this upper bound is undefined. Finally, the sum of integrals squared $\mathrm{I}(\bm{\omega})$ is upper bounded by:
\begin{equation}
    \begin{aligned}
        \mathrm{I}(\bm{\omega})&=\sum\limits_{i,j=1}^{d}|K_{ij}(\bm{\omega})|^2\\
        &\le d^2(2\pi)^{2d}\left(\frac{2}{\pi}\right)^{2s}\left(1+\frac{2}{d-s}\right)^2\prod\limits_{l\in S}|\omega_l|^{-2}.
    \end{aligned}
\end{equation}
\end{proof}
\subsection{Proof of theorem~\ref{theorem: depolarizing noise mean and variance}}\label{appendix: noisy amplitude embedding proofs}
\begin{proof}
Under the combined quantum channel $\tilde{\Lambda_i}$ acting on each gate of the unitary $W(\bm{\theta})=\prod\limits_{i=1}^MW_i(\bm{\theta}_i)$ with $M$ gates, the net evolution becomes $\Tilde{\Lambda}=\prod\limits_{i=1}^M\mathcal{E}_i\circ\Lambda_i(\bm{\theta}_i)$, and the single-qubit depolarizing channel acts $Q\ge M$ times independently on the system (the single-qubit channel acts at least once per gate). Hence, we can adapt \Cref{eq: amplitude embedding Fourier coefficients} and \Cref{eq: Fourier coefficients first moment} to obtain
    \begin{equation}
        \tilde{c}_{\bm{\omega}}(\bm{\theta})=\frac{1}{(2\pi)^d}\sum\limits_{i,j,i',j'}\sum\limits_{k}\int d^{(d)}Ve^{-i\bm{\omega}\cdot\bm{x}}\operatorname{tr}\left(Oe_{i'j'}\right)\operatorname{tr}\left(p_ke_{i'j'}\tilde{W}_k(\bm{\theta})^{\dagger}e_{ij}\tilde{W}_k(\bm{\theta})\right)\operatorname{tr}(\rho_{\bm{x}}e_{ij}),
\end{equation}
where we realize that each perturbed building block takes the form
\begin{equation}
\tilde{\Lambda}_i[\cdot]=\mathcal{E}_i\circ\Lambda_i[\cdot]=\sum_{k_i}p_{k_i}E_{k_i}W_i\cdot W_i^{\dagger}E_{k_i}^{\dagger}.
\end{equation}
with $E_{k_i}$ and $p_{k_i}$ denoting, respectively, the Kraus operators of the channel acting on gate $i$ and the corresponding probabilities. Therefore, 
\begin{equation}
\tilde{\Lambda}[\cdot]=\sum_{k}p_{k}W_k\cdot W_k^{\dagger},
\end{equation}
where $\tilde{W}_k=\prod_{i=1}^{M}E_{k_i}W_i$ is the unitary $W(\bm{\theta})$ perturbed by the depolarizing channels along a specific sequence of Kraus operators $k=(k_0,k_1,\cdots,k_{M})$, and $p_k=\prod_{i=1}^{M}p_{k_i}$. As a result, the first moment of the noisy Fourier coefficient is expressed as\
\begin{equation}
    \begin{aligned}
        \mathop{\mathbb{E}}\limits_{\bm{\theta}\sim\Theta}\left(\tilde{c}_{\bm{\omega}}(\bm{\theta})\right)&=\frac{1}{(2\pi)^d}\sum\limits_{i,j,i',j'}\sum_k\int d^{(d)}Ve^{-i\bm{\omega}\cdot\bm{x}}\operatorname{tr}\left(Oe_{i'j'}\right)\operatorname{tr}\left(p_k e_{i'j'}\mathop{\mathbb{E}}\limits_{\bm{\theta}\sim\Theta}\left[\tilde{W}_k(\bm{\theta})^{\dagger}e_{ij}\tilde{W}_k(\bm{\theta})\right]\right)\operatorname{tr}(\rho_{\bm{x}}e_{ij})\\
        &=\frac{1}{(2\pi)^d}\sum\limits_{i,j,i',j'}\sum_k\int d^{(d)}Ve^{-i\bm{\omega}\cdot\bm{x}}\operatorname{tr}\left(Oe_{i'j'}\right)\operatorname{tr}\left(p_k e_{i'j'}\mathop{\mathbb{E}}\limits_{U\sim\text{U}(d)}\left[U^{\dagger}e_{ij}U\right]\right)\operatorname{tr}(\rho_{\bm{x}}e_{ij})\\
        &=\sum\limits_{k}p_k\frac{\operatorname{tr}\left(O\right)\delta_{\bm{\omega}}^0}{d}\\
        &=\frac{\operatorname{tr}\left(O\right)\delta_{\bm{\omega}}^0}{d},
    \end{aligned}
\end{equation}
which coincides with the first moment of the noiseless case. In the second equality we apply the 2-design (hence also 1-design) assumption of $\text{Im}(\mathcal{W})$ and the fact that Haar measure is invariant for all perturbation realizations $k$. In the third line, we isolate the terms with respect to indices $k$, the remaining steps are identical to those in \Cref{eq: Fourier coefficients first moment}, and we apply the Born rule $\sum_kp_k=1$. 

Following similar steps as in the proof of \Cref{theorem: noiseless mean and variance}, we find the square of the Fourier coefficients as
    \begin{equation}\label{eq: noisy Fourier coefficient second moment phase 1}
        \begin{aligned}
            |\tilde{c}_{\bm{\omega}}(\bm{\theta})|^2=\frac{1}{(2\pi)^{2d}}\left|\sum\limits_{i,j,i',j'}\sum\limits_{s}\int d^{(d)}Ve^{-i\bm{\omega}\cdot\bm{x}}\operatorname{tr}\left(Oe_{i'j'}\right)\operatorname{tr}\left(p_s e_{i'j'}\tilde{W}_s(\bm{\theta})^{\dagger}e_{ij}\tilde{W}_s(\bm{\theta})\right)\operatorname{tr}(\rho_{\bm{x}}e_{ij})\right|^2&\\
            =\sum\limits_{\substack{i,i',j,j'\\k,k',l,l'}}\sum\limits_{s,s'}p_sp_{s'}\iint d^{(2d)}VV'\frac{e^{-i\bm{\omega}\cdot(\bm{x}-\bm{x}')}}{{(2\pi)^{2d}}}\operatorname{tr}\left(O\otimes O^* e_{i'j'}\otimes e_{k'l'}\right)\operatorname{tr}\left(e_{i'j'}\tilde{W}_s(\bm{\theta})^{\dagger}e_{ij}\tilde{W}_s(\bm{\theta})\right)&\\
            \times\operatorname{tr}\left(e_{k'l'}\tilde{W}_{s'}(\bm{\theta})^{\dagger}e_{kl}\tilde{W}_{s'}(\bm{\theta})\right)^*\operatorname{tr}\left(\rho_{\bm{x}}\otimes \rho_{\bm{x}'}^*e_{ij}\otimes e_{kl}\right)&\\
            =\sum\limits_{\substack{i,i',j,j'\\k,k',l,l'}}\sum\limits_{s,s'}p_sp_{s'}\iint d^{(2d)}VV'\frac{e^{-i\bm{\omega}\cdot(\bm{x}-\bm{x}')}}{{(2\pi)^{2d}}}O_{i'j'}O_{k'l'}^*(\rho_{\bm{x}})_{ij}(\rho_{\bm{x}'}^*)_{kl}(\tilde{W}_s)_{ji'}(\tilde{W}_{s'})_{kl'}(\tilde{W}_s)_{ij'}^*(\tilde{W}_{s'})_{lk'}^*&,
        \end{aligned}
    \end{equation}
where in the third line we adopt the same arguments as in \Cref{eq: trace to tensor indices} to simplify the equation. The second moment then splits naturally into diagonal ($s=s'$) and off-diagonal ($s\neq s'$) contributions
    \begin{equation}\label{eq: noisy Fourier coefficient second moment phase 2}
        \begin{aligned}
            &\mathop{\mathbb{E}}\limits_{\bm{\theta}\sim\Theta}\left(|\tilde{c}_{\bm{\omega}}(\bm{\theta})|^2\right)=\\
            &\sum\limits_{s}p_s^2\sum\limits_{\substack{i,i',j,j'\\k,k',l,l'}}\iint d^{(2d)}VV'\frac{e^{-i\bm{\omega}\cdot(\bm{x}-\bm{x}')}}{{(2\pi)^{2d}}}O_{i'j'}O_{k'l'}^*(\rho_{\bm{x}})_{ij}(\rho_{\bm{x}'}^*)_{kl}\mathop{\mathbb{E}}\limits_{\bm{\theta}\sim\Theta}\left[(\tilde{W}_s)_{ji'}(\tilde{W}_{s})_{kl'}(\tilde{W}_s)_{ij'}^*(\tilde{W}_{s})_{lk'}^*\right]\\
            &+\sum\limits_{s\ne s'}p_s p_{s'}\sum\limits_{\substack{i,i',j,j'\\k,k',l,l'}}\iint d^{(2d)}VV'\frac{e^{-i\bm{\omega}\cdot(\bm{x}-\bm{x}')}}{{(2\pi)^{2d}}}O_{i'j'}O_{k'l'}^*(\rho_{\bm{x}})_{ij}(\rho_{\bm{x}'}^*)_{kl}\mathop{\mathbb{E}}\limits_{\bm{\theta}\sim\Theta}\left[(\tilde{W}_s)_{ji'}(\tilde{W}_{s'})_{kl'}(\tilde{W}_s)_{ij'}^*(\tilde{W}_{s'})_{lk'}^*\right]\\
            &=\sum\limits_{s}p_s^2\times\frac{1}{d^2-1}\biggl[\left(\operatorname{tr}(O)\operatorname{tr}(O^*)-\frac{\|O\|_F^2}{d}\right)\delta_{\bm{\omega}}^0+\left(\|O\|_F^2-\frac{\operatorname{tr}(O)\operatorname{tr}(O^*)}{d}\right)\frac{\mathrm{I}(\bm{\omega})}{(2\pi)^{2d}}\biggr]\\
            &+\sum\limits_{s\ne s'}p_s p_{s'}\sum\limits_{\substack{i,i',j,j'\\k,k',l,l'}}\iint d^{(2d)}VV'\frac{e^{-i\bm{\omega}\cdot(\bm{x}-\bm{x}')}}{{(2\pi)^{2d}}}O_{i'j'}O_{k'l'}^*(\rho_{\bm{x}})_{ij}(\rho_{\bm{x}'}^*)_{kl}\mathop{\mathbb{E}}\limits_{\bm{\theta}\sim\Theta}\left[(\tilde{W}_s)_{ji'}(\tilde{W}_{s'})_{kl'}(\tilde{W}_s)_{ij'}^*(\tilde{W}_{s'})_{lk'}^*\right],
        \end{aligned}
    \end{equation}
where we use the results of \cref{eq: Fourier coefficient second moment phase 1,eq: Fourier coefficient second moment phase 2} to directly obtain the first half of line two, under the assumption that integration over the perturbed unitary $W_k$ still respects the 2-design Haar measure. For the second half of line two, we employ the covariance identity $E(XY)=E(X)E(Y)+Cov(X,Y)$, hence
\begin{equation}\label{eq: noisy Fourier coefficient second moment covariance}
    \begin{aligned}
        &\mathop{\mathbb{E}}\limits_{\bm{\theta}\sim\Theta}\left[(\tilde{W}_s)_{ji'}(\tilde{W}_{s'})_{kl'}(\tilde{W}_s)_{ij'}^*(\tilde{W}_{s'})_{lk'}^*\right]\\
        &=\mathop{\mathbb{E}}\limits_{\bm{\theta}\sim\Theta}\left[(\tilde{W}_s)_{ji'}(\tilde{W}_s)_{ij'}^*\right]\mathop{\mathbb{E}}\limits_{\bm{\theta}\sim\Theta}\left[(\tilde{W}_{s'})_{lk'}^*(\tilde{W}_{s'})_{kl'}\right]+\mathop{\text{Cov}}\limits_{\bm{\theta}\sim\Theta}\left[(\tilde{W}_s)_{ji'}(\tilde{W}_s)_{ij'}^*, (\tilde{W}_{s'})_{lk'}^*(\tilde{W}_{s'})_{kl'}\right]\\
        &=\mathop{\mathbb{E}}\limits_{\tilde{W_s}\sim\text{U}(d)}\left[(\tilde{W}_s)_{ji'}(\tilde{W}_s)_{ij'}^*\right]\mathop{\mathbb{E}}\limits_{\tilde{W_{s'}}\sim\text{U}(d)}\left[(\tilde{W}_{s'})_{lk'}^*(\tilde{W}_{s'})_{kl'}\right]+\mathop{\text{Cov}}\limits_{\tilde{W_s},\tilde{W_{s'}}\sim\text{U}(d)}\left[(\tilde{W}_s)_{ji'}(\tilde{W}_s)_{ij'}^*, (\tilde{W}_{s'})_{lk'}^*(\tilde{W}_{s'})_{kl'}\right]\\
        &=\mathop{\mathbb{E}}\limits_{\tilde{W_s}\sim\text{U}(d)}\left[(\tilde{W}_s)_{ji'}(\tilde{W}_s)_{ij'}^*\right]\mathop{\mathbb{E}}\limits_{\tilde{W_{s'}}\sim\text{U}(d)}\left[(\tilde{W}_{s'})_{lk'}^*(\tilde{W}_{s'})_{kl'}\right],
    \end{aligned}
\end{equation}
where on the second line we use the 2-design and the invariance assumption again. Since we also make the premise that the perturbed unitaries $\tilde{W}(\bm{\theta})_s,\tilde{W}(\bm{\theta})_{s'}$ induced by arbitrary two distinct realizations $s\ne s'$ are uncorrelated, the covariance vanishes in the third line. Therefore, substituting back to \Cref{eq: noisy Fourier coefficient second moment phase 2} and applying the Weingarten identity, its second half yield
    \begin{equation}\label{eq: noisy Fourier coefficient second moment phase 3}
        \begin{aligned}
            &\sum\limits_{s\ne s'}p_s p_{s'}\sum\limits_{\substack{i,i',j,j'\\k,k',l,l'}}\iint d^{(2d)}VV'\frac{e^{-i\bm{\omega}\cdot(\bm{x}-\bm{x}')}}{{(2\pi)^{2d}}}O_{i'j'}O_{k'l'}^*(\rho_{\bm{x}})_{ij}(\rho_{\bm{x}'}^*)_{kl}\frac{\delta_{ji}\delta_{i'j'}\delta_{lk}\delta_{k'l'}}{d^2}\\
            &=\sum\limits_{s\ne s'}p_s p_{s'}\iint d^{(2d)}VV'\frac{e^{-i\bm{\omega}\cdot(\bm{x}-\bm{x}')}}{{(2\pi)^{2d}}}\frac{\operatorname{tr}(O)\operatorname{tr}(O^*)\operatorname{tr}(\rho_{\bm{x}})\operatorname{tr}(\rho_{\bm{x}'}^{\dagger})}{d^2}\\
            &=\sum\limits_{s\ne s'}p_s p_{s'}\iint d^{(2d)}VV'\frac{e^{-i\bm{\omega}\cdot(\bm{x}-\bm{x}')}}{{(2\pi)^{2d}}}\frac{\operatorname{tr}(O)\operatorname{tr}(O^*)}{d^2}\\
            &=\sum\limits_{s\ne s'}p_s p_{s'}\frac{\operatorname{tr}(O)\operatorname{tr}(O^*)}{d^2}\delta_{\bm{\omega}}^0.
        \end{aligned}
    \end{equation}
On the second line, the tensor indices are contracted; on the third line, we impose the condition that the trace of the density matrix is always equal to one; and in the last line, we apply the Dirac-$\delta$-Function integration for complex exponentials. Lastly, we focus on the sums over the probabilities, such that

\begin{equation}\label{eq: sum ps squared}
\begin{aligned}
     \sum_{s}p_s^2&=\sum\limits_{(s_1,s_2,\cdots,s_Q)}\prod\limits_{i=1}^{Q} p_{s_i}^2\\
     &=\prod\limits_{i=1}^{Q}\sum_{s_i}p_{s_i}^2\\
     &=\left(\sum\limits_{s=1}^Kp_s^2\right)^Q,
\end{aligned}
\end{equation}
where on the third line we use $p_{s_i}=p_s$ since the channels are identical across the circuit. Similarly, 
\begin{equation}\label{eq: sum psps'}
\begin{aligned}
     \sum_{s\ne s'}p_s p_{s'}&=\sum_{s}p_s\sum_{s'\ne s}p_{s'}\\
     &=\sum_{s}p_s(1-p_{s})\\
     &=\sum_{s}p_s-\sum_{s}p_s^2\\
     &=1-\left(\sum\limits_{s=1}^Kp_s^2\right)^Q,
\end{aligned}
\end{equation}
where the first and last equalities invoke the Born rule, and substituted \Cref{eq: sum ps squared} to line four. After all, combining \Cref{eq: sum ps squared,eq: sum psps',eq: noisy Fourier coefficient second moment phase 3} and subtituting to \Cref{eq: noisy Fourier coefficient second moment phase 2}, we have
    \begin{equation}
        \begin{aligned}
            &\mathop{\mathbb{E}}\limits_{\bm{\theta}\sim\Theta}\left(|\tilde{c}_{\bm{\omega}}(\bm{\theta})|^2\right)=\\
            &=\frac{1}{d^2-1}\biggl[\left(\operatorname{tr}(O)\operatorname{tr}(O^*)-\frac{\|O\|_F^2}{d}\right)\delta_{\bm{\omega}}^0+\left(\|O\|_F^2-\frac{\operatorname{tr}(O)\operatorname{tr}(O^*)}{d}\right)\frac{\mathrm{I}(\bm{\omega})}{(2\pi)^{2d}}\biggr]\left(\sum\limits_{k=1}^Kp_k^2\right)^Q\\
            &+\frac{\operatorname{tr}(O)\operatorname{tr}(O^*)}{d^2}\delta_{\bm{\omega}}^0\left(1-\left(\sum\limits_{k=1}^Kp_k^2\right)^Q\right)\\
            &=\frac{1}{d^2-1}\biggl[\left(\frac{\operatorname{tr}(O)\operatorname{tr}(O^*)}{d^2}-\frac{\|O\|_F^2}{d}\right)\delta_{\bm{\omega}}^0+\left(\|O\|_F^2-\frac{\operatorname{tr}(O)\operatorname{tr}(O^*)}{d}\right)\frac{\mathrm{I}(\bm{\omega})}{(2\pi)^{2d}}\biggr]\left(\sum\limits_{k=1}^Kp_k^2\right)^Q\\
            &=\frac{1}{d^2-1}\biggl[\left(\frac{\operatorname{tr}(O)^2}{d^2}-\frac{\|O\|_F^2}{d}\right)\delta_{\bm{\omega}}^0+\left(\|O\|_F^2-\frac{\operatorname{tr}(O)^2}{d}\right)\frac{\mathrm{I}(\bm{\omega})}{(2\pi)^{2d}}\biggr]\left(\sum\limits_{k=1}^Kp_k^2\right)^Q.
        \end{aligned}
    \end{equation}
Having obtained both the first and the second moments of the noisy Fourier coefficients. As a result, the variances are given by
    \begin{equation}
        \begin{aligned}
            &\mathbb{V}\text{ar}_{\bm{\theta}}(\tilde{c}_{\omega}(\bm{\theta}))=\mathop{\mathbb{E}}\limits_{\bm{\theta}\sim\Theta}\left(|\tilde{c}_{\bm{\omega}}(\bm{\theta})|^2\right)-\left|\mathop{\mathbb{E}}\limits_{\bm{\theta}\sim\Theta}\left(\tilde{c}_{\bm{\omega}}(\bm{\theta})\right)\right|^2\\
            &=\frac{1}{d^2-1}\biggl[\left(\frac{\operatorname{tr}(O)^2}{d^2}-\frac{\|O\|_F^2}{d}\right)\delta_{\bm{\omega}}^0+\left(\|O\|_F^2-\frac{\operatorname{tr}(O)^2}{d}\right)\frac{\mathrm{I}(\bm{\omega})}{(2\pi)^{2d}}\biggr]\left(\sum\limits_{k=1}^Kp_k^2\right)^Q-\frac{\operatorname{tr}(O)^2\delta_{\bm{\omega}}^0}{d^2}.
        \end{aligned}
    \end{equation}

\begin{DIFnomarkup}
\begin{center}
\begingroup
\sffamily
\definecolor{QCovNavy}{HTML}{07345C}
\definecolor{QCovGreen}{HTML}{42743C}
\definecolor{QCovPurple}{HTML}{793CA3}
\definecolor{QCovTeal}{HTML}{007F91}
\definecolor{QCovMuted}{HTML}{526A7B}
\definecolor{QCovRule}{HTML}{CBDDE5}
\newcommand{\QCovCopy}[2]{%
 \foreach \y in {0,7,14,21}{\draw[QCovNavy!65,line width=.8pt] (0,\y)--(63,\y);}
 \foreach \x/\y in {3/17.5,3/3.5,25/10.5,47/17.5,47/3.5}{
  \draw[draw=QCovGreen,fill=white,line width=1pt,rounded corners=.9mm] (\x,\y-5.75) rectangle (\x+13,\y+5.75);
  \ifnum#2=1\node[text=QCovGreen,font=\fontsize{17}{20}\selectfont,inner sep=0pt] at (\x+10,\y+3.4){$*$};\fi
 }
 \ifnum#1=0
  \def\QCovErrors{20/21/X/orange!65,42/7/Y/QCovPurple!55,20/0/Z/QCovTeal!65}
 \else
  \def\QCovErrors{20/14/Y/QCovPurple!55,42/21/X/orange!65,42/0/Z/QCovTeal!65}
 \fi
 \foreach \x/\y/\label/\col in \QCovErrors{
  \node[draw=gray,fill=\col,rounded corners=.7mm,minimum width=7mm,minimum height=6mm,inner sep=0pt,font=\fontsize{20}{24}\selectfont] at (\x,\y){$\label\ifnum#2=1{}^*\fi$};
 }
}
\newcommand{\QCovPair}[1]{%
\begin{tikzpicture}[x=1mm,y=1mm,baseline={([yshift=-1mm]current bounding box.center)},font=\fontsize{23}{27}\selectfont]
 \path[use as bounding box] (-7,-15) rectangle (153,36);
 \QCovCopy{#1}{#1}
 \pgfmathtruncatemacro{\QCovConjugate}{1-#1}
 \begin{scope}[xshift=83mm]\QCovCopy{#1}{\QCovConjugate}\end{scope}
 \node[font=\fontsize{28}{32}\selectfont] at (73,10.5){$\times$};
 \draw[QCovRule,line width=.9pt] (0,-7)--(0,-9)--(146,-9)--(146,-7);
 \ifnum#1=0
  \node at (31.5,31){$\tilde W_s$};\node at (114.5,31){$\tilde W_s^*$};
  \node[anchor=east,inner sep=.3mm] at (-1,10.5){$i'$};\node[anchor=west,inner sep=.3mm] at (64,10.5){$j$};
  \node[anchor=east,inner sep=.3mm] at (82,10.5){$j'$};\node[anchor=west,inner sep=.3mm] at (147,10.5){$i$};
  \node[fill=white,inner sep=.6mm,text=QCovMuted] at (73,-10){same path $s$};
 \else
  \node at (31.5,31){$\tilde W_{s'}^*$};\node at (114.5,31){$\tilde W_{s'}$};
  \node[anchor=east,inner sep=.3mm] at (-1,10.5){$k'$};\node[anchor=west,inner sep=.3mm] at (64,10.5){$l$};
  \node[anchor=east,inner sep=.3mm] at (82,10.5){$l'$};\node[anchor=west,inner sep=.3mm] at (147,10.5){$k$};
  \node[fill=white,inner sep=.6mm,text=QCovMuted] at (73,-10){same path $s'$};
 \fi
\end{tikzpicture}%
}
\resizebox{\linewidth}{!}{%
 {\fontsize{23}{27}\selectfont $\displaystyle
 \mathop{\operatorname{Cov}}_{\tilde W_s,\tilde W_{s'}\sim\operatorname{U}(d)}
 \left(\QCovPair{0}\hspace{3mm},\hspace{3mm}\QCovPair{1}\right)=0$}%
}
\endgroup
\captionof{figure}{
    Schematic illustration of the covariance condition for distinct Kraus
    trajectories $s\neq s'$. White boxes represent the shared noiseless
    gates, while colored boxes indicate single-qubit Pauli errors.
    Each pair represents a product of matrix elements from the same
    trajectory; $*$ denotes complex conjugation.
    The vanishing covariance is an assumption used in the derivation.
}
\label{fig:uncorrelated-kraus-trajectories}
\end{center}
\end{DIFnomarkup}
\end{proof}

\section{Rescaled Monte Carlo Integration error estimation}\label{appendix: Monte Carlo Integration error estimation}
According to \Cref{eq: I(w)}, the $\mathrm{I}(\bm{\omega})$ consists of sums of squared integrals, $\mathrm{I}(\bm{\omega})=\sum_{ij}|K_{ij}(\bm{\omega})|^2$, where we can let 
\begin{equation}
    K_{ij}(\bm{\omega})=\int d^{(d)}V\frac{e^{-i\bm{\omega}\cdot\bm{x}}x_ix_j}{\|\bm{x}\|_2^2}=\int d^{(d)}Vg_{ij}(\bm{x})e^{-i\bm{\omega}\cdot\bm{x}}.
\end{equation}
To evaluate this integral for arbitrary $i,j$ using Monte Carlo methods, we first implement a control-variate technique to reduce its statistical variance. Decomposing $g_{ij}(\bm{x})=(g_{ij}(\bm{x})-c)+c$, where $c$ is some real constant, we can rewrite $K_{ij}(\bm{\omega})$,
\begin{equation}
\begin{aligned}
        K_{ij}(\bm{\omega})&=\int_D d^{(d)}V(g_{ij}(\bm{x})-c)e^{-i\bm{\omega}\cdot\bm{x}}+c\int_D d^{(d)}Ve^{-i\bm{\omega}\cdot\bm{x}}\\
        &=\int_D d^{(d)}V\text{ Y}_c+cV\delta_{\bm{\omega}}^0,
\end{aligned}
\end{equation}
where $\mathrm{Y}_c(\bm{x})=(g_{ij}(\bm{x})-c)e^{-i\bm{\omega}\cdot\bm{x}}$ is regarded a random variable with $\text{Y}_c$ act as its own PDF on $\bm{x}$. The  Hence, the estimator $\hat{K}_{ij}(\bm{\omega})$ by Monte Carlo method can be expressed as
\begin{equation}
    \hat{K}_{ij}(\bm{\omega})=V\frac{\sum_{k=1}^N\text{Y}_c(\bm{x}_k)}{N}+cV\delta_{\bm{\omega}}^0=V\bar{\text{Y}}_c+cV\delta_{\bm{\omega}}^0.
\end{equation}
Therefore, $\mathbb{V}\text{ar}(\hat{K}_{ij}(\bm{\omega}))=\mathbb{V}\text{ar}(V\bar{\text{Y}}_c)=\frac{V^2}{N}\mathbb{V}\text{ar}(\text{Y}_c)$, where
\begin{equation}\label{eq: var Yc}
    \mathbb{V}\text{ar}(\text{Y}_c)=\mathbb{E}(|\text{Y}_c|^2)-|\mathbb{E}(\text{Y}_c)|^2.
\end{equation}
To find the infimum of this variance, we let 
\begin{equation}\label{eq: Eg=c condition}
\begin{aligned}
    &\frac{d}{dc}\mathbb{E}(|\text{Y}_c|^2)-|\mathbb{E}(\text{Y}_c)|^2=0\\
    &\frac{d}{dc}\mathbb{E}(|g_{ij}-c|^2)-|\mathbb{E}(g_{ij}-c)|^2=0\\
    &\frac{d}{dc}\mathbb{E}(g_{ij}^2-2cg_{ij}+c^2)-\mathbb{E}(g_{ij})^2-\mathbb{E}(g_{ij})\mathbb{E}(c)+\mathbb{E}(c)^2=0\\
    &-2\mathbb{E}(g_{ij})+2\mathbb{E}(c)-2\mathbb{E}(g_{ij})+2\mathbb{E}(c)=0\\
    &\Rightarrow \mathbb{E}(g_{ij})=\mathbb{E}(c)=c,
\end{aligned}
\end{equation}
where on the second line we drop the $e^{-i\bm{\omega}\cdot\bm{x}}$ term from complex conjugation, and dropping all conjugates in line three by noticing that $g_{ij},c$ are always real and $\mathbb{E}(c)=c$ is a constant. Ultimately, using the condition \Cref{eq: Eg=c condition}, \Cref{eq: var Yc} becomes 
\begin{equation}\label{eq: var Yc expanded}
\begin{aligned}
    \mathbb{V}\text{ar}(\text{Y}_c)&=\mathbb{E}(|g_{ij}-\mathbb{E}(g_{ij})|^2)-|\mathbb{E}(\text{Y}_c)|^2\\
    &=\mathbb{V}\text{ar}(g_{ij})-|\mathbb{E}(\text{Y}_c)|^2
\end{aligned}
\end{equation}
On the other hand, without the control variate, the estimator is simply
\begin{equation}
    \hat{K}_{ij}(\bm{\omega})=V\frac{\sum_{k=1}^N\text{Y}(\bm{x}_k)}{N},
\end{equation}
where $\text{Y}(\bm{x})=g_{ij}(\bm{x})e^{-i\bm{\omega}\cdot\bm{x}}$, and $\mathbb{V}\text{ar}(\hat{K}_{ij}(\bm{\omega}))=\frac{V^2}{N}\mathbb{V}\text{ar}(\text{Y})$. Expanding the variance gives 
\begin{equation}\label{eq: var Y expanded}
    \begin{aligned}
        \mathbb{V}\text{ar}(\text{Y})&=\mathbb{E}(|\text{Y}|^2)-|\mathbb{E}(\text{Y})|^2\\
        &=\mathbb{E}(g_{ij}^2)-|\mathbb{E}(\text{Y})|^2\\
        &=\mathbb{V}\text{ar}(g_{ij})+\mathbb{E}(g_{ij})^2-|\mathbb{E}(\text{Y})|^2.
    \end{aligned}
\end{equation}
Observing \Cref{eq: var Yc expanded} and \Cref{eq: var Y expanded}, a key observation is that $\mathbb{V}\text{ar}(\text{g}_{ij}), \mathbb{E}(g_{ij})$ is frequency-independent, hence providing an lower bound for $\mathbb{V}\text{ar}(\text{Y}), \mathbb{V}\text{ar}(\text{Y}_c)$ in the limit of large frequencies. Together with the empirical evidence from \Cref{fig: amplitude_embedding_coefficients_variance_scaling}, which shows that $\mathrm{I}(\bm{\omega})$, hence $K_{ij}(\bm{\omega})$, decays exponentially with respect to the frequency norms, it implies that the the Monte Carlo integration variance becomes proportionally more significant relative to the value of the integral itself. However, at the extreme of high frequencies regime, exactly because $\mathbb{E}(\text{Y})=K_{ij}(\bm{\omega})\sim 0$, \Cref{eq: var Y expanded} can be approximated as
\begin{equation}
    \mathbb{V}\text{ar}(\text{Y})=\mathbb{V}\text{ar}(g_{ij})+\mathbb{E}(g_{ij})^2>\mathbb{V}\text{ar}(\text{Y}_c),
\end{equation}
which now clearly demonstrates the utility of control variation in reducing statistical variance of Monte Carlo integration.

\newpage
\section{Supplementary plots for Zeroth coefficient \texorpdfstring{$c_{0}$}{c0} simulations}\label{appendix: supplementary plots for zeroth coefficient}
\begin{figure*}[thb!]
    \centering
    \includegraphics[width=\columnwidth]{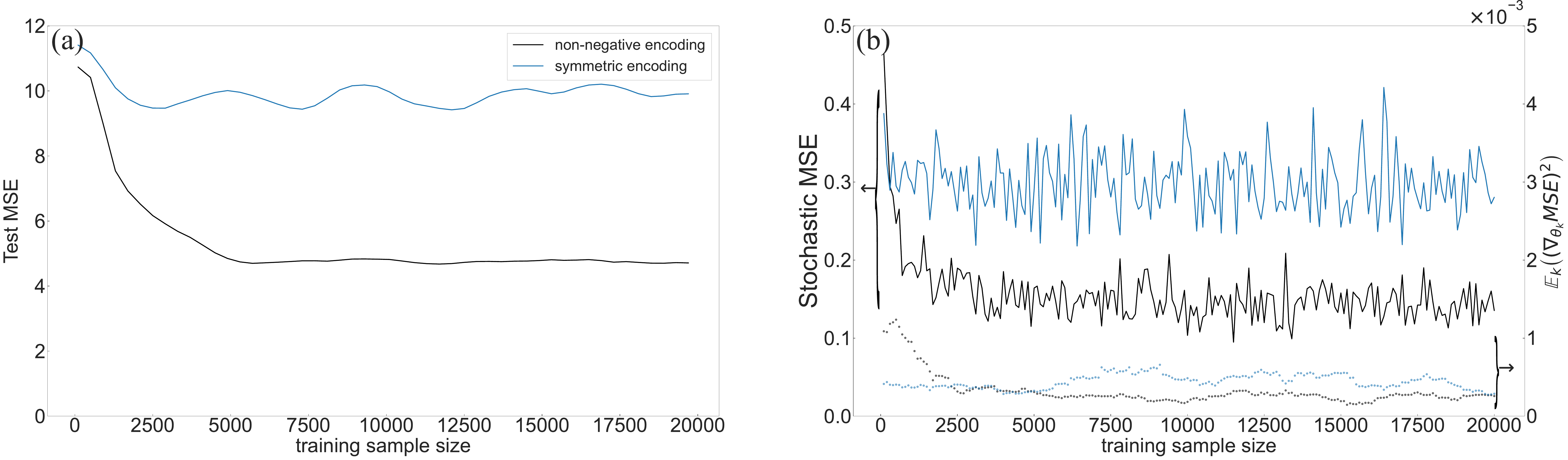}
    \caption{This figure continues from \Cref{fig:symmetric_and_positive_real}, which plots the (a) mean squared error (MSE) for non-negative and symmetric encoding with respect to the test set, as well as their corresponding (b) stochastic MSE and gradients squared.}
    \label{fig:symmetric_and_positive_MSE}
\end{figure*}
\begin{figure*}[thb!]
\centering
    \includegraphics[width=1\columnwidth]{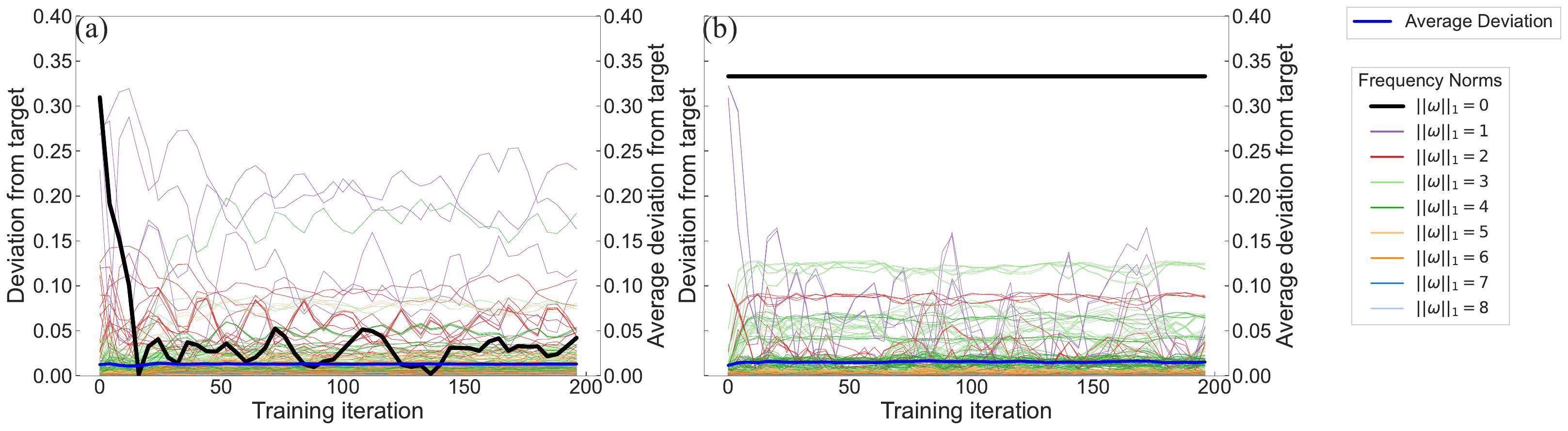}
    \caption{This figure continues from \Cref{fig:symmetric_and_positive_real}, which plots the (a) non-negative encoding and (b) symmetric encoding average absolute and proportional deviation of the sampled Fourier coefficients $c_{\bm{\omega}'}$ to their target values $c_{\bm{\omega}}$, $|c_{\bm{\omega}}-c_{\bm{\omega}'}|$ and $\mathbb{E}_{\bm{\omega}}\left(|c_{\bm{\omega}}-c_{\bm{\omega}'}|\right)$, versus the number of training iterations, respectively on left and right axis. The legends use the same colour for the coefficients with the same $L_1$ norm $\|\bm{\omega}\|_1$. }
    \label{fig:symmetric_and_positive_deviations}
\end{figure*}
\newpage
\section{Supplementary plots for depolarizing channel simulations}\label{appendix: supplementary plots for depol channel}

\begin{figure*}[thb!]
    \centering
    \begin{subfigure}{.42\textwidth}
        \centering
        \includegraphics[width=\columnwidth]{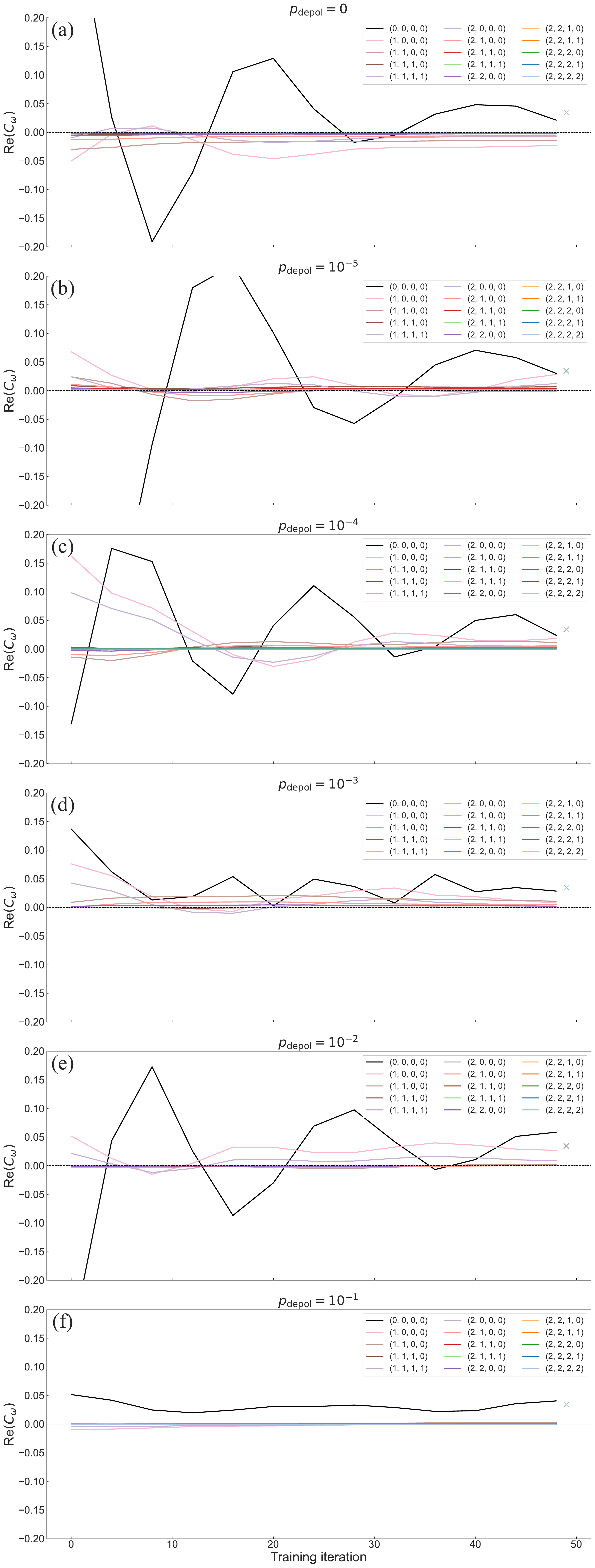}
        \phantomcaption
        \label{fig:2q_depol_nonzero_coeff_real_parts}
    \end{subfigure}
    \begin{subfigure}{.42\textwidth}
        \centering
        \includegraphics[width=\columnwidth]{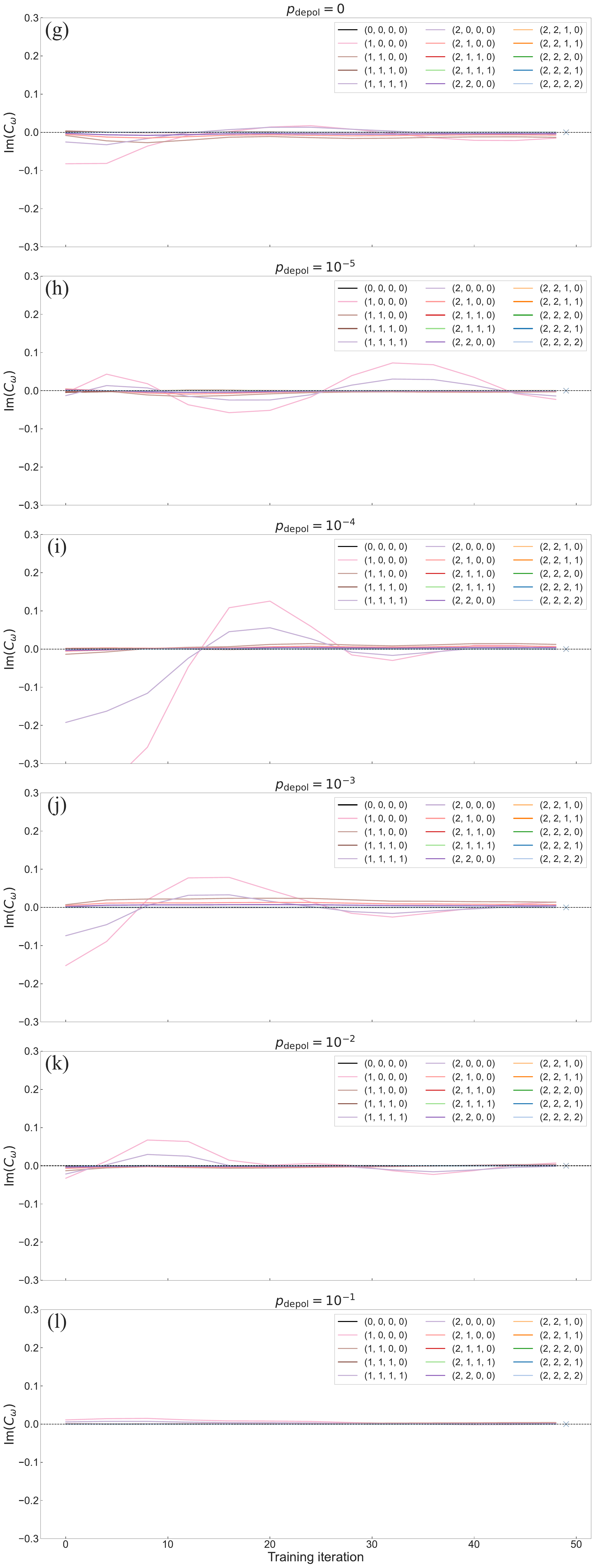}
        \phantomcaption
        \label{fig:2q_depol_nonzero_coeff_imag_parts}
    \end{subfigure}
    \caption{(a)-(f) Real part of the sampled Fourier coefficients with non-zero target-coefficients, which is marked by the cross at the end of each plot. (g)-(l) Imaginary part of the sampled Fourier coefficients with non-zero target-coefficients. Note that although the real part of the target coefficients is non-zero, the imaginary part is always set to zero.}
    \label{fig:2q_depol_nonzero_coeffs_real_imag}
\end{figure*}

\begin{figure*}[thb!]
    \centering
    \begin{subfigure}{.43\textwidth}
        \centering
        \includegraphics[width=\columnwidth]{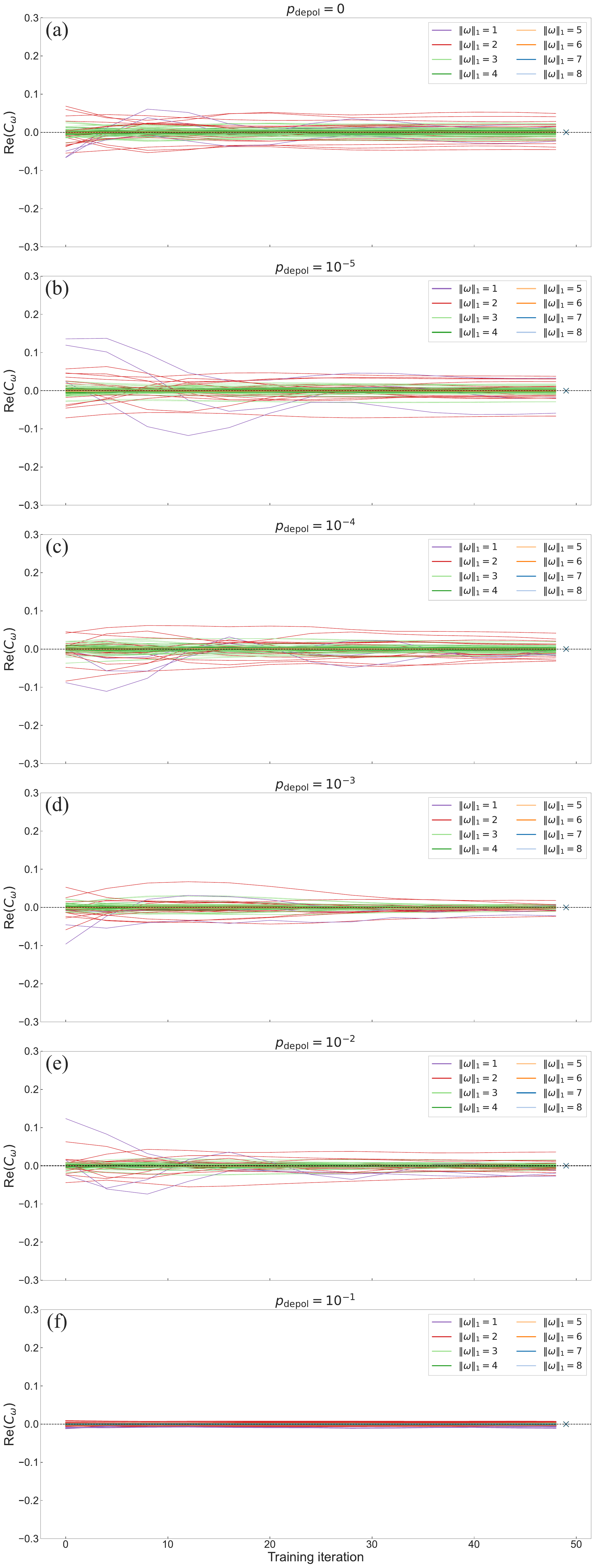}
        \caption{Real part of the sampled Fourier coefficients with zero target-coefficients}
        \label{fig:2q_depol_zero_coeff_real_parts}
    \end{subfigure}
    \begin{subfigure}{.43\textwidth}
        \centering
        \includegraphics[width=\columnwidth]{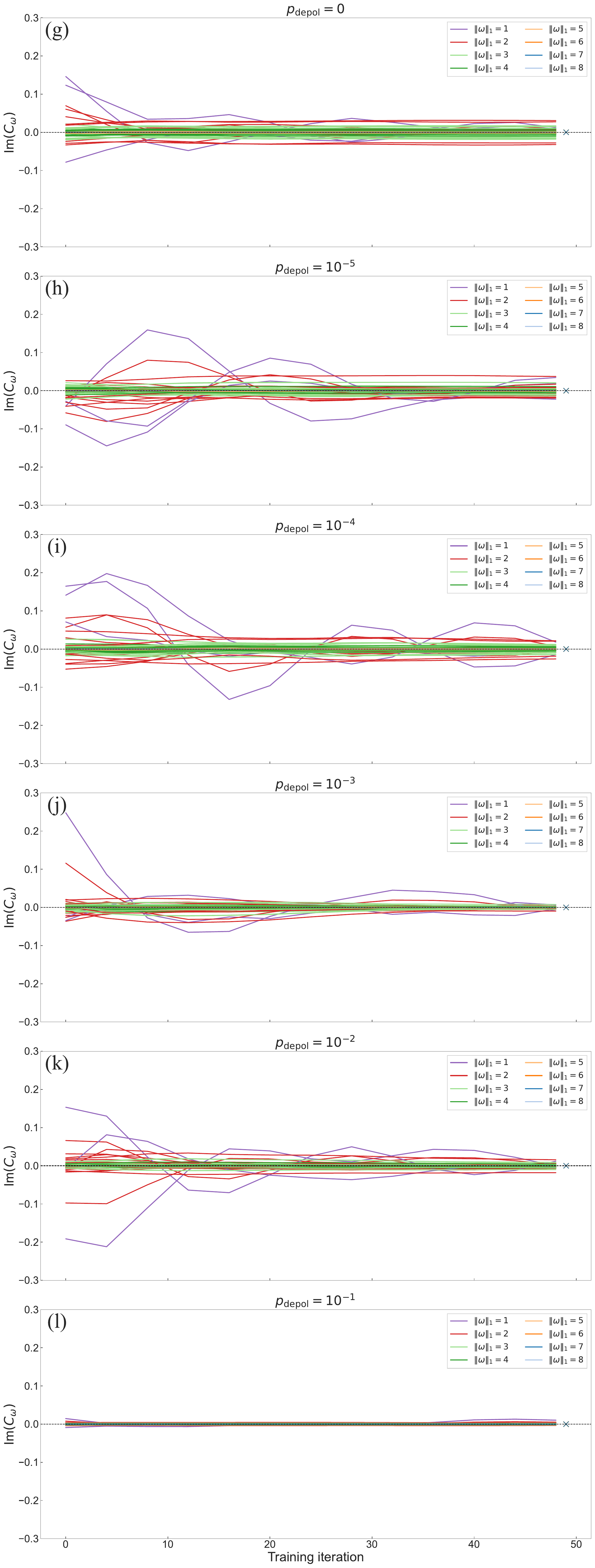}
        \caption{Imaginary part of the sampled Fourier coefficients with zero target-coefficients}
        \label{fig:2q_depol_zero_coeff_imag_parts}
    \end{subfigure}
    \caption{(a)-(f) Real part of the sampled Fourier coefficients with zero target-coefficients, which is marked by the cross at the end of each plot. (g)-(l) Imaginary part of the sampled Fourier coefficients with zero target-coefficients.}
    \label{fig:2q_depol_zero_coeffs_real_imag}
\end{figure*}

\newpage

\section{Supplementary plots for non-integer target frequencies simulations}\label{appendix: supplementary plots for non-int freqs}
\begin{figure*}[thb!]
    \centering
    \begin{subfigure}{.45\textwidth}
        \centering
        \includegraphics[width=\columnwidth]{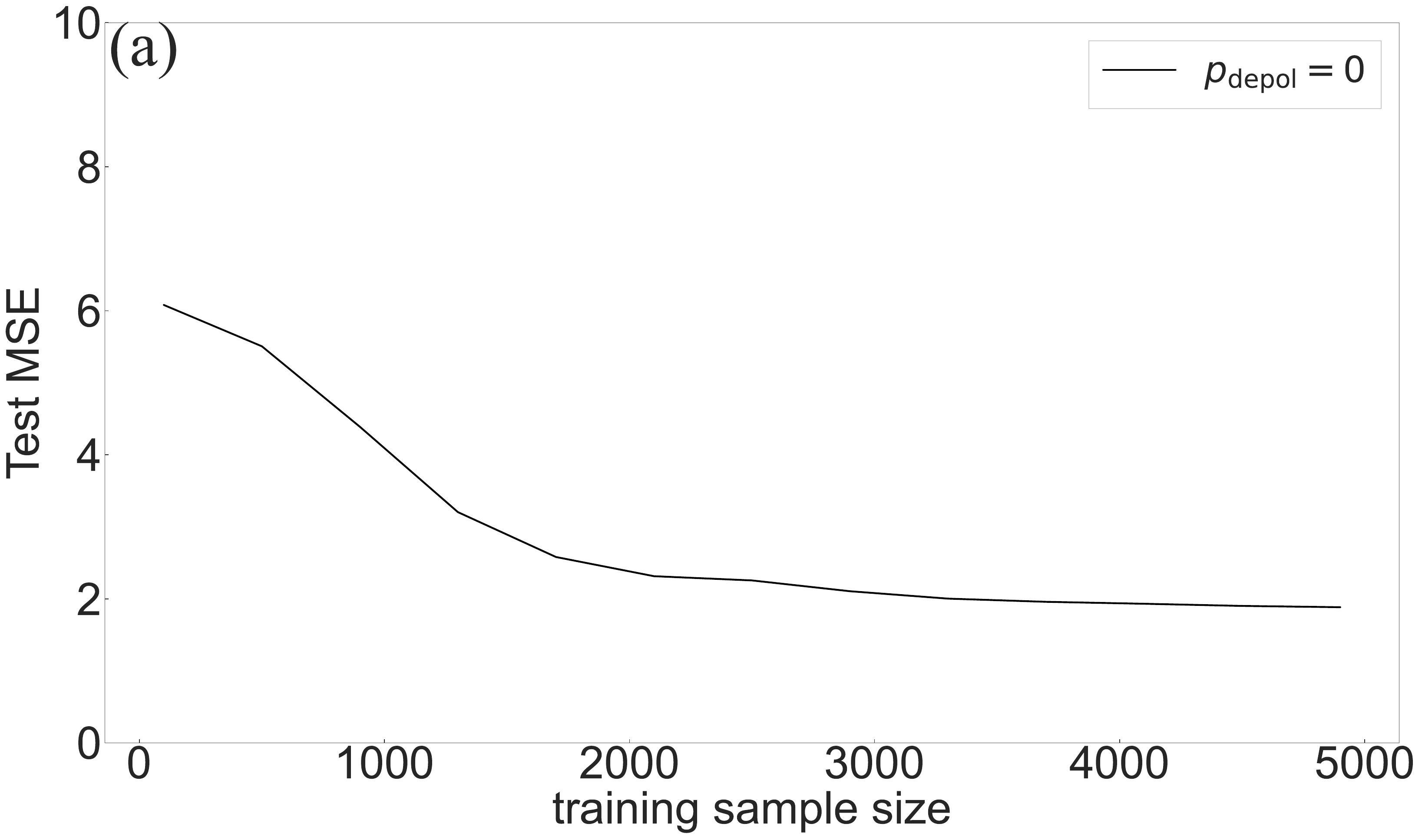}
        \label{fig:2q_non_int_freqs_depol_MSE}
    \end{subfigure}
    \begin{subfigure}{.45\textwidth}
        \centering
        \includegraphics[width=\columnwidth]{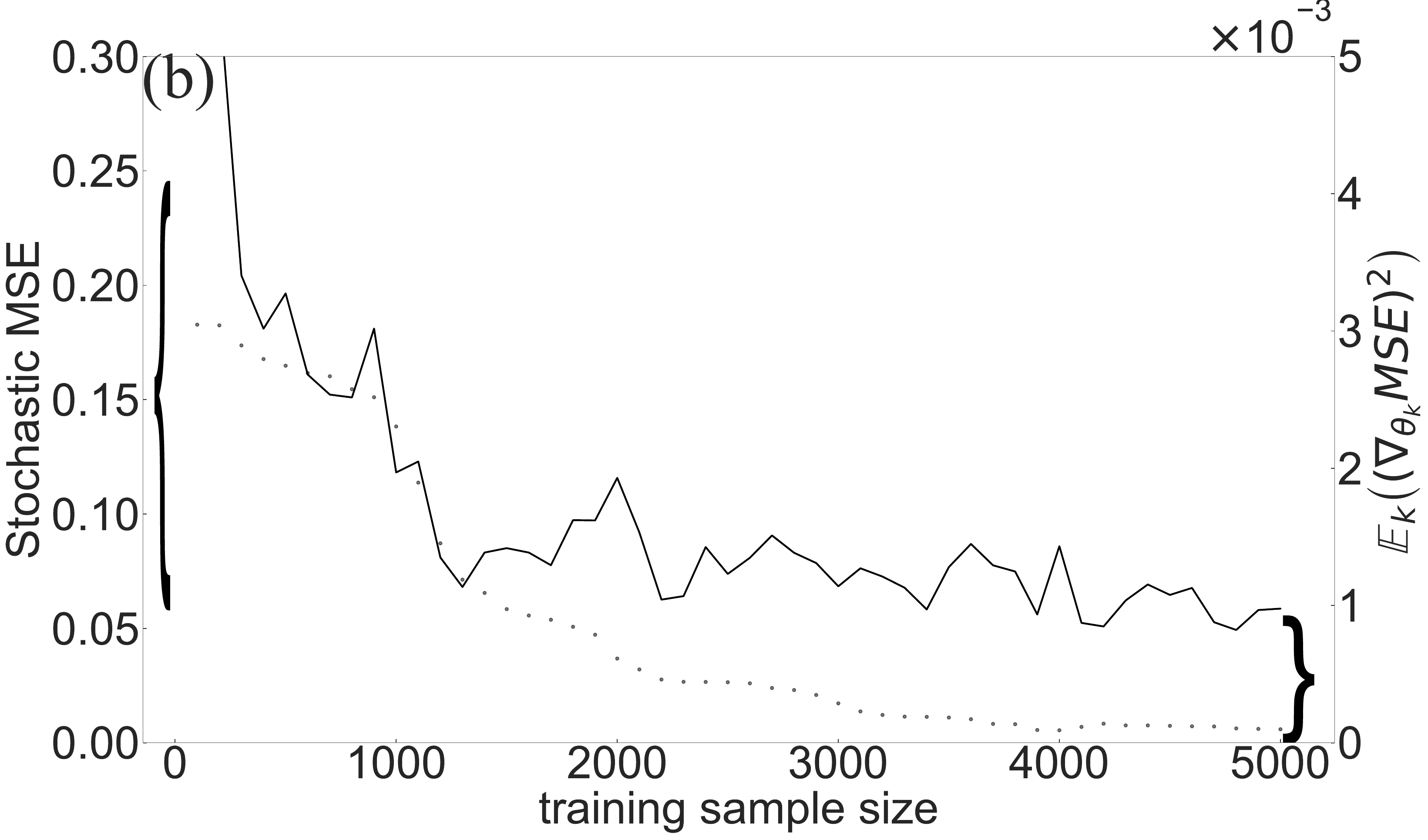}
        \label{fig:2q_non_int_freqs_depol_MSE_gradients}
    \end{subfigure}
    \begin{subfigure}{.45\textwidth}
        \centering
        \includegraphics[width=\columnwidth]{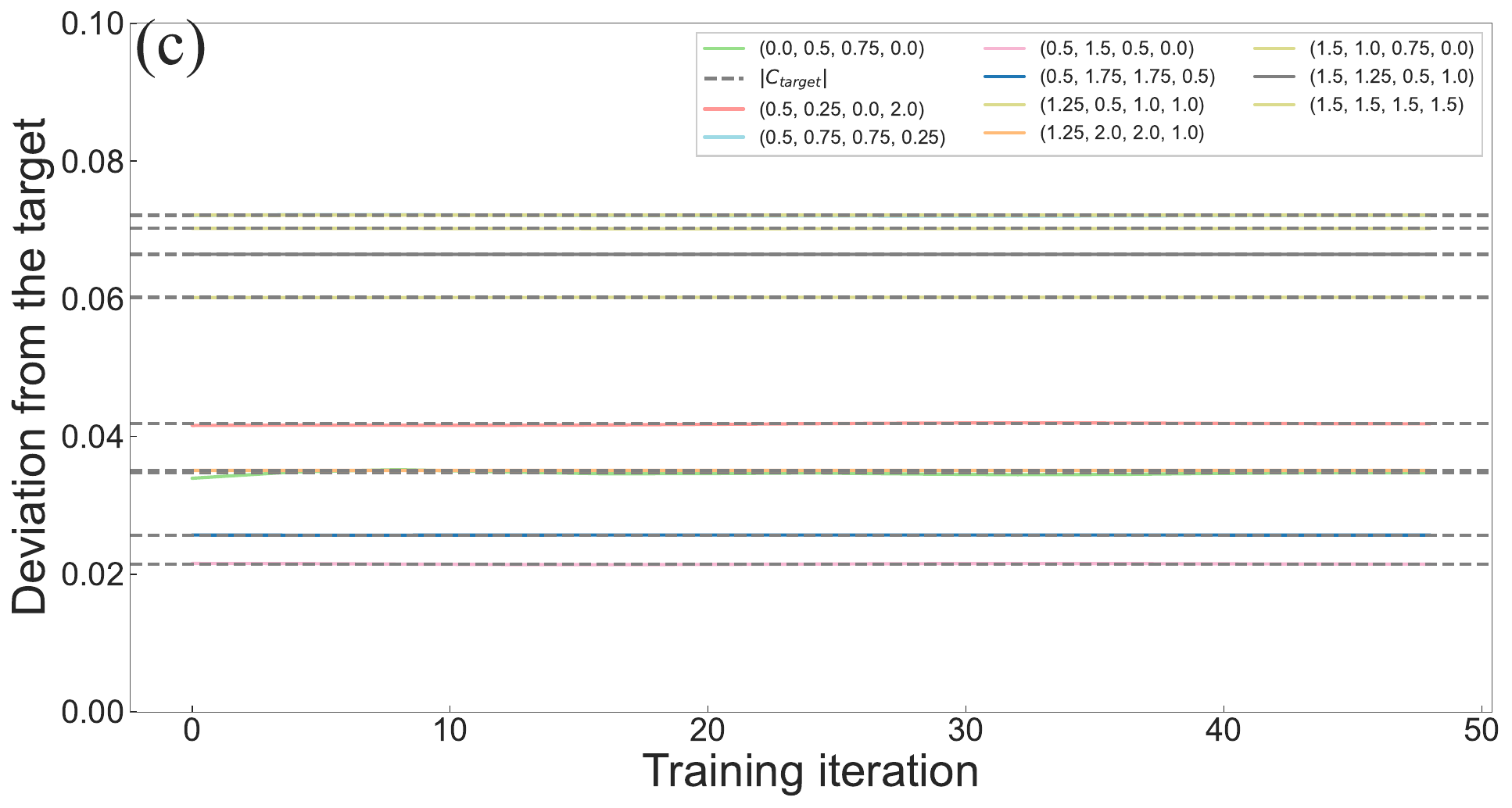}
        \label{fig:2q_non_int_freqs_depol_nonzero_coeff_deviations}
    \end{subfigure}
    \begin{subfigure}{.45\textwidth}
        \centering
        \includegraphics[width=\columnwidth]{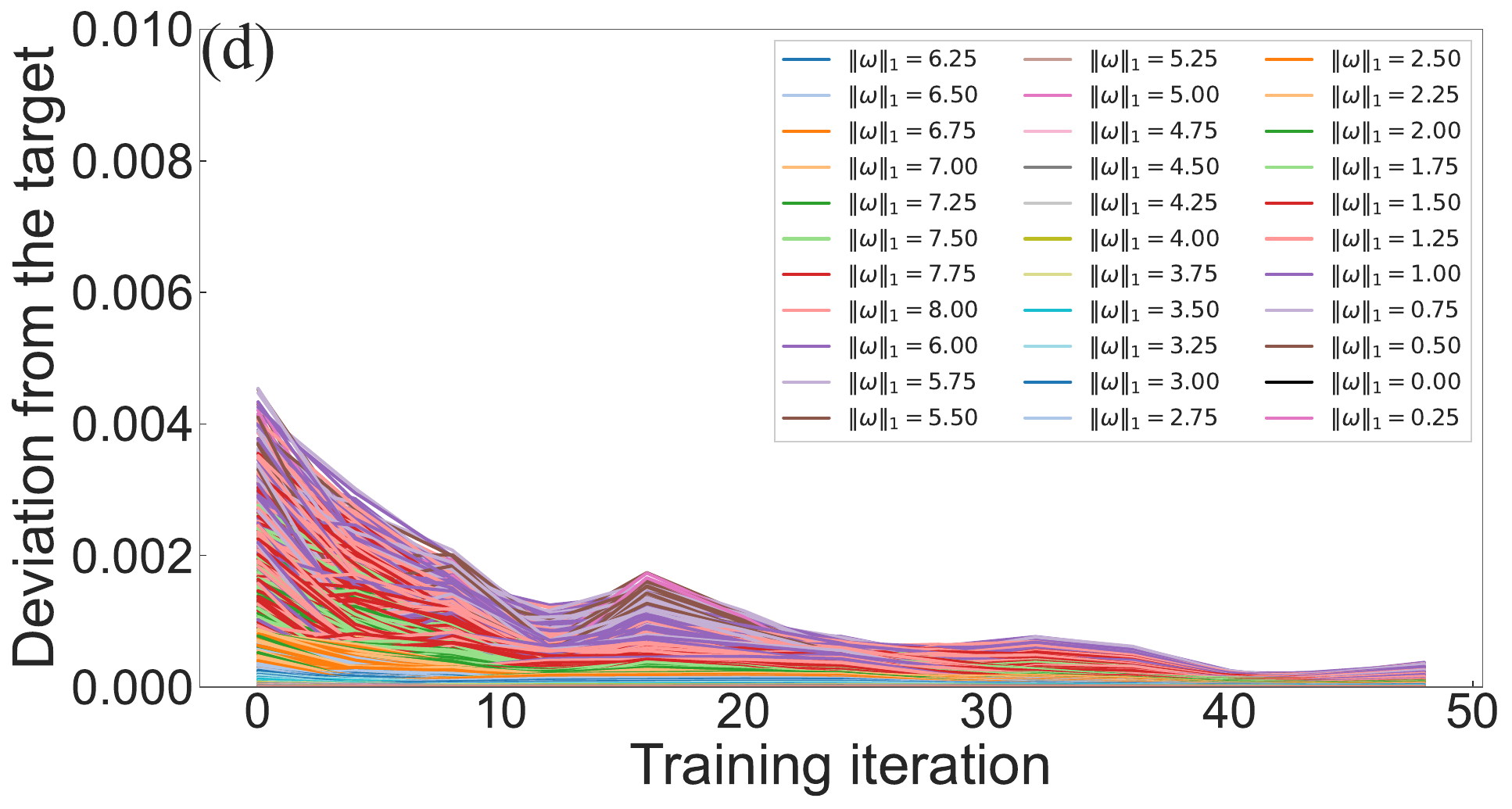}
        \label{fig:2q_non_int_freqs_depol_zero_coeff_deviations}
    \end{subfigure}
    \caption{This figure continues from \Cref{fig: non int freqs plots} (a) The MSE, (b) stochastic MSE per batch, and gradient norm squared trained with respect to a function with non-zero Fourier coefficients on non-integer frequencies. (c) and (d) are the deviation of the Fourier coefficients from non-zero and zero target values, respectively.}
    \label{fig: non int freqs plots continue}
\end{figure*}

%
%






\twocolumngrid
\renewcommand{\href}[2]{#2}
\bibliography{main}
\end{document}